\documentclass[3p]{elsarticle}

\usepackage[fontsize=8pt]{scrextend}

\usepackage[english]{babel}

\usepackage[linesnumbered,ruled,vlined]{algorithm2e}

\usepackage{amsmath,amsfonts}
\usepackage{natbib}
\usepackage{amsthm}
\usepackage{amssymb}
\usepackage{graphicx}
\usepackage{lipsum}
\usepackage{relsize}
\usepackage{mathtools}
\usepackage{url}
\usepackage{braket}
\usepackage{tabularx}
\usepackage{tikz}
\usepackage{adjustbox}
\usetikzlibrary{3d,calc}
\usepackage{tikz-3dplot}
\usetikzlibrary{arrows,calc,backgrounds}

\usepackage{epstopdf}

\newcommand{\degree}{\text{deg}}

\newcommand{\RR}{\mathbb{R}}
\newcommand{\cut}{\text{cut}}
\newcommand{\mc}{\text{Max-Cut}}
\newcommand{\MC}{\texttt{MAX-CUT}}

\newtheorem{theorem}{Theorem}
\newtheorem{prop}[theorem]{Proposition}
\newtheorem{lemma}[theorem]{Lemma}
\newtheorem{obs}[theorem]{Observation}
\newtheorem{corollary}[theorem]{Corollary}
\newtheorem{conjecture}[theorem]{Conjecture}

\title{Theoretical Approximation Ratios for Warm-Started QAOA on 3-Regular Max-Cut Instances at Depth $p=1$ }

\author[1]{Reuben Tate\corref{cor1}}
\ead{rtate@lanl.gov}
\author[1]{Stephan Eidenbenz}
\ead{eidenben@lanl.gov}
\cortext[cor1]{Corresponding author}
\affiliation[1]{organization={CCS-3: Information Sciences, Los Alamos National Laboratory},
city={Los Alamos, NM},
country={United States}}

\usepackage{amsopn}

\newcommand{\swap}{\ifmmode \text{{\fontspec{Symbola}\symbol{"1F503}}} \else {\fontspec{Symbola}\symbol{"1F503}} \fi}

\newcommand\restr[2]{{
  \left.\kern-\nulldelimiterspace 
  #1 
  \vphantom{\big|} 
  \right|_{#2} 
  }}

\ifpdf
  \DeclareGraphicsExtensions{.eps,.pdf,.png,.jpg}
\else
  \DeclareGraphicsExtensions{.eps}
\fi

\newenvironment{nalign}{
    \begin{equation}
    \begin{aligned}
}{
    \end{aligned}
    \end{equation}
    \ignorespacesafterend
}

\begin{document}
\begin{abstract}
    We generalize Farhi et al.'s 0.6924-approximation result technique of the Max-Cut Quantum Approximate Optimization Algorithm (QAOA) on 3-regular graphs to obtain provable lower bounds on the approximation ratio for warm-started QAOA. Given an initialization angle $\theta$, we consider warm-starts where the initial state is a product state where each qubit position is angle $\theta$ away from either the north or south pole of the Bloch sphere; of the two possible qubit positions the position of each qubit is decided by some classically obtained cut encoded as a bitstring $b$. 
    
    We illustrate through plots how the properties of $b$ and the initialization angle $\theta$ influence the bound on the  approximation ratios of warm-started QAOA. We consider various classical algorithms (and the cuts they produce which we use to generate the warm-start). Our results strongly suggest that  there does not exist any choice of initialization angle that yields a (worst-case) approximation ratio that simultaneously beats standard QAOA and the classical algorithm used to create the warm-start. 
    Additionally, we show that at $\theta=60^\circ$, warm-started QAOA is able to (effectively) recover the cut used to generate the warm-start, thus suggesting that in practice, this value could be a promising starting angle to explore alternate solutions in a heuristic fashion. 

\end{abstract}
\maketitle

\section{Introduction}
Over the past decade, quantum devices have significantly improved in their capabilities; with such improvements, there has also been increased interest, both practical and theoretical, in using such devices to solve challenging problems in combinatorial optimization. In 2014, for gate-based quantum devices, Farhi et al. \cite{farhi2014quantum} introduced the Quantum Approximate Optimization Algorithm (QAOA) which is a general framework for obtaining approximately optimal solutions to combinatorial optimization problems. The QAOA algorithm involves a parametrized quantum circuit which, with the certain choice of parameters, is effectively a Trotterization of the Quantum Adiabatic Algorithm (QAA); however, optimizing these parameters can potentially lead to even better solutions (in expectation). With this connection between QAOA and QAA, a higher circuit depth $p$ of QAOA with certain parameters corresponds to a finer Trotterization of QAA.

In the classical optimization literature, the term \emph{approximation algorithm} is typically used to denote an algorithm that is able to guarantee a solution whose objective value is at least a certain fraction of the objective of the optimal solution. Farhi had primarily viewed his algorithm as an approximation algorithm, proving that for the \MC{} problem on 3-regular graphs, QAOA returns a cut with at least $0.6924$ as many edges compared to the optimal cut \cite{farhi2014quantum}. Since then, several have proposed modifications to QAOA (both to the circuit ansatz itself \cite{FGGN17,HWORVB17,WRDR20,ZTBMBE20,BKKT19,bartschi2020grover,jiang2017near} and the parameter optimization procedures \cite{GLLAS21,SS21,SMKS22,SLLOH22,zhou2020quantum}). While these QAOA variants have often been shown to perform empirically well (compared to the standard QAOA algorithm) for certain test beds of problems and problem instances, it is often the case that little is known in regard to theoretical guarantees, i.e., these QAOA variants are effectively treated as a heuristic. Even for the standard QAOA, outside of certain negative-results \cite{FGG20,BKKT19,bravyi2022hybrid} and low-depth regimes \cite{wurtz2021maxcut,FGGN17}, the field's understanding of QAOA is still extremely limited from a theoretical perspective.

In prior works \cite{Tate2023warmstartedqaoa,tate2023bridging,egger2021warm,cain2022qaoa,FGGN17}, some have explored the notion of \emph{warm-starting} QAOA, i.e., constructing a different initial state for the QAOA circuit in a way determined by a previously found solution. Like other QAOA variants, it was found that warm-started QAOA tends to empirically yield better solutions compared to standard QAOA; however little was known regarding the theoretical performance of these methods. This work aims to further the field's theoretical understanding of the effect that warm-starting has on the QAOA algorithm. Such a theoretical analysis is difficult for the general types of warm-starts considered in \cite{tate2023bridging,Tate2023warmstartedqaoa}; however, by making certain assumptions regarding the structure of the warm-start, the analysis becomes amenable and new theoretical results regarding warm-stated QAOA emerge as a result.

More specifically, we consider the \MC{} problem on 3-regular graphs; for \MC, the goal is to find a partition of the vertices into two groups that maximizes the number of edges between the partitions. Given a \MC{} solution encoded by a bitstring $b$ of a particular graph $G=(V,E)$ and some initialization angle $\theta \in [0,\pi]$, the warm-start state $\ket{b_\theta}$ is product state that is constructed by placing each qubit on the Bloch sphere at an angle $\theta$ away from either the north or south pole of the Bloch sphere (with the placement being determined by the bitstring $b$); these warm-starts are defined in a more formal manner in Section \ref{sec:initialStates}.


We obtain our results by significantly generalizing Farhi's proof technique for single-round QAOA. Our main result, stated formally below and formally in Theorem \ref{thm:lowerBoundVaryingKappa}, is a proof of a guaranteed approximation ratio that warm-start QAOA achieves in a single round for every angle $\theta \in [0, \pi/2]$. 

\begin{theorem}[Informal Statement of Theorem \ref{thm:lowerBoundVaryingKappa}] 
    As a function of the initialization angle $\theta$ of each initial qubit from either north or south pole of the Bloch sphere and the total fraction of edges $\kappa$ that the determining bitstring $b$ cuts (compared to the total number of edges in the graph), one can obtain a lower bound on the approximation ratio of warm-started \mc{} QAOA on 3-regular graphs by solving a continuous optimization problem. 
\end{theorem}

For three different classical \mc{} algorithms, we numerically calculate such bounds on the approximation ratio (which are plotted in Figure \ref{fig:plotAR_varyingConstraints}). On certain subsets of 3-regular graphs, these algorithms return cuts with certain guaranteed properties.  In particular, we consider, case (i) with $b$ a one-exchange locally optimum solution (which is at least a fraction of $2/3$ within the global optimum on 3-regular graphs), case(ii) where $b$ is a cut that is guaranteed to contain a fraction of $4/5$ of all edges, and case (iii) a guaranteed fraction of $17/21$ of all edges. 

For all cases, $\pi/2$ we recover -- as expected -- Farhi's original approximation ratio of $0.6924$  at $\theta = \pi/2$. 
For case (i) with locally optimal starting states, we show that the approximation ratio guarantees actually decrease slightly with increasing $\theta$ before eventually settling at $0.6924$ after some oscillations. For the two cases (ii) and (iii) with high-approximation ratios at $\theta=0$, we observe a decrease up until about $\theta = \pi/4$, followed by a marked increase back to the original starting ratio at $\theta = \pi/3$ before settling again at Farhi's ratio at $\theta = \pi/2$. We also formally explain the the peak at $\theta = \pi/2$.
The non-monotonic behavior is notable and surprising. Our results suggest that promising QAOA experiments could be done by setting $\theta=\pi/3$ as this allows the algorithm to broadly explore the search space (unlike at $\theta = 0$).

Provable approximation ratio results for QAOA are still very rare. Our work gives a large number of such results along the values of parameter $\theta$, thus giving new insights into the theory of QAOAs. While our results do not show a worst-case advantage of warm-start QAOA over either the traditional QAOA approach (with approximation ratio of $0.6924$) or the quality of a good starting state (such as in for cases (ii) and (iii)), the practical advantage of warm-start QAOA in NISQ and numerical experiments should be explored further.

In Section \ref{sec:notationBackground}, we give a background on both the QAOA algorithm, various results in combinatorial optimization, providing the necessary notation that will be needed for the rest of the work. In Section \ref{sec:initialStates}, we discuss the specific warm-starts used in this work in addition to the corresponding mixing Hamiltonians that ``align" with such warm-starts. In Section \ref{sec:approximationRatiosDepth1WarmStartQAOA}, we generalize Farhi et al.'s approach to obtain approximation ratios for depth-1 warm-started QAOA on 3-regular graphs. Finally, we provide a discussion and conclude in Section \ref{sec:discussionConclusion}.

\section{Notation and Background}
\label{sec:notationBackground}
Unless otherwise stated, we assume all graphs are simple, undirected, unweighted, loopless graphs. Formally, a graph $G$ is a tuple $(V,E)$ where $V$ is a finite set and the set of edges $E$ is a set of two-element subsets of $V$. We also use the notation $V(G)$ and $E(G)$ to denote the vertices and edges of a graph $G$. We will typically use $n_G := |V(G)|$ and $m_G := |E(G)|$ to represent the number of vertices and edges in a graph respectively; we drop the subscripts when the graph $G$ is clear from context. Unless otherwise stated, we assume that $V(G) = [n_G] := \{1,2,\dots,n_G\}$. A $k$-regular graph is a graph where every vertex has degree $k$; we will sometimes refer to a 3-regular graph as a \emph{cubic} graph. A graph is considered \emph{sub-cubic} if the degree of every vertex is at most 3. We use $K_n$ to denote the complete graph on $n$ vertices. A graph is considered \emph{triangle-free} if it contains no cycles of length 3. If $G=(V,E)$ is a graph and $E' \subseteq E$, then we say that $H$ is an \emph{edge-induced subgraph} induced by $E'$ if the edge-set of $H$ is $E'$ and $H$ has no isolated (degree-0) vertices, i.e., $V(H) = \bigcup E'.$

For later convenience, we define a \emph{marked graph} as a tuple $(G,e)$ where $G$ is a graph and $e \in E(G)$ is an edge in the graph; we refer to $e$ as the \emph{marked edge} of a marked graph $(G,e)$. We also define a colored graph as a tuple $(G, f)$ where $G$ is a graph and $f: V(G) \to C$ is a \emph{vertex-coloring}, i.e., a function from the vertices of $G$ to a finite color set $C$. In the case that $|C|=2$ (with $C = \{c_1, c_2\}$) and $V \subseteq [k]$ with $V = \{v_1, v_2, \dots, v_n\}$ with $v_1 < v_2 < \cdots < v_n$, we will typically represent $f_C$ as a bitstring $b$ where the $j$th entry $b_j$ is $0$ if $f(v_j) = c_1$ and $b_j$ is 1 if $f(v_j) = c_2$. We say that a graph $G$ is $k$-colorable if there exists a \emph{proper} vertex-coloring $f$ with $k$ colors where adjacent vertices always have different colors, i.e., $f(u) \neq f(v)$ for all $\{u,v\} \in E(G)$. Lastly, a \emph{colored marked graph} is a $3$-tuple $(G,f,e)$ where $G$ is a graph, $f$ is a vertex coloring on $G$, and $e \in E(G)$.

Given a graph $G=(V,E)$, any subset of vertices $S \subseteq V$ induces a \emph{cut} $(S, V \setminus S)$ where $S$ and $V \setminus S$ are a partition of the vertices $V$ of $G$. We use $\cut_G(S)$ to denote\footnote{For a proposition $P$, the expression $\mathbf{1}[P]$ evaluates to 1 if $P$ is true and evaluates to 0 otherwise.} the number of edges whose endpoints are in different partitions, i.e., $\cut_G(S) = \sum_{e \in E(G)} \mathbf{1}[|e \cap S| = 1]$; we refer to $\cut_G(S)$ as the \emph{cut value} of $S$. We drop the subscript $G$ when the graph $G$ is clear from context. The \MC{} problem is to find a subset of vertices $S$ so that the number of cut edges $\cut(S)$ is maximized, i.e., $\mc(G) := \max_{S \subseteq V} \cut(S)$. For any cut $(S, V \setminus S)$ with $V \subseteq [k]$ with $V = \{v_1, v_2, \dots, v_n\}$ with $v_1 < v_2 < \cdots v_n$, there is an associated bitstring $b$ where $b_j = 0$ if $v_j \in S$ and $b_j = 1$ otherwise. As an abuse of notation, we will sometimes refer to the bitstring $b$ itself as a cut and use $\cut(b)$ to denote the cut value of the cut associated with $b$. We say that a cut $(S, V \setminus S)$ is an \emph{optimal cut} if $\cut(S) = \mc(G)$.

Consider an arbitrary combinational optimization maximization problem on $n$-length bitstrings determined by a cost function $c:\{0,1\}^n \to \RR$. For any $k=1,\dots,n$ and bitstring $b \in \{0,1\}^n$, let $N_k(b)$ be the set of bitstrings whose Hamming distance from $b$ is at most $k$. Consider algorithm $k$-BLS (Algorithm \ref{alg:localSearch}) parametrized by $k$: this algorithm iteratively improves a solution by flipping up to $k$ bits as a time. For any given combinatorial optimization problem, we say that a bitstring $b$ is (locally) $k$-BLS optimal if for all $b' \in N_k(b)$, we have that $c(b') \leq c(b)$; in other words, $b$ is a potential output for the $k$-BLS algorithm.

\begin{algorithm}
\label{alg:localSearch}
    \caption{\footnotesize $k$-Bitflip Local Search ($k$-BLS)}
    \KwData{Integer $n$, function $c: \{0,1\}^n \rightarrow \mathbb{R}$}
    \KwResult{Locally optimal bitstring $b$}
    
    \BlankLine
    \textbf{Initialization:} Choose a bitstring $b \in \{0,1\}^n$ uniformly at random\;
    
    \While{there exists $b' \in N_k(b)$ with $c(b') > c(b)$}{
        Set $b := b'$\;
    }
    
    \Return{$b$}\;
    
\end{algorithm}

\subsection{Approximation Ratio}
\label{sec:approximationRatio}
For \MC{} on unit-weight graphs, we define the approximation ratio in the standard way. Let $\mathcal{A}$ be a (potentially randomized) \MC{} algorithm, let $G$ be a unit-weight graph, let $\mathcal{A}(G)$ be the expected cut value obtained by running $\mathcal{A}$ on $G$, then we define the \emph{instance-specific} approximation ratio $\alpha_{\mathcal{A}}(G)$ as
$$\alpha_{\mathcal{A}}(G) = \frac{\mathcal{A}(G)}{\mc(G)}.$$

For variational algorithms like QAOA, we have the freedom to either consider $\mathcal{A}$ to be the QAOA circuit at \emph{specific} choices of variational parameters, or we can consider the optimization of such parameters as part of the algorithm $\mathcal{A}$ itself. For simplicity, for QAOA and its variants, we consider $\mathcal{A}$ to be an algorithm that uses the instance-dependent optimal choice of variational parameters. In Section \ref{sec:optimizationDetails}, for both standard and warm-started QAOA, we more closely inspect the effect of including or excluding the optimization of various parameters as part of the algorithm description. Later in this work, more notation will be introduced which will allow us to give a more explicit definition of $\mathcal{A}(G)$ (and hence $\alpha_\mathcal{A}(G)$).

For a family of graphs $\mathcal{G}$, we say that algorithm $\mathcal{A}$ has a (theoretical) approximation ratio (over $\mathcal{G}$) of $\alpha_\mathcal{A}(\mathcal{G})$ if
$$\alpha_\mathcal{A}(\mathcal{G}) = \min_{G \in \mathcal{G}} \alpha_\mathcal{A}(G);$$
when the class of graphs $\mathcal{G}$ is clear from context, we will just write $\alpha_\mathcal{A}$. Throughout this work, we typically consider $\mathcal{G}$ to be one of the following: (1) all unit-weight graphs, (2) all unit-weight 3-regular graphs, or (3) all unit-weight triangle-free 3-regular graphs.

\subsection{Classical Approximation Results for Max-Cut}
\label{sec:classicalMaxCutResults}
Next, we review some of the known classical results for \MC{} that will be relevant for this work, beginning with Lemma \ref{thm:edge_ratio_maxcut_relation}.
\begin{lemma}
\label{thm:edge_ratio_maxcut_relation}
    Let $\mathcal{A}$ be a \MC{} algorithm such that, for any graph with $m$ edges, the solution returned cuts at least $\kappa \cdot m$ edges where $\kappa \in [0,1]$. Then $\mathcal{A}$ is an $\kappa$-approximation algorithm for \MC.
\end{lemma}
\begin{proof}
    This follows immediately from the fact that the \MC{} of a unit-weight graph is bounded above by the total number of edges in the graph.
\end{proof}
Note that the converse of the above lemma is not necessarily true: an $\alpha$-approximation algorithm for \MC{} does not necessarily cut at least an $\alpha$ fraction of the edges of the graph.

For \MC{} on graphs with non-negative weights, it is well-known that 1-BLS gives a $1/2$ approximation ratio\footnote{For \MC{} with unit-weight graphs on $m$ edges, 1-BLS requires at most $m$ iterations. For weighted graphs, this might not be possible; however, a $\frac{1}{2(1+\epsilon/4)}$ approximation is possible in $O(\frac{1}{\epsilon}n \log n)$ iterations \cite{Chekuri}.}. In the case of the 1-BLS algorithm on $3$-regular graphs, this approximation ratio can be improved to $2/3$; while we believe that this result is already known, we were unable to locate a source for the result and thus, for the sake of completion, we provide a proof as seen in Proposition \ref{thm:localMaxCut3Reg} below.

\begin{prop}
\label{thm:localMaxCut3Reg}
For \MC{} on 3-regular graphs, 1-BLS gives a cut containing at least $2/3$ of the edges in the graph (and hence yields a $2/3$-approximation by Lemma \ref{thm:edge_ratio_maxcut_relation}).
\end{prop}
\begin{proof}    Let $G = (V,E)$ be a 3-regular graph and let $(S, V \setminus S)$ be the cut returned by 1-BLS. For each vertex $v \in V$, let $\delta(v)$ be the number of neighbors of $v$ that are on the other side of the cut. Since $\degree(v) = 3$, then it must be that $\delta(v) \geq 2$, otherwise, the bit corresponding to $v$ could be flipped to obtain a new cut with more cut edges. We can also write $\delta(v) \geq \frac{2}{3}\degree(v)$.
    
    Observe that the number of cut edges is given by $\sum_{v \in V} \delta(v)/2$. Thus, the number of edges that are cut is at least
    $$\sum_{v \in V} \delta (v)/2 \geq \sum_{v \in V} \left(\frac{2}{3}\degree(v)\right)/2 = \frac{1}{3}\sum_{v \in V} \degree(v) = \frac{2}{3}|E|.$$
\end{proof}


Outside of an optimization perspective, many structural graph theorists are interested in studying the quantity $\mc(G)$ itself; more specifically, they are primarily concerned with bounding the quantity
$$b(G) := \frac{\mc(G)}{m_G},$$
called the \emph{bipartite density} for various classes of graphs. The term is called as such since the maximum cut value of a graph is equivalent to the number of edges of the largest bipartite subgraph of a graph. In some of the structural graph theory literature, in addition to providing a bound of the form $b(G) \geq \xi$, such work also provides (or implies) a polynomial time algorithm for finding a bipartite graph (and a corresponding cut) that contains a $\xi$ fraction of the edges, i.e., because of Lemma \ref{thm:edge_ratio_maxcut_relation}, these are $\xi$-approximation algorithms for \MC{}.

Erd\"os showed that if $G$ is $2r$-colorable, then $b(G) \geq \frac{r}{2r-1}$ \cite{erdHos1975problems}. Sub-cubic graphs are known to be 4-colorable (as a result of Brook's theorem \cite{brooks1941colouring}); thus, for sub-cubic graphs, we have that $b(G) \geq 2/3$. This bound is tight for $K_4$, i.e., $b(K_4) = 2/3$. 
Locke \cite{locke1982maximum} and Staton \cite{staton1980edge} show that for all cubic graphs other than $K_4$, that $b(G) \geq 7/9$ and that there are infinitely many graphs where the bound achieves equality, i.e., $b(G) = 7/9$.
Hopkins and Staton \cite{hopkins1982extremal} showed that for triangle-free cubic graphs, $b(G) \geq 4/5$; later, Bondy and Locke \cite{bondy1986largest} extended this bound to triangle-free sub-cubic graphs and came up with an algorithm that produces a bipartite subgraph with at least $4/5$ of the edges. Zhu \cite{zhu2009bipartite} showed that aside from 7 exceptions, all 2-connected triangle-free subcubic graphs satisfy $b(G) \geq 17/21 \approx 0.80952$; moreover, it is shown that the proof gives rise to a linear-time algorithm which yields a cut with at least a $17/21$ fraction of the edges. For convenience, we will denote the set of 7 exceptions as $\mathbf{Bad}$.  


The graphs in $\mathbf{Bad}$ are exactly the set of triangle-free subcubic graphs for which $b(G) = 4/5$. The exceptions in $\mathbf{Bad}$ were found by Bondy and Locke \cite{bondy1986largest} and moreover, they showed that the 2 exceptions (amongst the 7) which are cubic are the \emph{only} triangle-free cubic graphs for which $b(G) = 4/5$. Bondy and Locke conjectured that the graphs in $\mathbf{Bad}$ were the only triangle-free sub-cubic graphs for which $b(G) = 4/5$; this conjectured was later proved by Xu and Yu \cite{xu2008triangle}. 

In the realm of optimization, many \mc{} algorithms yield approximation ratios better than the algorithms described above, e.g., the Goemans-Williamson algorithm \cite{goemans1995improved} is a $0.878$-approximation algorithm for general graphs (with positive weights). For cubic graphs, Halperin et al. \cite{halperin2004max} proposed an SDP-based 0.9326-approximation algorithm and a combinatorial $22/27\approx 0.8148$-approximation algorithm. For subcubic graphs, Bazgan and Tuza \cite{bazgan2008combinatorial} provide a combinatorial algorithm with a $5/6\approx 0.8333$-approximation ratio. We will later see in this work that a warm-start obtained from a cut with a large fraction of the total number of edges will be of importance; however, to the authors' knowledge, the algorithms just described \cite{goemans1995improved,halperin2004max,bazgan2008combinatorial} do not have such guarantees on the total fraction of edges cut beyond the trivial guarantee in Proposition \ref{thm:approxRatioToFractionEdgesCut}.

\begin{prop}
\label{thm:approxRatioToFractionEdgesCut}
    Let $\mathcal{A}$ be an $\alpha$-approximation \mc{} algorithm on cubic graphs. Then, with the exception of graphs isomorphic to $K_4$, $\mathcal{A}$ will return cuts that, in expectation, cut a $(7/9)\alpha$-fraction of the edges.
\end{prop}
\begin{proof}
    Let $\mathcal{A}(G)$ denote the expected cut value obtained by running $\mathcal{A}$ on a cubic graph $G$. Then,
    $$\mathcal{A}(G) = \alpha \cdot \mc(G) = \alpha \cdot (b(G)m_G) \geq (7/9)\alpha m_G,$$
    and thus,
    $$\mathcal{A}(G)/m_G \geq (7/9)\alpha.$$
    Here, we used that for cubic graphs that are non-isomorphic to $K_4$, that $b(G) \geq 7/9$ \cite{locke1982maximum,staton1980edge}.
\end{proof}

We say that an algorithm is of type $\mathcal{A}_\kappa$ if it takes as input, a 3-regular graph $G$ (possibly with additional restrictions) and returns a cut $b$ that contains at least a $\kappa$-fraction of the edges, i.e., $\cut(b)/m_G \geq \kappa$ where $b$ is also locally-optimal with respect to 1-BLS. Similarly, we say that an algorithm is of type $\mathcal{A}'_\kappa$ if it is of type $\mathcal{A}_\kappa$ but the guarantee is with respect to 3-regular triangle-free graphs (possibly with additional restrictions). By Lemma \ref{thm:edge_ratio_maxcut_relation}, algorithms of type $\mathcal{A}_\kappa$ (and $\mathcal{A}'_\kappa$) yield an approximation ratio of at least $\kappa$. With this notation, the 1-BLS algorithm is of type $\mathcal{A}_{2/3}$. Post processing the output of the algorithms by \cite{bondy1986largest} and \cite{zhu2009bipartite} with the 1-BLS algorithm yields algorithms of type $\mathcal{A}'_{4/5}$ and $\mathcal{A}'_{17/21}$ respectively. 

\subsection{QAOA}
\label{sec:QAOA}
We next review the QAOA algorithm and set up the needed notation that will be used throughout this work. We use $X,Y,Z$ to denote the standard Pauli matrices:
$$X = \begin{bmatrix}0 & 1 \\ 1 & 0\end{bmatrix}, \quad Y = \begin{bmatrix} 0 & -i \\ i & 0\end{bmatrix}, \quad Z = \begin{bmatrix} 1 & 0 \\ 0 & -1\end{bmatrix}.$$
For a multi-qubit system, we use $X_j, Y_j, Z_j$ to denote the operation of applying $X,Y,$ or $Z$ to the $j$th qubit respectively. We use $I$ and $\mathbf{0}$ to denote the identity and all-zeros matrix respectively; the dimensions of such matrices will be clear from context. For any square matrix $M$ and scalar $t$, we define the corresponding matrix,
$$U(M, c) := e^{-itM};$$
it is known that when $M$ is a Hermitian, then $U(M,t)$ must be unitary. 
For a combinatorial optimization problem determined by a cost function $c: \{0,1\}^n \to \mathbb{R}$ on $n$-length bitstrings, we define the corresponding cost Hamiltonian as the matrix $C$ such that $C\ket{b} = c(b)\ket{b}$.

Given a cost function $c$ on $n$-length bitstrings (with corresponding cost Hamiltonian $C$), a Hermitian matrix $B$ (called the mixing Hamiltonian) of appropriate size, and initial quantum state $\ket{s_0}$ in a $2^n$-dimensional Hilbert space, a circuit depth $p$, and variational parameters $\gamma = (\gamma_1, \dots, \gamma_p), \beta = (\beta_1, \dots, \beta_p)$, we define the following variational waveform of depth-$p$ QAOA as follows:
\begin{equation}\label{eqn:waveform}\ket{\psi_p(\gamma,\beta)} := U(B, \beta_p)U(C, \gamma_p) \cdots U(B, \beta_1)U(C, \gamma_1)\ket{s_0};\end{equation}

notationally, we drop $p$ or $(\gamma,\beta)$ from $\ket{\psi_p(\gamma,\beta)}$ whenever they are clear from context.

We use $F_p(\gamma,\beta)$ to denote the expected cost value obtained by applying the cost function $c$ to a measurement of $\ket{\psi_p(\gamma,\beta)}$, i.e.,
\begin{equation}\label{eqn:expectedCost} F^{(p)}(\gamma,\beta) =\bra{\psi_p(\gamma,\beta)}C\ket{\psi_p(\gamma,\beta)}.\end{equation}

For maximization problems, we define
\begin{equation}\label{eqn:expectedCostAtOptimalParams}M_p = \max_{\gamma,\beta} F^{(p)}(\gamma,\beta),\end{equation}
to be the expected cost value of QAOA at an optimal choice of the variational parameters $\gamma$ and $\beta$; $M_p$ is similarly defined for minimization problems. 

For many optimization problems, the cost Hamiltonian $C$ can be compactly expressed as a sum of terms with each term acting on a small number of qubits. For the \MC{} problem, given a graph $G = (V,E)$, the cost Hamiltonian can be written as,
$$C_G = \sum_{e \in E} C_e,$$
where,
\begin{equation}\label{eqn: costSingleEdge}C_e = \frac{1}{2}(I - Z_iZ_j).\end{equation}

For unconstrained optimization problems, the mixing Hamiltonian $B$ is usually taken to be the transverse field mixer, which, for an $n$-qubit system, is defined as
$$B_\text{TF} = \sum_{j=1}^n X_j.$$
Additionally, the starting state is usually taken to be an equal superposition of all $2^n$ bitstrings of length $n$:
$$\ket{s_0} = \ket{+}^{\otimes n} = \frac{1}{\sqrt{2^n}}\sum_{b \in \{0,1\}^n} \ket{b}.$$

In the context of graph-theoretic optimization problems, if $H$ is a subgraph of $G$, it is convenient to define the portion of the transverse field mixer that only acts on the vertices of $H$:
$$B_\text{TF,H} = \sum_{v \in V(H)}X_v.$$

In the case where $\ket{s_0}$ is the most-excited state of $H_B$, there exists a choice of angles $\gamma$ and $\beta$ for which the QAOA circuit can be viewed as a Trotterization of the Quantum Adiabatic Algorithm which is known to, under mild assumptions\footnote{Farhi et al. \cite{farhi2014quantum} make a connection between QAOA and Quantum Adiabatic Computing and sketch a proof for why the standard QAOA (with tranverse-field mixer and equal-superposition initial state) converges as a result of the quantum adiabatic theorem. Recent work by Binkowski et al. \cite{binkowski2023elementary} formalizes the proof of convergence and provides details regarding the conditions that the mixing Hamiltonian needs to satisfy for convergence to occur. All the mixers discussed in this work are known to satisfy such conditions. }, converge to the optimal solution given enough time; in other words,
$$\lim_{p \to \infty} M_p = \mc(G).$$

Much of the notation used in this work is in terms of the circuit depth $p$. We will frequently drop $p$ from the notation whenever the context is clear. Much of this work primarily focuses on the case where $p=1$.

In regards to approximation ratio as discussed in Section \ref{sec:approximationRatio}, for standard QAOA on the \MC{} problem, we take $\mathcal{A}(G) = \text{QAOA}(G) =  M_p$ where the cost Hamiltonian used is the standard \MC{} cost Hamiltonian $C_G$ and the mixing Hamiltonian is the transverse field mixer $B_{\text{TF}}$.


\section{Initial Product States and Aligned Mixers}
\label{sec:initialStates}
For the standard QAOA algorithm \cite{farhi2014quantum} and many of its variants, the equal superposition $\ket{s_0} = \ket{+}^{\otimes n}$ is used as the initial state. For the \MC{} problem, quantum measurement of $\ket{+}^{\otimes n}$ produces a uniform distribution of all $2^n$ cuts in the graph; put another way, each vertex, independent of the other vertices, has probability $1/2$ of being on one side of the cut or the other.

However, one can consider modifying the QAOA algorithm by using a different initial state $\ket{s_0}$. Often, such initial states are constructed as a function of \emph{classically} obtained solutions. In the context of \MC{} on a graph $G=(V,E)$, we will describe one such initialization method that is determined by a cut $(S, V \setminus S)$ of $G$ and a parameter $\theta$ which we refer to as the \emph{initialization angle}; to this end, we first introduce some helpful notation.

\subsection{Construction of Warm-Started States}
\label{sec:constructionOfWarmStartedStates}
Let $\vec{n} = (x,y,z)$ be a unit vector written in Cartesian coordinates. We let $\ket{\vec{n}}$ denote the single-qubit quantum state whose qubit position on the Bloch sphere is $\vec{n}$. For $\vec{n}=(x,y,z)$, we define the following single-qubit operation:
$$B_{\vec{n}} = xX+yY+zZ,$$
and let $B_{\vec{n},j}$ denote the operation of applying the operation $B_{\vec{n}}$ on the $j$th qubit. The unitary $U(B_{\vec{n},j}, \beta)$ can be geometrically interpreted as a single-qubit rotation by angle $2\beta$ about the axis that points in the $\vec{n}$ direction \cite{blochSphereRotations}.

For an $n$-qubit product state $\ket{s} = \bigotimes_{j=1}^n \ket{\vec{n}_j}$, we define an $n$-qubit operation $B_{\ket{s}}$ in terms of the single-qubit operations above:
$$B_{\ket{s}} =  \sum_{j=1}^n B_{\vec{n_j}, j}.$$

When $\ket{s}$ is a product state, one can show that $\ket{s}$ can be prepared and that $B_{\ket{s}}$ can be implemented in most quantum devices with a constant-depth circuit using standard single-qubit rotation gates about the $x,y,$ and $z$ axes. Moreover, they \cite{Tate2023warmstartedqaoa} remark that $\ket{s}$ is a ground state of $B_{\ket{s}}$ for any product state $\ket{s}$. Thus, as remarked in \cite{Tate2023warmstartedqaoa}, under mild assumptions\footnote{To show convergence with increased circuit depth, it suffices to choose an initial product state $\ket{s_0}$ where none of the qubits are positioned at the poles of the Bloch sphere. In some cases, this restriction can be slightly relaxed, e.g., due to bitflip-symmetry of the cost function; for \MC{} it is not an issue if exactly one of the qubits are positioned at the poles of the Bloch sphere.} on an initial product state $\ket{s_0}$, running QAOA with the standard phase separator $C$ and mixer $B_{\ket{s_0}}$ will yield the optimal solution as the circuit depth goes to infinity (assuming that $\gamma$ and $\beta$ are chosen optimally) \cite{Tate2023warmstartedqaoa}. In general, whenever the initial state of QAOA is the ground state of the mixer, we say that the mixer is \emph{aligned} with the initial state, and refer to this category of QAOA variants as \emph{QAOA with aligned mixers}.

Extending Equations \ref{eqn:waveform} and \ref{eqn:expectedCost}, we next  define the output state and expected cost of QAOA with aligned mixers as a function of the initial state:
\begin{equation}\label{eqn:warmStartedWaveform}\ket{\psi_p(\gamma,\beta,\ket{s})} = U(B_{\ket{s}}, \beta_p)U(C, \gamma_p) \cdots U(B_{\ket{s}}, \beta_1)U(C, \gamma_1)\ket{s},\end{equation}
and
$$F^{(p)}(\gamma,\beta, \ket{s}) =\bra{\psi_p(\gamma,\beta,\ket{s})}C\ket{\psi_p(\gamma,\beta,\ket{s})}.$$

Similar to Equation \ref{eqn:expectedCostAtOptimalParams}, we also define the expected cost with optimal choice of variational parameters as a function of $\ket{s}$:
$$M_p(\ket{s}) = \max_{\gamma,\beta} F^{(p)}(\gamma,\beta, \ket{s}).$$

Prior approaches for warm-started QAOA considered initial states of the form $\ket{s_0} = \bigotimes_{j=1}^n \ket{\vec{n}_j}$ where $\vec{n}_1,\dots, \vec{n}_n$ are obtained by some classical procedure. In the work by Egger et al. \cite{egger2021warm}, for problems whose corresponding QUBO (Quadratic Unconstrained Binary Optimization) formulation satisfies certain properties, they solve a relaxation of the QUBO and map the solutions to states in,
$$\textbf{Arc} = \{\cos(\theta/2)\ket{0}+\sin(\theta/2)\ket{1} : \theta \in (0, \pi)\},$$
i.e., points on the Bloch sphere that intersect with the $xz$-plane with non-negative $x$ coordinate; the blue arc in Figure \ref{fig:singleCutInitialization} corresponds to the possible qubit positions. For the \MC{} problem, Tate et al. \cite{Tate2023warmstartedqaoa, tate2023bridging} consider higher-dimensional relaxations of the Max-Cut problem (i.e. the Burer-Monteiro relaxation and the relaxation used in the Goemans-Williamson algorithm) which, after a possible projection, yield points all over the surface of the Bloch sphere.\footnote{However, we remark that Tate et al. \cite{Tate2023warmstartedqaoa} show that for every initial product state $\ket{s_0} = \bigotimes_{j=1}^n \ket{\vec{n}_j}$, there exists a different initial product state $\ket{s_0'} = \bigotimes_{j=1}^n \ket{\vec{n}_j'}$ with $\vec{n}_j' \in \textbf{Arc}$ for all $j$, such that, up to a global phase, QAOA with aligned mixers and initial state $\ket{s_0}$ returns the same state as QAOA with aligned mixers and initial state $\ket{s_0'}$. This property holds any optimization problem and corresponding cost Hamiltonian $C$.}

While arbitrary product states provide a rich space of warm-starts to consider for QAOA which often empirically yield high-quality solutions \cite{Tate2023warmstartedqaoa}, the theoretical analysis of warm-started QAOA with such states is highly non-trivial. To simplify matters, this work considers a much more restricted set of warm-starts which we now describe. 

Given $\theta \in \RR$, we first define two single-qubit quantum states:
$$\ket{0_\theta} = \cos(\theta/2)\ket{0} + \sin(\theta/2)\ket{1},$$
$$\ket{1_\theta} = \sin(\theta/2)\ket{0}+\cos(\theta/2)\ket{1}.$$
Note that when $\theta=0$, then $\ket{0_\theta} = \ket{0}$ and $\ket{1_\theta} = \ket{1}$. It is helpful to note the following relation between these two states:
\begin{equation}\label{eqn:zeroOneRelationship}\ket{1_\theta} = \ket{0_{\pi - \theta}}.\end{equation}

In the case that $\theta \in [0,\pi]$, the states $\ket{0_\theta}$ and $\ket{1_\theta}$ can be geometrically seen as points along the set $\textbf{Arc}$ that are angle $\theta$ away from the north and south poles of the Bloch sphere respectively. In the context of \mc{}, we will usually take $\theta \in [0, \pi/2]$ so that the closest poles to $\ket{0_\theta}$ and $\ket{1_\theta}$ are the north and south poles (at $\ket{0}$ and $\ket{1}$) respectively; we will later justify this restriction in Theorem \ref{thm:bitflipSameExpectedCut}. This geometric interpretation can be seen in Figure \ref{fig:singleCutInitialization}.

\begin{figure}
\centering
    \pgfmathsetmacro{\r}{2.6} %

\tdplotsetmaincoords{80}{120}
\begin{tikzpicture}[
tdplot_main_coords,
font=\footnotesize,
Helpcircle/.style={gray!70!black,
},
]

\pgfmathsetmacro{\h}{0.9*\r} %

\pgfmathsetmacro{\t}{60}
\coordinate[label=left:{$\ket{0_\theta}$}] (X1) at ({\r*sin(\t)},0,{\r*cos(\t)});
\coordinate[label=left:{$\ket{1_\theta}$}] (X2) at ({\r*sin(\t)},0,{-\r*cos(\t)}); 

\coordinate (M) at (0,0,0);
\coordinate[label=$\ket{0}$] (Top) at (0,0,\r);
\coordinate[label=below:{$\ket{1}$}] (Bot) at (0,0,-\r);

\tdplotdrawarc{(M)}{\r}{-65}{110}{anchor=north}{}
\tdplotdrawarc[dashed]{(M)}{\r}{110}{295}{anchor=north}{}


\tdplotsetrotatedcoords{90}{90}{0}%
\tdplotdrawarc[tdplot_rotated_coords, blue]{(M)}{\r}{180}{360}{anchor=north}{}
\tdplotdrawarc[tdplot_rotated_coords, dashed]{(M)}{\r}{0}{180}{anchor=north}{}
\tdplotdrawarc[tdplot_rotated_coords]{(M)}{0.5*\r}{180}{180+\t}{anchor=north}{$\theta$}
\tdplotdrawarc[tdplot_rotated_coords]{(M)}{0.5*\r}{0}{-\t}{anchor=north}{$\theta$}
\draw[] (M) -- (X1);
\draw[] (M) -- (X2);
\draw[] (Top) -- (Bot);

\begin{scope}[tdplot_screen_coords, on background layer]
\fill[ball color= gray!20, opacity = 0.1] (M) circle (\r); 
\end{scope}

\foreach \P in {X1,X2}{
\shade[ball color=blue] (\P) circle (3pt);
}

\begin{scope}[-latex, shift={(M)}, xshift=1.5*\r cm, yshift=0.1*\r cm]
\foreach \P/\s/\Pos in {(1,0,0)/x/right, (0,1,0)/y/below, (0,0,1)/z/right} 
\draw[] (0,0,0) -- \P node[\Pos, pos=0.9,inner sep=2pt]{$\s$};
\end{scope}

\end{tikzpicture}
\caption{\footnotesize\label{fig:singleCutInitialization} A geometric depiction of the states $\ket{0_\theta}$ and $\ket{1_\theta}$ on the Bloch sphere. The blue half-circle, $\textbf{Arc}$, in the $xz$-plane denotes all the possible positions for $\ket{0_\theta}$ and $\ket{1_\theta}$ as $\theta$ varies from $0$ to $\pi$. }
\end{figure}
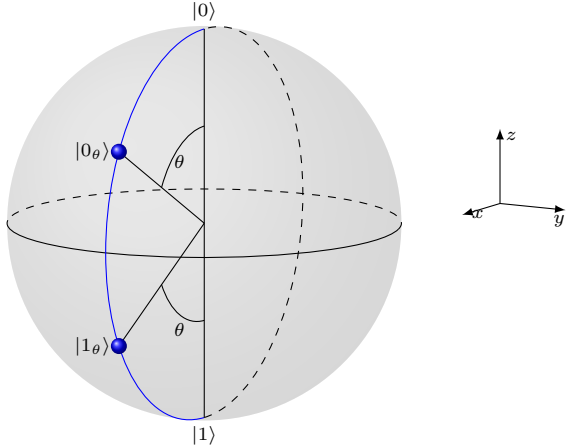

Given a bitstring $b \in \{0,1\}^n$, we define the initial state $\ket{b_\theta}$ as follows:

\begin{equation}\label{eqn:bitstringWarmstartState}\ket{b_\theta} = \ket{(b_1)_\theta} \otimes \ket{(b_2)_\theta} \otimes \cdots \otimes \ket{(b_n)_\theta},\end{equation}
i.e., the initial state is a product state where each qubit has initial position on the Bloch sphere of $\ket{0_\theta}$ or $\ket{1_\theta}$. Note that when $\theta=0$, we have that $\ket{b_\theta} = \ket{b}$. Additionally, when $\theta = \pi/2$, then $\ket{0_\theta} = \ket{1_\theta} = \ket{+}$ and hence $\ket{b_\theta} = \ket{+}^{\otimes n}$, the equal superposition used in the standard QAOA algorithm. Moreover, as a result of Equation \ref{eqn:zeroOneRelationship}, we also have that if $\bar{b}$ is the bitwise negation of $b$, then $\ket{b_\theta} = \ket{\bar{b}_{\pi-\theta}}$. For general $\theta \in [0, \pi]$, the probability of observing $\ket{b}$ as a result of measuring $\ket{b_\theta}$ is equal to $\cos^{2n}(\theta/2)$.

For the \MC{} problem, previous works have considered QAOA with initial states of the form $\ket{b_\theta}$ where $b$ typically corresponds to a ``good" cut. Egger et al. \cite{egger2021warm} consider a QAOA-variant (different than their QUBO-relaxation variant) where an initial state of the form $\ket{b_\theta}$ is used with $\theta = \pi/3$, but with an \emph{unaligned} mixer that has the property that there exists parameters $\gamma$ and $\beta$ such that depth-1 QAOA recovers the cut corresponding to $b$ (i.e. it returns exactly $\ket{b}$). Cain et al. \cite{cain2022qaoa} considered QAOA with initial states of the form $\ket{b}$ (i.e. $\theta = 0$) with the transverse field mixer; it was found that when the cut associated with $b$ was a ``good" cut, then this variant of QAOA yields little to no improvement regardless of the circuit depth used.

\subsection{Approximation Ratio for Warm-Started QAOA}
\label{sec:warmstarted_QAOA_AR}
In addition to previous discussion (Section \ref{sec:approximationRatio}) regarding whether or not the optimization of the variational parameters should be included as part of algorithm and its corresponding approximation ratio, there is some additional nuance in regards to defining the approximation ratio for warm-started QAOA. First, similar to $\gamma$ and $\beta$, we could also consider finding the instance-specific optimal choice of initialization angle $\theta$ as part of the warm-start QAOA algorithm. Instead, for this work, in order to see the effects of different choices of $\theta$, we consider various choices of fixed $\theta$ that are independent of the graph instance. Given a graph $G$ and a bitstring $b \in \{0,1\}^{n_G}$, we first define the instance-specific approximation ratio for depth-$p$ warm-stated QAOA on $G$ with initial state $\ket{b_\theta}$:
$$\alpha_{\text{WS-QAOA}_\theta^{(p)}}(G,b) = \frac{M_p(\ket{b_\theta})}{\mc(G)}.$$

Often, instead of working with a specific bitstring $b$ to generate the initial state $\ket{b_\theta}$, we consider a bitstring that is generated as a result of some classical algorithm $\mathcal{A}$; we will refer to the overall algorithm (with initialization angle $\theta$) with circuit depth $p$ as $(\mathcal{A}+\text{QAOA})_\theta^{(p)}.$

Note that for $(\mathcal{A}+\text{QAOA})_\theta^{(p)}$, if $\mathcal{A}$ itself is a randomized algorithm, then the overall algorithm has two sources of randomness: (1) randomness of the initial quantum state $\ket{b_\theta}$ as a result of the randomization of $\mathcal{A}$, and (2) the randomness inherent in quantum measurement. Letting $\Pr_\mathcal{A}(G,b)$ be the probability that the classical algorithm $\mathcal{A}$ with input $G$ returns bitstring $b$, defining the approximation ratio over all sources of expectation yields:
$$\alpha_{{(\mathcal{A}+\text{QAOA}})_\theta^{(p)}}(G) = \sum_{b \in \{0,1\}^n} \textstyle\Pr_\mathcal{A}(G,b) \cdot \alpha_{\text{WS-QAOA}_\theta^{(p)}}(G,b).$$

Reasoning over all possible cuts $b$ in the expression above does somewhat complicate the analysis. Instead, we can obtain a lower bound $\alpha'$ on this approximation ratio by instead reasoning over the cut $b$ (produced by $\mathcal{A}$) that is the worst for warm-started QAOA:
$$\alpha_{{(\mathcal{A}+\text{QAOA}})_\theta^{(p)}}(G) \geq \alpha'_{{(\mathcal{A}+\text{QAOA}})_\theta^{(p)}}(G) =  \min_{\substack{b \in \{0,1\}^n: \\ \Pr_\mathcal{A}(G,b) > 0} }\alpha_{\text{WS-QAOA}_\theta^{(p)}}(G,b).$$

In the case that $\mathcal{A}$ is a local-search algorithm such as Algorithm \ref{alg:localSearch}, then the condition {$\Pr_\mathcal{A}(G,b) \allowbreak > 0$} implies that the bitstring $b$ is locally-optimal with respect to algorithm $\mathcal{A}$ on graph $G$; the converse is also true assuming that $\mathcal{A}$ is \emph{initialized} with $b$ with non-zero probability (as is the case with Algorithm \ref{alg:localSearch}). Using the above lower bound, we can thus bound the theoretical approximation ratio as:
$$\alpha_{{(\mathcal{A}+\text{QAOA}})_\theta^{(p)}} \geq \alpha'_{{(\mathcal{A}+\text{QAOA}})_\theta^{(p)}} = \min_G \min_{\substack{b \in \{0,1\}^n: \\ \Pr_\mathcal{A}(G,b) > 0} }\alpha_{\text{WS-QAOA}_\theta^{(p)}}(G,b),$$
i.e. we minimize over all choices of graph $G$ and potential bitstrings $b$ that are produced by $\mathcal{A}$. 

Much of the notation above contains the superscript $p$ to denote the circuit depth; for convenience, we will often omit this superscript in this work whenever $p=1$.

\subsection{Basic Properties of Warm-Started QAOA}
For \MC, we (intuitively) expect that warm-started QAOA will yield better cuts if the bitstring $b$ used in the initialization $\ket{b_\theta}$ is better. Proposition \ref{thm:depth0ApproxRatio} and Corollary \ref{thm:depth0ApproxRatio2} show that this is the case for depth-0 QAOA (i.e. simplying measuring the initial state); we later show a similar result for depth-1 QAOA as well (Section \ref{sec:suitableConstraints}).


\begin{prop}
\label{thm:depth0ApproxRatio}
    Let $G = (V,E)$ be a graph and let $b$ be a bitstring corresponding to a cut that cuts exactly $\kappa$ fraction of the edges and let $\theta \in [0,\pi/2]$. Then the algorithm that simply measures $\ket{b_\theta}$ has an approximation ratio of at least $\frac{1}{4}\left( (2\kappa - 1)\cos(2\theta)+2\kappa+1\right).$
\end{prop}
\begin{proof}
Suppose that the bitstring $b$ corresponds to a cut that cuts at least an $\kappa$ fraction of the total number of edges (of which there are $m$ in total). Then the expected cut value is given by,

$$\kappa m(\sin^4(\theta/2)+\cos^4(\theta/2)) + (1-\kappa)m(2\sin^2(\theta/2)\cos^2(\theta/2)),$$

a detailed derivation of this calculation can be found in \cite{tate2023bridging}. Via some basic trigonometric identities and algebraic manipulation, it can be shown that the above expression is equivalent to,

$$\frac{m}{4}\left( (2\kappa - 1)\cos(2\theta)+2\kappa+1\right).$$

From this, quantum sampling of the state $\ket{b_\theta}$ yields an approximation ratio that is at least
$$\frac{\frac{m}{4}\left( (2\kappa - 1)\cos(2\theta)+2\kappa+1\right)}{\mc(G)} \geq \frac{\frac{m}{4}\left( (2\kappa - 1)\cos(2\theta)+2\kappa+1\right)}{m} = \frac{1}{4}\left( (2\kappa - 1)\cos(2\theta)+2\kappa+1\right).$$
\end{proof}

The bound on the approximation ratio above is in the case that the cut determined by $b$ cuts \emph{exactly} an $\kappa$ fraction of the edges; however, Corollary \ref{thm:depth0ApproxRatio2} below shows that this is also true if we know that $b$ cuts \emph{at least} an $\kappa$ fraction of the edges.

\begin{corollary}
\label{thm:depth0ApproxRatio2}
    Let $G = (V,E)$ be a graph and let $b$ be a bitstring corresponding to a cut that cuts at least an $\kappa$ fraction of the edges and let $\theta \in [0,\pi/2]$. Then the algorithm that simply measures $\ket{b_\theta}$ has an approximation ratio of at least $\frac{1}{4}\left( (2\kappa - 1)\cos(2\theta)+2\kappa+1\right).$
\end{corollary}
\begin{proof}
    Let $\kappa' \geq \kappa$ be the actual fraction of edges that are cut by the cut corresponding to $b$. Then, by the calculations above, the algorithm that measures $\ket{b_\theta}$ has an approximation ratio of at least,
    $$\frac{1}{4}\left( (2\kappa' - 1)\cos(2\theta)+2\kappa'+1\right) \geq \frac{1}{4}\left( (2\kappa - 1)\cos(2\theta)+2\kappa+1\right),$$
    where the inequality holds as $\frac{1}{4}\left( (2\kappa - 1)\cos(2\theta)+2\kappa+1\right)$ is a non-decreasing function in $\kappa$ (for fixed $\theta$). To see this, observe that $$\frac{\partial}{\partial \kappa}\left(\frac{1}{4}\left( (2\kappa - 1)\cos(2\theta)+2\kappa+1\right)\right) = \frac{1}{2}\left(\cos(2\theta)+1\right) \geq \frac{1}{2}\left(-1+1\right) = 0.$$
\end{proof}

If we wish to choose $\theta \in [0, \pi/2]$, so that the approximation ratio bound above is maximized, then the choice of $\theta$ will depend on $\kappa$. If $\kappa > \frac{1}{2}$, then $\theta = 0$ is the optimal choice, otherwise, if $\kappa < \frac{1}{2}$, then $\theta = \pi/2$ is the optimal choice. For $\kappa = \frac{1}{2}$, we have that $\frac{1}{4}\left( (2\kappa - 1)\cos(2\theta)+2\kappa+1\right)$ is a constant (equal to 0.5) independent of the choice of $\theta$.

Before jumping into the analysis of positive-depth warm-started QAOA, we make one more key observation: if $\bar{b}$ is the bitwise-negation of $b$, then, with aligned mixers, the expected cut value of QAOA initialized with $\ket{b_\theta}$ and $\ket{\bar{b}_\theta}$ is the same.
\begin{theorem}
\label{thm:bitflipSameExpectedCut}
    Let $G=(V,E)$ be a graph. Let $b \in \{0,1\}^n$ be a bitstring, let $\bar{b}$ be its bitwise-negation, let $\theta \in \RR$, let $\gamma$ and $\beta$ be any choice of variational parameters, and let $p$ be any non-negative circuit depth. Then, for the \MC{} problem on $G$, and for both choices of initialization ($\ket{b_\theta}$ and $\ket{\bar{b}_\theta}$), we have that the expected cut size obtained by warm-started QAOA is the same, i.e., 
    $$F^{(p)}(\gamma,\beta, \ket{b_\theta}) = F^{(p)}(\gamma, \beta, \ket{\bar{b}_\theta}).$$

\end{theorem}
The key idea for the proof of Theorem \ref{thm:bitflipSameExpectedCut} is that classically, the \MC{} problem has bitflip-symmetry, i.e., for any bitstring $b$ representing a cut, both $b$ and $\bar{b}$ cut the same set of edges. Moreover, the initial state $\ket{b_\theta}$ and corresponding mixer $B_{\ket{b_\theta}}$ are constructed in a way that allow this symmetry to be exploited. The proof of Theorem \ref{thm:bitflipSameExpectedCut} follows in a straightforward manner and the main results of our work do not hinge on this result and thus, we do not include the proof here; the interested reader can instead find a proof in the Supplementary Materials.

Theorem \ref{thm:bitflipSameExpectedCut} can be used to show that for \mc{} warm-started QAOA, it suffices to only consider $\theta \in [0,\pi/2]$. First, it was shown by Tate et al. \cite{Tate2023warmstartedqaoa} that for warm-started QAOA whose initial state is a product state, it suffices to only consider qubit positions along $\mathbf{Arc}$ and hence, we know we can restrict $\theta$ to the interval $[0, \pi]$. Moreover, if we pick $\pi/2 < \theta \leq \pi$, then we have that,
$$F^{(p)}(\gamma,\beta, \ket{b_\theta}) = F^{(p)}(\gamma, \beta, \ket{\bar{b}_\theta}) = F^{(p)}(\gamma, \beta, \ket{b_{\pi-\theta})},$$
where the last equality holds as $\ket{\bar{b}_\theta} = \ket{b_{\pi-\theta}}$ and hence QAOA on bitstring $b$ and angle $\pi/2 < \theta \leq \pi$ yields the same result as QAOA on the same bistring and angle $0 < \theta' \leq \pi/2$ with $\theta' = \pi - \theta$. Thus, for the \mc{} problem, it suffices to only consider only $\theta \in [0, \pi/2]$.

\section{Approximation Ratios for Depth-1 Warm-Started QAOA on 3-Regular Graphs}
\label{sec:approximationRatiosDepth1WarmStartQAOA}
As shown in Farhi et al.'s seminal QAOA paper, QAOA is often referred to as a local quantum algorithm, e.g., in the context of \MC, the result of QAOA is determined by the local structure around each edge in the graph. We formalize this notion for the \MC{} problem below in Section \ref{sec:edgeNeighborhoods} and review the technique for obtaining the $0.6924$ approximation ratio for standard QAOA on 3-regular graphs found Farhi et al. \cite{farhi2014quantum} in Section \ref{sec:graphStructures}. In this work, we extend the notion of local structure to account for information encoded in the warm-start (Section \ref{sec:coloredEdgeNeighborhoods}) and proceed to extend the technique of Farhi et al. to obtain approximation ratios for warm-started QAOA initialized with $\ket{b_\theta}$ as a function of the initialization angle $\theta$ (Sections \ref{sec:coloredGraphStructures} and \ref{sec:suitableConstraints}).

\subsection{Edge-Neighborhoods for Max-Cut}
\label{sec:edgeNeighborhoods}
In this subsection, we review how the expected cost of standard \MC{} QAOA can be calculated by only considering the local structure around each edge in a graph, which we refer to as the edge-neighborhood. The results in this subsection are effectively identical to those found in Farhi et al.'s original work; we review these results here as they will later be extended for the warm-start case in Section \ref{sec:coloredEdgeNeighborhoods}. We begin with some needed definitions and notation regarding these edge-neighborhoods.

Fix the circuit depth $p$. For each edge $e$ in a graph $G=(V,E)$, we define a (marked) subgraph $(G^{(p)}_e, e)$, which we call the \emph{(depth-$p$) edge-neighborhood centered about $e$}, recursively; we also refer to $e$ as the \emph{central edge} of $(G^{(p)}_e, e)$. We will often simply write $G^{(p)}_e$ to mean $(G^{(p)}_e, e)$ since the marked (central) edge is clear from context. When $p=0$, $G^{(0)}_e$ is simply defined as  the graph consisting of only the edge $e$ (and its vertices). The graph $G^{(p)}_e$ is defined as the edge-induced subgraph of $G$ whose edges are all the edges of $G^{(p-1)}_e$ as well as any edges in $G$ that are incident to at least one edge in $G^{(p-1)}_e$. This paper primarily focuses on the depth $p=1$ case, so we will often drop the superscript of $G^{(p)}_e$ and simply write $G_e$ when this is the case.


For two edges $e,f$ in $G$, we say that $G_e^{(p)}$ and $G_f^{(p)}$ are of the same \emph{type} and write $G_e^{(p)} \sim G_f^{(p)}$ if $G_e^{(p)}$ is isomorphic to $G_f^{(p)}$ (via some graph isomorphism $\varphi$ from $V(G_e^{(p)})$ to $V(G_f^{(p)})$) and if $\varphi(e) = f$. Note that $\sim$ is an equivalence relation. For any graph $G$, let $\mathcal{G}(G)$ be a set of equivalence classes for $\sim$ that exist in $G$. For any $e \in E$, we refer to the equivalence class $[G_e^{(p)}]$ as the edge-neighborhood \emph{type} centered around $e$. For each $[g] \in \mathcal{G}$, let $d_g(G)$, the degeneracy of $g$, be the number of times that $[g]$ appears as an edge-neighborhood type in $G$, i.e., 
$$d_g(G) = |\{e \in E(G): g \sim G_e^{(p)}\}|;$$
when the graph $G$ is clear from context, we will simply write $d_g$ instead of $d_g(G)$. Figure \ref{fig:graphDecompExample} provides an example on a 5-node graph that illustrates how the degeneracies of edge-neighborhood types are counted.

\begin{figure}
    \centering

    \hspace{1.5cm}\begin{tabular}{lc}
  $G:$ &
  \begin{minipage}{0.6\textwidth}\includegraphics[scale=0.5]
    {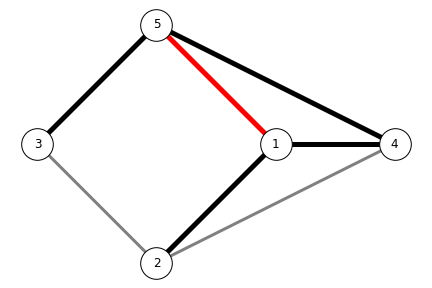}\end{minipage}\\
\end{tabular}

    \vspace{1cm}
    
    \begin{tabular}{c|c|c}
    $H_1:$ & $H_2:$ & $H_3:$\\
  \raisebox{0.8\height}{\includegraphics[scale=0.2] {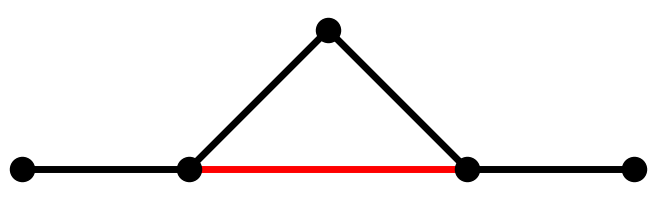}} &
  \includegraphics[scale=0.2]
    {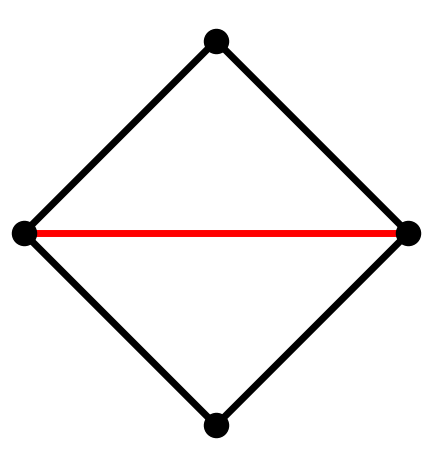} &
  \includegraphics[scale=0.2]
    {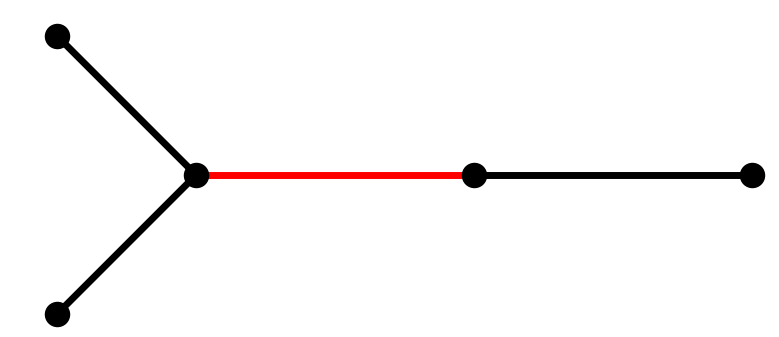}\\
    $d_{H_1} = 4:$ & $d_{H_2} = 1:$ & $d_{H_3} = 2:$ \\
    $\{1,2\}, \{1,5\}, \{2,4\}, \{4,5\} $ & $\{1,4\}$& $\{2,3\}, \{3,5\}$
\end{tabular}
\caption{\footnotesize\label{fig:graphDecompExample} The graph $G$ has 3 different edge neighborhood types, which we refer to as $H_1, H_2, H_3$ in this example. For each edge neighborhood type, the central edge is depicted in red. The degeneracies for each edge neighborhood type is also shown along with the edges associated with the degeneracy; note that the sum of the degeneracies adds up to the number of edges in $G$. In $G$, for the edge $e = \{1,5\}$, the depth-1 edge neighborhood $G_e$ is shown with solid black edges with the central edge $e$ in red; note that $G_e$ is of the same type as $H_1$. 
}
\end{figure}

Let $[g] \in \mathcal{G}$ be an edge-neighborhood type with central edge $e$, then we define:

$$f^{(p)}_g(\gamma,\beta) = \bra{\psi^{(p)}_g(\gamma,\beta)}C_e\ket{\psi^{(p)}_g(\gamma,\beta)},$$
where $\ket{\psi_g^{(p)}(\gamma,\beta)}$ corresponds to the quantum state resulting from running the standard depth-$p$ \MC{} QAOA circuit on the subgraph $g$ (see Equation \ref{eqn:waveform}); also note that the observable $C_e$ above corresponds to an indicator for whether or not the edge $e$ is cut (see Equation \ref{eqn: costSingleEdge}). In the case that $g = G_e$ is an edge-neighborhood of $G$ with central-edge $e$, Farhi et al. \cite{farhi2014quantum} showed that the probability of the edge $e$ being cut as a result of the standard QAOA on $G$ is exactly equal to $f_g^{(p)}(\gamma, \beta)$ and hence we refer to $f_g^{(p)}(\gamma, \beta)$ as the \emph{central-edge cut probability of $g$}. Due to linearity, an immediate consequence of this is that the expected cut value of QAOA, $F^{(p)}_G$, can be expressed in terms of the edge-neighborhood types found in $G$ and their degenerencies, i.e., 

\begin{equation}\label{eqn:subgraphDecomposition} F^{(p)}_G(\gamma,\beta) = \sum_{[g] \in \mathcal{G}(G)} d_g f_g^{(p)}(\gamma,\beta).\end{equation}


The fact that $f_g^{(p)}$ corresponds to edge-cut probabilities can be proven by keeping track of the qubits of the final QAOA state that are correlated with one another via the terms in the cost Hamiltonian and applying a commutativity argument; we refer the reader to Farhi et al.'s original QAOA paper \cite{farhi2014quantum} and the work of Love and Wurtz \cite{wurtz2021maxcut} for a more detailed proof of this and Equation \ref{eqn:subgraphDecomposition}.

Given Equation \ref{eqn:subgraphDecomposition}, observe that if the number of edge-neighborhood types for a graph $G$ is relatively small and if the (representatives of the) edge-neighborhood types have few vertices, then $F_G^{(p)}$ can be calculated on a classical device (via simulation) relatively quickly, even if the overall graph $G$ has a large number of vertices.

We make this more concrete by considering depth-1 \MC{} QAOA on 3-regular graphs. In this setting, there are 3 different edge-neighborhood types that can appear; we denote these types by $[g_4], [g_5],$ and $[g_6]$ whose graphs have $4,5,$ and $6$ vertices respectively as seen in Figure \ref{fig:cubicSubgraphTypes}. On a classical device, for any given cubic graph $G$ on $n$ vertices and parameters $\gamma$ and $\beta$, one can quickly calculate $f_{g_4}(\gamma,\beta), f_{g_5}(\gamma,\beta),$ and $f_{g_6}(\gamma,\beta)$ and moreover the degeneracies $d_{g_4}, d_{g_5}, d_{g_6}$ can be classically calculated in $\mathcal{O}(\text{poly}(n))$, and thus, by Equation \ref{eqn:subgraphDecomposition}, the expected cut value obtained via QAOA at specific parameters can be quickly determined. This technique can be used to more quickly find optimal $(\gamma,\beta)$ angles. While this technique can be used to quickly find expectation values classically, one still needs an $n$-qubit quantum device to actually observe the bitstrings that yield such expectations. 

\begin{figure}
    \centering
   \begin{tabular}{c|c|c}
    $g_5:$ & $g_4:$ & $g_6:$\\
 \raisebox{0.9cm}{\includegraphics[scale=0.2]{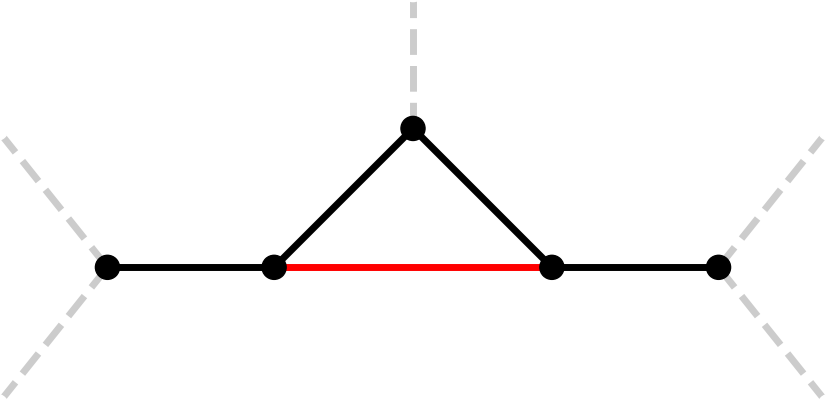}} &
  \includegraphics[scale=0.2]
    {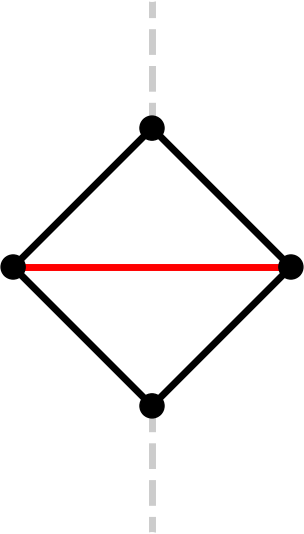} &
  \raisebox{0.5cm}{\includegraphics[scale=0.2]
    {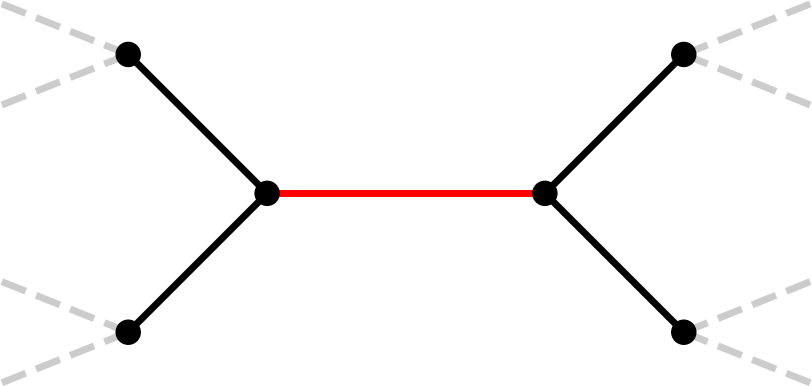}}
\end{tabular}
\caption{\footnotesize\label{fig:cubicSubgraphTypes} The three depth-1 edge-neighborhood types, labeled $g_4,g_5,g_6$, that can be found in cubic graphs. For each edge neighborhood type, the red edge denotes the central edge and the solid black lines denote the actual edges in the edge neighborhood; the gray dashed edges represent edges outside the edge-neighborhood that reside in the remainder of the cubic graph. Note that the other endpoints of the gray edges (not shown) are not necessarily distinct. }
\end{figure}

\subsection{Colored Edge-Neighborhoods for Max-Cut}
\label{sec:coloredEdgeNeighborhoods}
If we use a warm-start $\ket{b_\theta}$ such as that described in Section \ref{sec:initialStates}, then simply counting the edge-neighborhood types as described earlier does not suffice since the information regarding the initial qubit positions is lost. Instead, we consider \emph{colored} edge-neighborhood types in order to encode such information about the initial qubit positions. In this section, we consider 2-colorings\footnote{Our notion of graph coloring should not be confused with what is sometimes called a \emph{proper} graph coloring in the graph theory literature; i.e., in our colorings, vertices that are adjacent are permitted to have the same color. The coloring of vertices in this work is primarily used as a visual aid.} of graphs with each color corresponding to one of the 2 different qubit positions on the Bloch sphere in the kinds of warm-starts $\ket{b_\theta}$ described in Section \ref{sec:initialStates}. If one considers $k$ possible qubit positions on the Bloch sphere, then one may generalize the approach below to $k$-colored graphs; however, the analysis becomes exponentially more difficult as $k$ increases, hence, we limit ourselves to $k=2$ in this work. In this section, we show that by coloring the edge-neighborhoods, we can show a result that is similar to Equation \ref{eqn:subgraphDecomposition} for warm-started QAOA, i.e., that the expected cut value can be determined by the local information around each edge.

For the rest of this section, we consider the initialization angle $\theta$ of the warm-start state $\ket{b_\theta}$ to be fixed; the analysis for each value of $\theta$ will be nearly identical. Now, consider running \MC{} QAOA with graph $G$ with $n$ vertices and a warm-start $\ket{b_\theta}$ obtained via bitstring $b \in \{0,1\}^n$. The bitstring induces a colored graph $(G,b)$, this is visually seen in Figure \ref{fig:labeledGraphDecompExample} where we use the convention that 0 denotes the color yellow and 1 denotes the color green. For each $e \in E(G)$, we define $G^{(p)}_e(b)$ the \emph{colored} edge-neighborhoods about $e$ as colored mark graphs (Section \ref{sec:notationBackground}): the underlying graph and the central edge is defined in the same way compared to the (uncolored) edge-neighborhood $G^{(p)}_e$ and the vertices of $G^{(p)}_e(b)$ are colored so that they are consistent with the coloring of the overall graph coloring determined by $b$. We say that two colored edge-neighborhoods $(G,e,b)$ and $(H,f,b')$ are of the same type if there exists a graph isomorphism $\varphi: V(G) \to V(H)$ with the following properties:
\begin{enumerate}
    \item We have that $\varphi(e) = \varphi(f)$, i.e., the central edges are mapped to one another.
    \item For all $u \in V(G)$, if $\varphi(u) = v$, then $b_u = b'_v$, i.e., $\varphi$ is color-preserving with respect to the graph colorings $b$ and $b'$.
\end{enumerate}

Figure \ref{fig:labeledGraphDecompExample} illustrates different colored edge-neighborhoods types for a given graph. We use $\mathcal{G}(G,b)$ to represent the equivalence classes corresponding to the colored edge-neighborhood types that exist in the colored graph $(G,b)$. Similar to what was seen in Section \ref{sec:edgeNeighborhoods}, for a colored-edge neighborhood $g$ whose bitstring coloring is given by $b_g$, we can define the \emph{central-edge cut probability of $g$} as follows:
$$f^{(p)}_g(\gamma,\beta,\theta) = \bra{\psi^{(p)}_g(\gamma,\beta,\ket{(b_g)_\theta})}C_e\ket{\psi^{(p)}_g(\gamma,\beta,\ket{(b_g)_\theta})},$$

where $\ket{\psi^{(p)}_g(\gamma,\beta,\ket{(b_g)_\theta})}$ is the output of warm-started QAOA run on the underlying graph of $g$ with warmstart $(b_g)_\theta$ as defined in Equations \ref{eqn:warmStartedWaveform} and \ref{eqn:bitstringWarmstartState}.

The degeneracy of a colored edge-neighborhood type is also defined similarly as before and moreover, the overall expected cut value of QAOA can also be expressed in terms of the degeneracies of the colored edge-neighborhood types as seen in the following theorem.

\begin{theorem}
\label{thm:coloredSubgraphDecomposition}
Let $G$ be a graph, $b \in \{0,1\}^{n_G}$ be a bitstring, $p \in \mathbb{N}$ be a circuit depth, and let $\theta \in [0, \pi]$ be an initialization angle, then the expected cut value of warm-started QAOA on $G$ with warm-start $\ket{b_\theta}$ can be expressed as follows:
\begin{equation}\label{eqn:coloredSubgraphDecomposition} F^{(p)}_G(\gamma,\beta, \ket{b_\theta}) = \sum_{[g] \in \mathcal{G}(G,b)} d_g f_g^{(p)}(\gamma,\beta, \theta).\end{equation}
\end{theorem}
\begin{proof}
    The proof is nearly identical to the proof of Equation \ref{eqn:subgraphDecomposition} as the warm-start state does not change which terms can commute past the observable $C_e$ for each edge $e$ in the graph.
\end{proof}

\begin{figure}
    \centering

    \hspace{1.5cm}\begin{tabular}{lc}
  $G:$ &
  \begin{minipage}{0.8\textwidth}\includegraphics[scale=0.5]
    {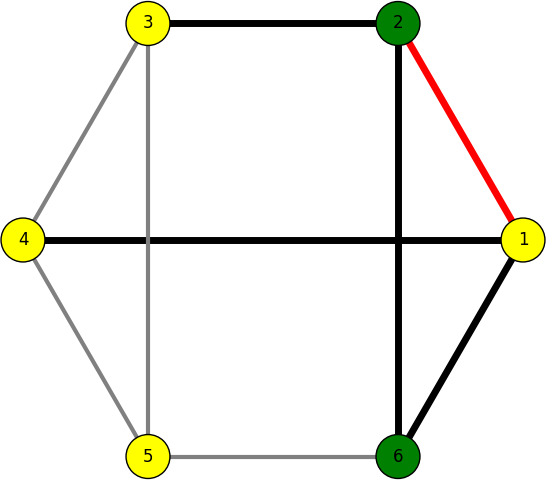}\end{minipage}\\
\end{tabular}

    \vspace{1cm}

    \begin{tabular}{ccc}
    $H_1:$ & $H_2:$ & $H_3:$\\
  \includegraphics[scale=0.19] {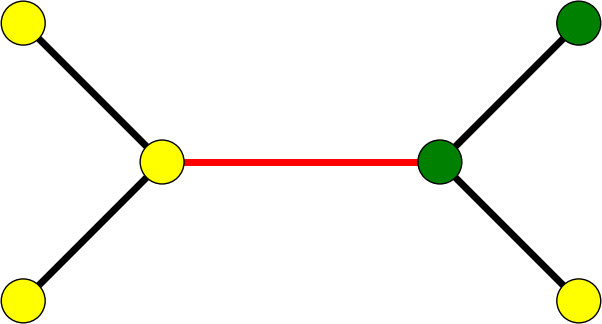} &
  \includegraphics[scale=0.19]
    {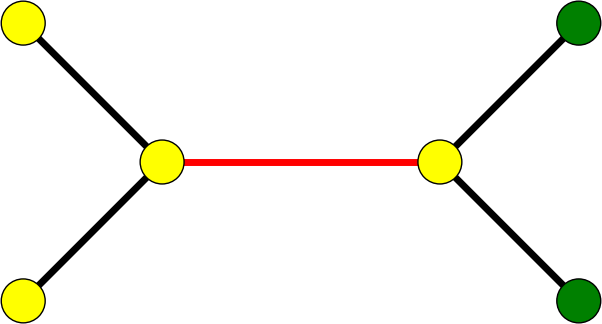} & \raisebox{0.9cm}{\includegraphics[scale=0.19] {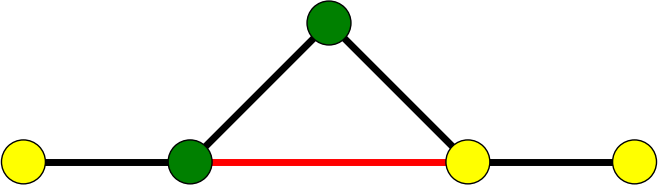}}\\
    $d_{H_1} = 2:$ & $d_{H_2} = 1:$ & $d_{H_3} = 2: $\\
    $\{2,3\}, \{5,6\}$ & $\{1,4\}$ & $\{1,2\}, \{1,6\}$
\end{tabular}

    \vspace{1cm}
    
    \begin{tabular}{ccc}
    $H_4:$ & $H_5:$ & $H_6:$\\
  \includegraphics[scale=0.19]
    {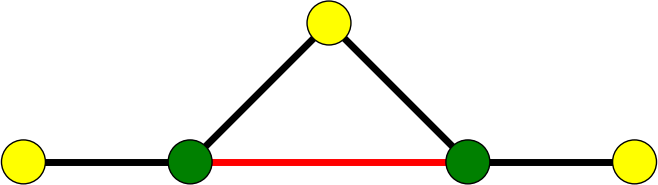} &
  \includegraphics[scale=0.19]
    {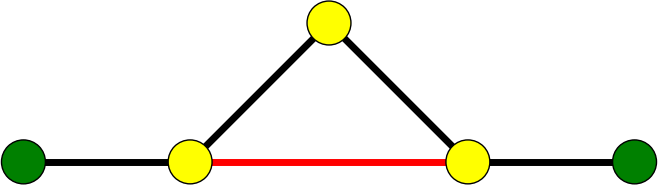} & \includegraphics[scale=0.19]
    {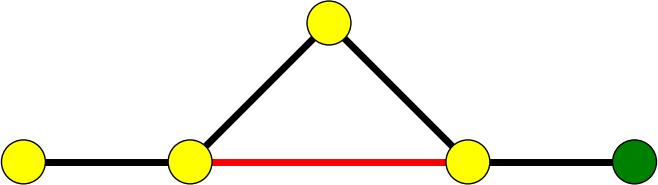} \\
    $d_{H_4} = 1:$ & $d_{H_5} = 1:$ & $d_{H_6} = 2:$\\
$\{2,6\}$ & $\{3,5\}$ & $\{3,4\}, \{4,5\}$
\end{tabular}
\caption{\footnotesize\label{fig:labeledGraphDecompExample} An example graph $G$ is depicted with a 2-coloring with corresponding bitstring $b=010001$ (the yellow and green vertex colors correspond to 0 and 1 respectively). It has 6 different colored edge neighborhood types, which we refer to as $H_1,\dots, H_6$ in this example. For each colored edge neighborhood type, the central edge is depicted in red. The degeneracies for each edge neighborhood type is also shown along with the edges associated with the degeneracy; note that the sum of the degeneracies adds up to the number of edges in $G$. In $G$, for the edge $e = \{1,2\}$, the depth-1 colored edge neighborhood $G_e$ is shown with solid black edges with the central edge $e$ in red; note that $G_e$ is of the same type as $H_3$. 
}
\end{figure}

In the case of 3-regular graphs, if we consider all the ways to 2-color the vertices of the non-colored edge neighborhoods $g_4,  g_5, g_6$ seen in Figure \ref{fig:cubicSubgraphTypes}, we would have a total of $2^4+2^5+2^6 = 112$ different colored subgraphs. Due to the symmetries in the underlying graphs of $g_4,g_5,g_6$, the actual number of types of colored edge neighborhoods (for 3-regular graphs) is smaller than 112; however, we can further simplify the number of cases to consider by making certain assumptions about the cut induced by the bitstring $b$ that is used for the warm-start $\ket{b_\theta}$. Intuitively, we expect that constructing warm-starts from ``good" cuts will cause QAOA to produce higher-quality cuts in expectation. In particular, for the rest of this paper, we restrict our attention to warm-starts $\ket{b_\theta}$ where $b$ corresponds to a cut in the graph that is locally optimal with respect to 1-BLS (Algorithm \ref{alg:localSearch}). Interestingly, by placing such a restriction on the cuts, this forbids certain types of colored edge neighborhoods that would have been considered otherwise; this is because such forbidden colored edge neighborhoods could be recolored in a way so that the cut resulting from the new coloring is guaranteed to increase the cut value for the rest of the graph, which is a contradiction on the local optimality of $b$. In Figure \ref{fig:ImprovementExample}, we show such an example of a forbidden colored edge-neighborhood in addition to the bit that needs to be flipped in the bitstring coloring in order to improve the overall cut value for the overall graph.

For 3-regular graphs, we have systematically found all possible non-forbidden depth-1 colored-edge neighborhoods in the case that the colored graph $(G,b)$ is such that $b$ is 1-BLS optimal as seen in the following theorem.
\begin{theorem}
    Let $(G,b)$ be a colored graph where $G$ is 3-regular and the bitstring $b$ has the property of being 1-BLS optimal with respect to the \MC{} problem on $G$. Then for any edge $e \in E(G)$, the depth-1 colored-edge neighborhood $G_e(b)$ must be of the same type as one of the 15 different colored-edge neighborhoods found in Figure \ref{fig:optimalLabeledSubgraphs}.
\end{theorem}
\begin{proof}
    By definition, the underlying graph itself for both the uncolored and colored edge-neighborhoods must be the same, and hence, as was already established in Farhi et al. \cite{farhi2014quantum}, the underlying graphs for the colored edge-neighborhoods must be of the same type as $g_4, g_5, g_6$ seen in \ref{fig:cubicSubgraphTypes}. One can then exhaustively consider all possible ways to color $g_4, g_5,$ and $g_6$ and eliminate any colorings that are inconsistent with $b$ being 1-BLS optimal; we describe this process in more detail in Section \ref{sec:subgraphConstructionDetails}. The result of this process is exactly the colored edge-neighborhoods found in Figure \ref{fig:optimalLabeledSubgraphs}.
\end{proof}

We remark that for every colored edge-neighborhood type in Figure \ref{fig:optimalLabeledSubgraphs}, there is a colored edge-neighborhood type that is exactly the same with the colors flipped, e.g., $[g_{4,1}]$ and $[g_{4,3}]$; another example is $[g_{4,2}]$ with itself. Amongst each pair of such edge-neighborhood types, QAOA will yield the same landscape of expected cut values for every choice of $\gamma$ and $\beta$ due to Theorem \ref{thm:bitflipSameExpectedCut}; while this can possibly help with the analysis and/or the computation of $F_G^{(p)}(\gamma, \beta, \ket{b}_\theta)$, it is important to distinguish between the edge-neighborhood types (and their degeneracies) within each pair as we will later see in Section \ref{sec:suitableConstraints}. 

\begin{figure}
    \centering


    \begin{tabular}{c|ccc}
    & & Colored Edge Neighborhood & Cut Value\\ \hline
    & $G_e(b):$ & \raisebox{-1.3cm}{\includegraphics[scale=0.2]{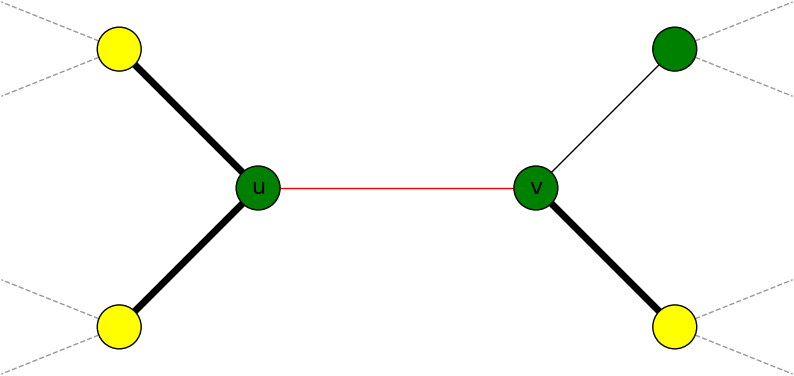}} & $\cut(b)=m'+3$ \\ \hline
    $\xrightarrow[\text{flip $v$}]{}$ & $G_e(b'):$ &  \raisebox{-1.3cm}{\includegraphics[scale=0.20]{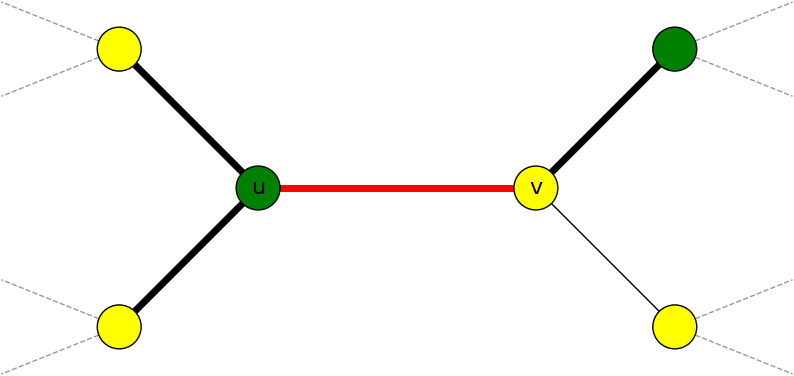}} & $\cut(b')=m'+4$
    \end{tabular}
    
    \caption{\footnotesize\label{fig:ImprovementExample} Consider a 3-regular graph $G$ with a coloring/cut determined by a bitstring $b$ and let $e = (u,v) \in E(G)$ and suppose the colored edge neighborhood $G_e(b)$ looks as depicted at the top of the figure. We show that in such a case, regardless of the colors of the remaining nodes in $G$, that it must be that the overall cut $b$ is not locally optimal with respect to single-bitflips (and hence any $g \sim G_e(b)$ is a forbidden type of colored edge neighborhood under the restriction that the colored edge neighborhood is generated from a locally optimal cut). In the colored edge neighborhood, the bold lines correspond to edges that would necessarily be cut in accordance with $b$. Given $b$, let $m'$ denote the number of edges that are cut outside the underlying graph of $G_e(b)$. Observe that we can recolor the vertices that are incident to the central edge without changing the value of $m'$. In particular, let $b'$ be $b$ but with the bit corresponding to vertex $v$ being flipped. Observe that causes an improvement not only in the ``local cut value" but the cut value of the entire graph.}
    
\end{figure}
    
\begin{figure}
    \centering
    \includegraphics[scale=0.08]{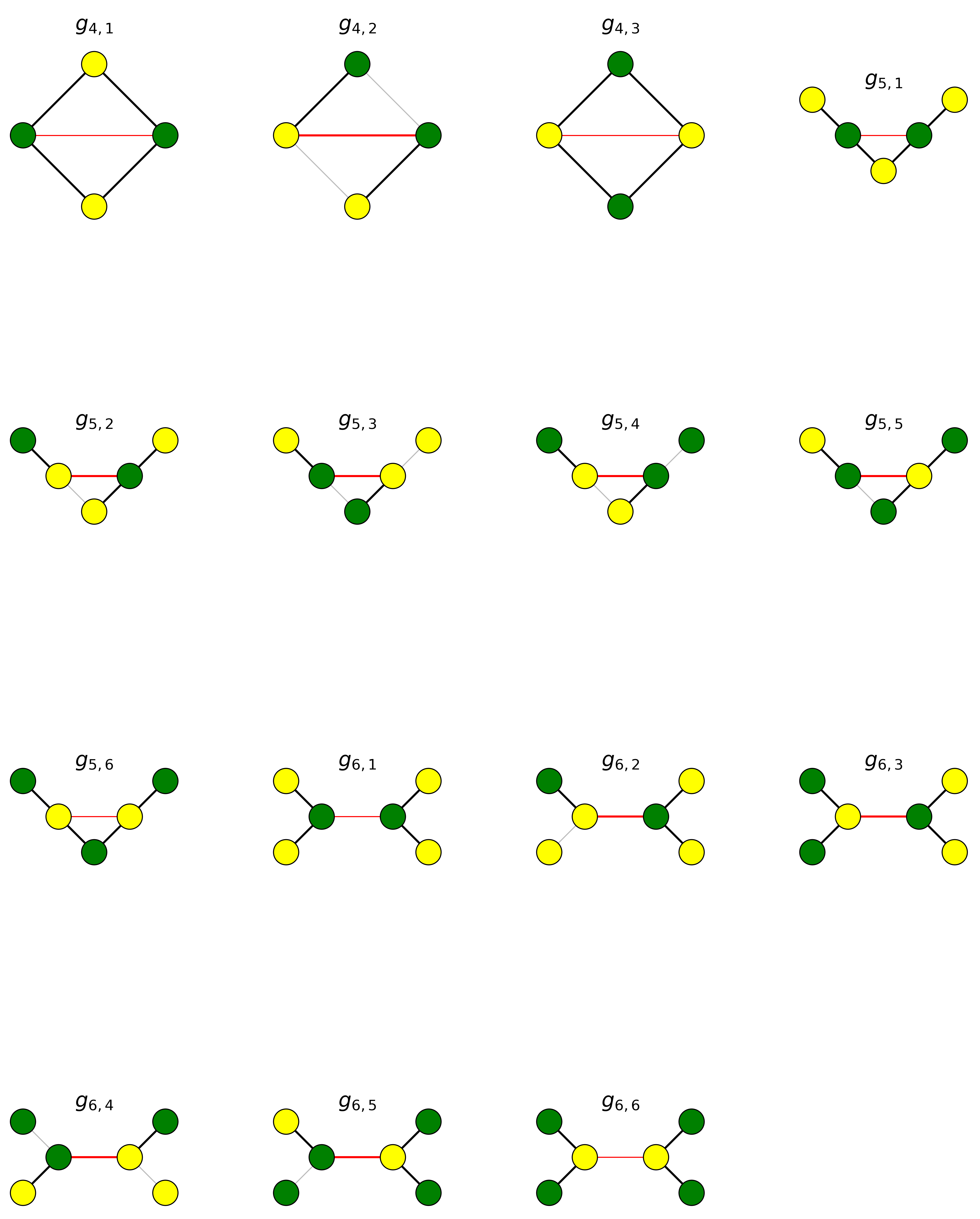}
    \caption{\footnotesize\label{fig:optimalLabeledSubgraphs} Listed are the 15 types of colored edge neighborhoods that can possibly exist for 3-regular graphs under the assumption that the cut in the graph induced by the bitstring/coloring is 1-BLS optimal. The yellow and green vertex colors correspond to 0 and 1 respectively and the red line denotes the central edge of each colored edge neighborhood. The bold edges correspond to edges that would be cut in the cut corresponding to the bitstring that determines the edge-neighborhood coloring (i.e. there is a bold edge if the colors of two adjacent vertices are different).}
\end{figure}

\subsection{Graph Structures for Depth-1 Max-Cut QAOA}
\label{sec:graphStructures}
If a triangle exists in a graph, then it follows that any cut of that graph will cut at most two of the edges of the triangle. In Farhi's original QAOA paper \cite{farhi2014quantum}, this insight motivates considering certain graph structures that contain triangles. These graph structures, which \cite{farhi2014quantum} refers to as ``isolated triangles" and ``crossed squares" are \emph{induced} subgraphs that can possibly exist in a 3-regular graph; a more detailed description and a diagram of these structures can be found in Figure \ref{fig:graphStructures}. Within both of these graph structures, there is a subset of edges which we will refer to as the \emph{core} edges; these are denoted by solid lines in Figure \ref{fig:graphStructures}. We also define the core vertices of each graph structure as the vertices which are incident to at least one core edge. Farhi et al. \cite{farhi2014quantum} showed that by counting these graph structures, one can show that depth-1 QAOA has an approximation ratio of 0.6924 on 3-regular graphs. We briefly review Farhi et al.'s proof of the 0.6924-approximation as we will later adapt the proof technique used to arrive at an approximation ratio for warm-started QAOA in Sections \ref{sec:coloredGraphStructures} and \ref{sec:suitableConstraints}.

\begin{figure}
    \centering
    \begin{tabular}{ccc}
    \textbf{S}:\hspace{0.5cm} \raisebox{-0.5\height}{\rotatebox{90}{\includegraphics[scale=0.2]{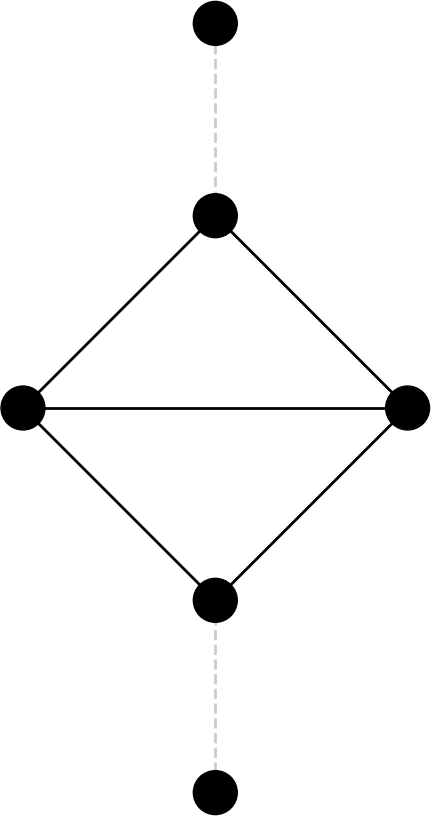}}}
    & \hspace{2cm} & \textbf{T}: \raisebox{-0.5\height}{\includegraphics[scale=0.2]{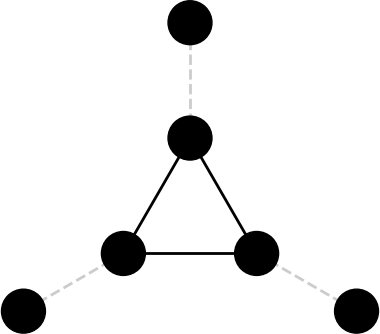}}
    \end{tabular}
    \caption{\label{fig:graphStructures} \footnotesize Above are the two graph structures (the crossed squares \textbf{S} and the isolated triangles \textbf{T}) of interest for depth-1 QAOA on 3-regular graphs. We say that a graph contains one of these structures if it contains an \emph{induced} subgraph that is isomorphic to the graph structure. As discussed in the main text, the solid edges and the vertices incident to them are referred to as the core edges and core vertices respectively.}
\end{figure}

The core edges in these graph structures can help one to count the degeneracies $d_{g_4}, d_{g_5}, d_{g_6}$ of the depth-1 edge neighborhood types for 3-regular graphs. Looking at $g_4$ in Figure \ref{fig:cubicSubgraphTypes}, recall that the endpoints of the dashed edges do not necessarily go to distinct vertices. Depending on how these dashed edges connect to possibly the same vertex, by inspection, it can be seen that it must be the case that if $G_e \sim g_4$ for some edge $e$ in a 3-regular graph $G$, then $e$ must also lie amongst one of the core edges of some isolated triangle or some crossed square in $G$. Moreover, since the core edges of all isolated triangles and crossed squares are disjoint from one another, then $e$ is contained in the core edge set of \emph{exactly} one of these graph structures. Similarly, it can be shown that if $G_e \sim g_5$, then $e$ must appear in exactly one of the crossed squares of $G$. We will use $T$ and $S$ to denote the number of isolated triangles and crossed squares in $G$ respectively.

Next, observe that in an isolated triangle, the edge neighborhood type of all 3 of the core edges correspond to $g_5$. For the crossed square, 4 of the core edges have edge neighborhood type of $g_5$ and 1 of the core edges has edge neighborhood type of $g_4$. Combining this with the previous observations, it can be seen that:
\begin{equation}
\label{eqn:degenAsFuncOfGraphStructs}
d_{g_4} = S, \quad d_{g_5} = 4S+3T.
\end{equation}

For convenience, let $R = d_{g_6}$, then, recalling Equation \ref{eqn:subgraphDecomposition}, one can calculate\footnote{In Farhi et al.'s original analysis, they use that $R = d_{g_6} = |E(G)| - d_{g_4} - d_{g_5} = \frac{3n}{2} - S - (4S+3T) = \frac{3n}{2} - 5S - 3T$ to eliminate $R$ in their analysis. We use the formulation with $R$ explicit in order to smoothen the transition to the formulation of the warm-started case.} the expected cut-value of depth-1 QAOA on a 3-regular graph $G$ as
\begin{equation}\label{eqn:expectedCutAsFunctionOfGraphStructures} F_G(\gamma,\beta) = F_{R,S,T}(\gamma,\beta) := S f_{g_4}(\gamma,\beta) + (4S+3T) f_{g_5}(\gamma,\beta) + R f_{g_6}(\gamma, \beta).\end{equation}

For a 3-regular graph $G$, the expected cut value of QAOA with optimized variational parameters can be expressed solely as a function of $R,S,T$:
\begin{equation}
\label{eqn:mInTermsOfRST}
    M(R,S,T) = \max_{\gamma,\beta} F_{R,S,T}(\gamma,\beta).
\end{equation}

Observe that the core edge set of every isolated triangle and crossed square contains at least one triangle and moreover, such triangles between different graph structures must be vertex-disjoint. Thus, for a 3-regular graph $G$ with $m$ edges, we have that for any cut, at least $S+T$ edges in $G$ must remain uncut, and hence $\mc(G) \leq m - S - T$. Thus, for a graph $G$ characterized by $m,R,S,T$, we have the instance-specific approximation ratio of $G$ is,
$$\alpha_\text{QAOA}(G) = \frac{M(R,S,T)}{\mc(G)} \geq \frac{M(R,S,T)}{m-S-T}.$$

For convenience, we scale out $R,S,T$ by $m$ to obtain $r := R/m, s := S/m, t: T/m$. In Farhi et al.'s original QAOA paper, the variables are instead scaled out $n$; for this work, we choose to scale out by $m$ for reasons that will become clear in Section \ref{sec:coloredGraphStructures} when we consider these graph structures in the context of warm-starts. With this rescaling, we have that,
$$\alpha_\text{QAOA}(G) \geq \frac{M(r,s,t)}{1-s-t}.$$

For any 3-regular graph, there are certain constraints on $r,s,t$ that always hold. Clearly $r,s,t \geq 0$. In addition, since the core vertices of the every isolated triangle and crossed square are disjoint from one another, then $4S+3T \leq n = \frac{2}{3}m$ which implies that $4s+3t \leq \frac{2}{3}.$ Additionally, $r,s,t$ should be consistent with the total number of edges in the graph\footnote{As mentioned in the previous footnote, Farhi et al. \cite{farhi2014quantum} do not have a variable $r$ or $R$ in their analysis as it can be eliminated via a substitution of the other variables. We instead present the minimization in a way where this potential substitution is instead encoded in the constraints as $5s+3t+r=1$.}, i.e., $d_{g_4}+d_{g_5}+d_{g_6} = S+(4S+3T)+R = m$ implies that $5s+3t+r=1$.

Thus, the approximation ratio of depth-1 standard QAOA  on 3-regular graphs can be bounded below as:
$$\alpha_{\text{QAOA}} \geq \min_{(r,s,t) \in \mathcal{P}} \frac{M(r,s,t)}{1-s-t},$$

where the constraint polytope $\mathcal{P}$ is defined as,
$$\mathcal{P} = \{(r,s,t) \in \mathbb{R}^3: r,s,t \geq 0, 5s+3t+r=1, 4s+3t \leq 2/3\}.$$

In the above minimization, Farhi et al. \cite{farhi2014quantum} found (numerically) that the minimum value is achieved at $s=t=0$ and $r=1$ and that the minimum value is $0.6924$, thus $\alpha_\text{QAOA} \geq 0.6924$. Moreover, as seen in Proposition \ref{thm:QAOATight}, this bound is in fact tight.

\begin{prop}
\label{thm:QAOATight}
    There exists infinitely many 3-regular graphs $G$, such that $\alpha_\text{QAOA}(G) = 0.6924$.
\end{prop}
\begin{proof}
    By carefully going through the analysis above, it becomes clear that equality is achieved if there exists a graph $G$ for which $s=t=0, r=1$ and where $\mc(G) = m - S - T = m$. This is clearly achieved for any 3-regular bipartite graph since such a graph has no triangles and all edges are cut by the max-cut (effectively) by definition. There exists an infinite number of $3$-regular bipartite graphs, which we show the construction for in Figure \ref{fig:cubicBipartiteGraph}.
\end{proof}

We now take a closer look at the optimization problem that bounds the approximation ratio. Observe:
$$\alpha_{\text{QAOA}} = \min_{(r,s,t) \in \mathcal{P}} \frac{M(r,s,t)}{1-s-t} = \min_{(r,s,t) \in \mathcal{P}}\max_{\gamma,\beta} \frac{F_{r,s,t}(\gamma,\beta)}{1-s-t} \geq \max_{\gamma,\beta}\min_{(r,s,t) \in \mathcal{P}} \frac{F_{r,s,t}(\gamma,\beta)}{1-s-t},$$
in other words, because of the max-min inequality (details in Section \ref{sec:optimizationDetails}), it appears that swapping the order the maximization and minimization produces worse bounds for the approximation ratio. Interestingly, for QAOA, numerical evidence suggests that such swapping has no effect, i.e., 
$$\min_{(r,s,t) \in \mathcal{P}}\max_{\gamma,\beta} \frac{F_{r,s,t}(\gamma,\beta)}{1-s-t} = \max_{\gamma,\beta}\min_{(r,s,t) \in \mathcal{P}} \frac{F_{r,s,t}(\gamma,\beta)}{1-s-t},$$
and thus there is no obvious downside to performing the optimization $\max_{\gamma,\beta}\min_{(r,s,t) \in \mathcal{P}}\frac{F_{r,s,t}(\gamma,\beta)}{1-s-t}$ instead. In fact, this version of the optimization has certain advantages:
\begin{enumerate}
    \item It is ``easier" to numerically solve as the inner optimization can be solved exactly using linear programming techniques,
    \item It is possible to obtain numerical objective values that are themselves lower bounds of the true objective value, meaning we cannot accidentally overestimate the approximation ratio,
    \item The optimal $\gamma^*,\beta^*$ from the optimization can be used for \emph{any} instance, i.e., for any instance, one can be run QAOA with parameters $(\gamma,\beta) = (\gamma^*, \beta^*)$ and obtain a $0.6924$-approximation ratio without having to do any further optimization of $(\gamma,\beta)$.
\end{enumerate}

For the optimization problem $\max_{\gamma,\beta}\min_{(r,s,t) \in \mathcal{P}}$, the optimal values of $\gamma^*,\beta^*$ are as follows:
$$\gamma^* = 0.616$$
$$\beta^* = 0.393,$$
which is consistent with the values used in the work of Wurtz and Love \cite{wurtz2021maxcut}.

\sloppy
Details regarding the numerical optimization of $\min_{(r,s,t) \in \mathcal{P}} \max_{\gamma,\beta}\frac{F_{r,s,t}(\gamma,\beta)}{1-s-t}$ and $\max_{\gamma,\beta}\min_{(r,s,t) \in \mathcal{P}} \frac{F_{r,s,t}(\gamma,\beta)}{1-s-t}$ and the consequences of swapping the maximization and minimization (including the list of advantages above) can be found in Section \ref{sec:optimizationDetails}.

\begin{figure}
    \centering
    \raisebox{-0.5\height}{\includegraphics[scale=.3]{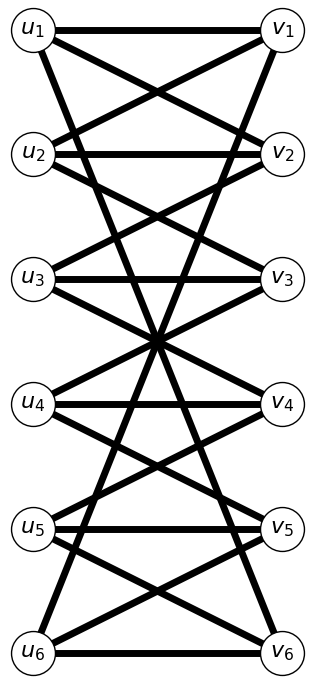}} \hspace{1.5cm} $=$ \hspace{1.5cm} \raisebox{-0.5\height}{\includegraphics[scale=.3]{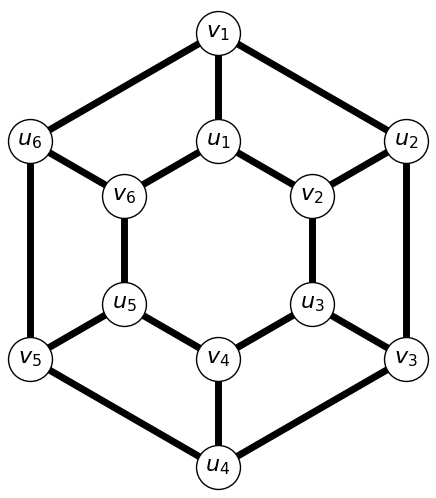}}
    \caption{\footnotesize\label{fig:cubicBipartiteGraph} On the left, is an example of a 12-node 3-regular bipartite graph. In general, for any integer $k \geq 3$, we can construct a $3$-regular bipartite graph $G=(V,E)$ on $n=2k$ nodes as follows. Label the vertices as $V = \{u_1,u_2,\dots, u_{k},v_1, v_2, \dots, v_k\}.$. The edge $e$ is included in $E$ if and only if $e$ is of the form $\{u_i, v_j\}$ where $|i - j|\leq 1$ or if $e \in \{ \{u_1, v_k\}, \{u_k, v_1\}\}.$ Additionally, in the case that $n$ is divisible by 4, we can redraw such a graph in a way that will be convenient later in this work when we consider warm-started QAOA. In such a case, observe that $G$ has two large cycles with vertex sequences $(u_1, v_2, u_3, v_4, \dots,  u_{k-1}, v_{k}, u_1)$ and $(v_1, u_2, v_3, u_4, \dots, v_{k-1}, u_{k}, v_1)$ with the remaining edges (outside the cycles) being of the form $\{u_i, v_i\}$ for some $i\in [k]$. As an example, this insight can be used to redraw the 12-node 3-regular bipartite graph as seen on the right side of the figure.
    }
\end{figure}

\newpage

\subsection{Colored Graph Structures for Depth-1 Max-Cut QAOA}
\label{sec:coloredGraphStructures}
We now consider the graph structures (e.g. isolated triangles and crossed squares) from the previous subsection in the context of colored graphs where the coloring encodes a warm-start $\ket{b_\theta}$ for some bitstring $b$. We can now look at \emph{colored} versions of these graph structures, i.e., colored isolated triangles and colored crossed squares. If the warm-start $\ket{b_\theta}$ used has no restriction on the cut corresponding to the bitstring $b$ that is used, then there are numerous ways to color such graph structures. However, if we place certain restrictions on $b$, e.g., requiring that $b$ correspond to a locally optimal 1-BLS cut, then the possible colorings for these graph structures is limited, similar to how the possible colored edge neighborhoods were limited as seen in Section \ref{sec:coloredEdgeNeighborhoods}. Details regarding how these colored graph structures were found can be found in Section \ref{sec:subgraphConstructionDetails}.

In Figure \ref{fig:coloredGraphStructures}, we list all 11 possible colored graph structures under the restriction that the initial warm-start state $\ket{b_\theta}$ has $b$ being locally optimal with respect to 1-BLS. The colored isolated triangles are labeled as $\mathbf{T_0, T_1, T_2, T_3}$ and the colored crossed squares are labeled as $\mathbf{S_0, S_0', S_{1,1}, S_{1,2}, S_{1,3}, S_2, S_2'}$. The first subscript denotes the number of non-core nodes in the graph structure that are colored green, i.e., the bitstring describing the coloring/cut is 1 for such nodes; this is not always enough to uniquely determine the colored graph structure so additional primes or subscripts are used in the notation when necessary. If we instead require that $b$ be at least locally optimal with 4-BLS, then the graph structures $S'_0$ and $S'_2$ do not need to be considered. For convenience, we will unbold the name of each colored graph structure in order to refer to the number of such structures that exist within the graph $G$, e.g., $S_0$ is defined as the number of times $\mathbf{S_0}$ appears in $G$. Additionally, we group the counts of these colored graph structures and the degeneracies of the $g_{6,i}$'s as follows:

$$\mathbf{R} = (R_1, \dots, R_6),$$
$$\mathbf{S} = (S_0, S_0', S_{1,1}, S_{1,2}, S_{1,3}, S_2, S_2'),$$
$$\mathbf{T} = (T_0, T_1, T_2, T_3),$$
and we use the unbolded $R,S,T$ to denote the sum of the entries in $\mathbf{R}, \mathbf{S}, \mathbf{T}$ respectively.

For each core edge in each colored graph structure in Figure \ref{fig:coloredGraphStructures}, we determine its (depth-1) colored edge neighborhood type; this information is included in the figure. Similar to Equation \ref{eqn:degenAsFuncOfGraphStructs}, we can determine the degeneracies of certain colored subgraphs in terms of the counts of certain colored graph structures. For example, using the information in Figure \ref{fig:coloredGraphStructures}, we find that,
$$d_{g_{5,5}} = 2T_1+4S_2+2S_{1,1}.$$

In general, the degeneracies of $g_{4,1},\dots, g_{4,3}, g_{5,1}, \dots, g_{5,6}$ can be calculated as a function of the counts of these colored graph structures. For convenience, we define $R_i = d_{g_{6,i}}$ for the remaining colored edge neighborhoods on six nodes. Using the subgraph decomposition formula in Equation \ref{eqn:coloredSubgraphDecomposition}, we can write the expected cut value of warm-started QAOA as a function of colored graph structure counts and these $R_i$ variables:
\begin{nalign}
\label{eqn:expectedCutValueInTermsOfColoredGraphStructures}
    F_G(\gamma, \beta, \ket{b_\theta})  = & F_{\mathbf{R},\mathbf{S},\mathbf{T}}(\gamma,\beta,\theta)\\
    := & (S_0'+S_2+S_{1,1})f_{g_{4,1}}(\gamma,\beta,\theta)\\
    + & (S_{1,2})f_{g_{4,2}}(\gamma,\beta,\theta)\\
    + & (S_0+S_{1,3}+S_2')f_{g_{4,3}}(\gamma,\beta,\theta)\\
    + & (T_0+T_1+S_{1,2})f_{g_{5,1}}(\gamma,\beta,\theta)\\
    + & (2T_2+4S_0+2S_{1,3})f_{g_{5,2}}(\gamma,\beta,\theta)\\
    + & (2T_0+4S_0'+2S_{1,1}+S_{1,2})f_{g_{5,3}}(\gamma,\beta,\theta)\\
    + & (2T_3+S_{1,2}+2S_{1,3}+4S_2')f_{g_{5,4}}(\gamma,\beta,\theta)\\
    + & (2T_1+2S_{1,1}+4S_2)f_{g_{5,5}}(\gamma,\beta,\theta)\\
    + & (T_2+T_3+S_{1,2})f_{g_{5,6}}(\gamma,\beta,\theta)\\
    + & R_1 f_{g_{6,1}}(\gamma,\beta,\theta)\\
    + & R_2 f_{g_{6,2}}(\gamma,\beta,\theta)\\
    + & R_3 f_{g_{6,3}}(\gamma,\beta,\theta)\\
    + & R_4 f_{g_{6,4}}(\gamma,\beta,\theta)\\
    + & R_5 f_{g_{6,5}}(\gamma,\beta,\theta)\\
    + & R_6 f_{g_{6,6}}(\gamma,\beta,\theta).
\end{nalign}

\begin{figure}
    \centering
    \centerline{ 
    \begin{tabular}{ccc}
    \begin{tabular}{c}
    $\mathbf{S_0}$\\\\
    Neighborhoods: \\
    $g_{4,3}, 4g_{5,2}$
    \end{tabular}
    \raisebox{-0.5\height}{{\includegraphics[scale=0.2]{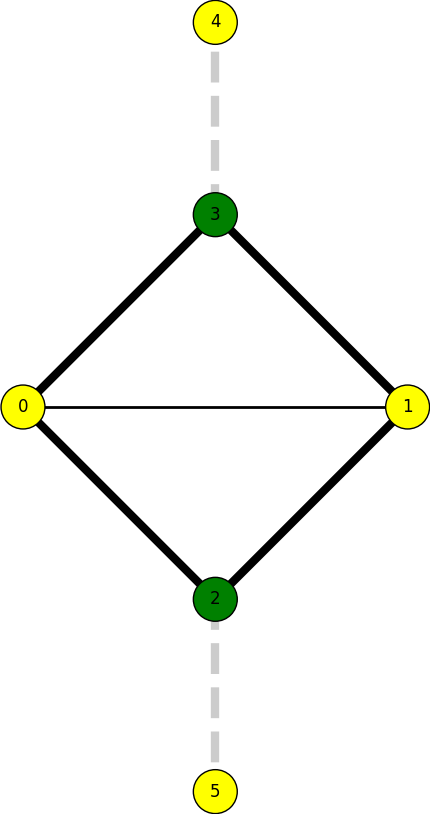}}} & 
    
    \begin{tabular}{c}
    $\mathbf{S_0'}$\\\\
    Neighborhoods: \\
    $g_{4,1}, 4g_{5,3}$
    \end{tabular}\raisebox{-0.5\height}{{\includegraphics[scale=0.2]{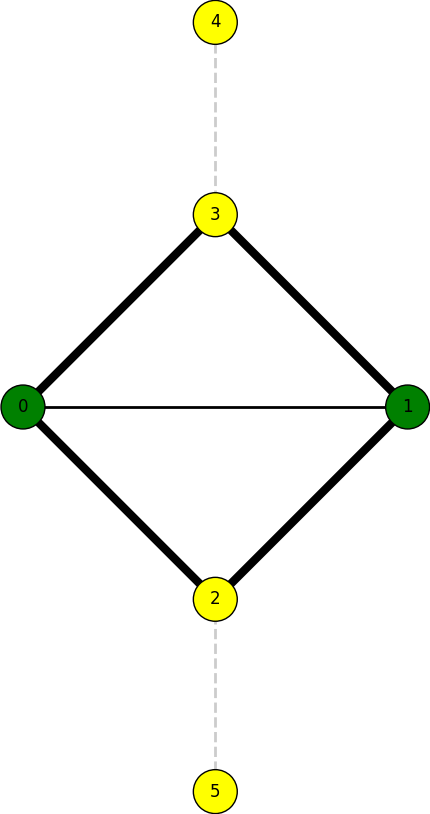}}} & 
    
    \begin{tabular}{c}
    $\mathbf{S_{1,1}}$\\\\
    Neighborhoods: \\
    $g_{4,1}, 2g_{5,3}, 2g_{5,5}$
    \end{tabular}\raisebox{-0.5\height}{{\includegraphics[scale=0.2]{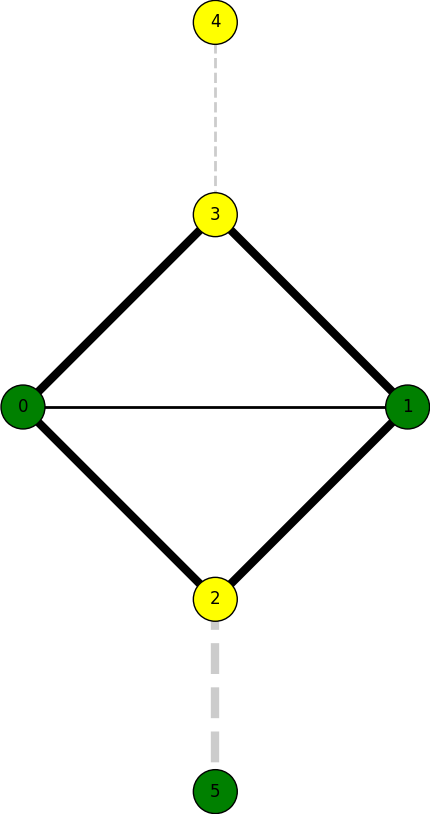}}} \\\\

    \begin{tabular}{c}
    $\mathbf{S_{1,2}}$\\\\
    Neighborhoods: \\
    $g_{4,2},g_{5,1}, g_{5,3},$\\
    $g_{5,4}, g_{5,6}$
    \end{tabular}\raisebox{-0.5\height}{{\includegraphics[scale=0.2]{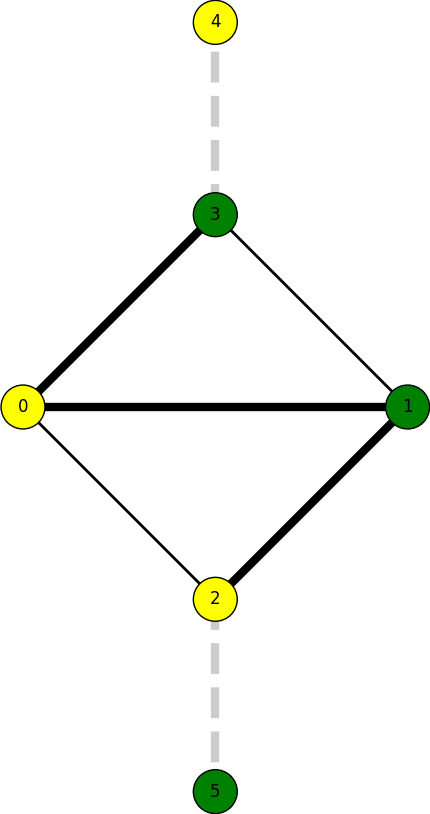}}} &

    \begin{tabular}{c}
    $\mathbf{S_{1,3}}$\\\\
    Neighborhoods: \\
    $g_{4,3}, 2g_{5,2}, 2g_{5,4}$
    \end{tabular}\raisebox{-0.5\height}{{\includegraphics[scale=0.2]{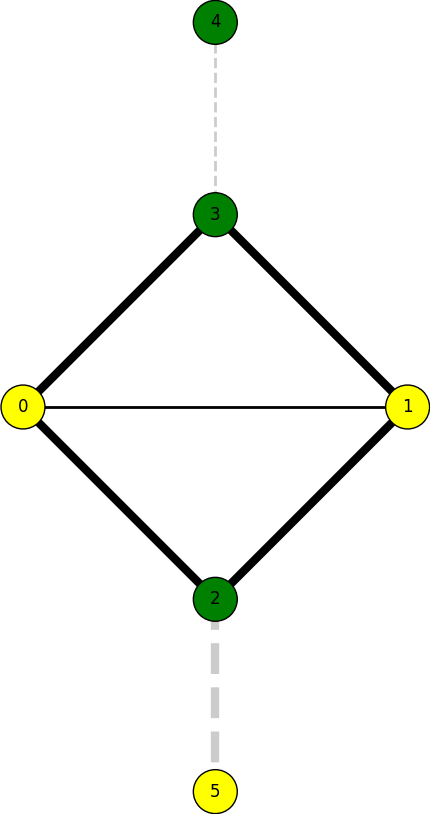}}} &

    \begin{tabular}{c}
    $\mathbf{S_2}$\\\\
    Neighborhoods: \\
    $g_{4,1}, 4g_{5,5}$
    \end{tabular}\raisebox{-0.5\height}{{\includegraphics[scale=0.2]{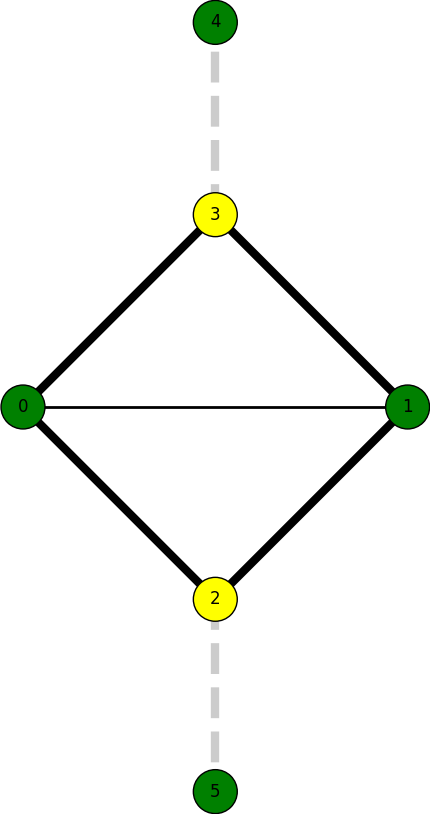}}} \\

    \begin{tabular}{c}
    $\mathbf{S_2'}$\\\\
    Neighborhoods: \\
    $g_{4,3}, 4g_{5,4}$
    \end{tabular}\raisebox{-0.5\height}{{\includegraphics[scale=0.2]{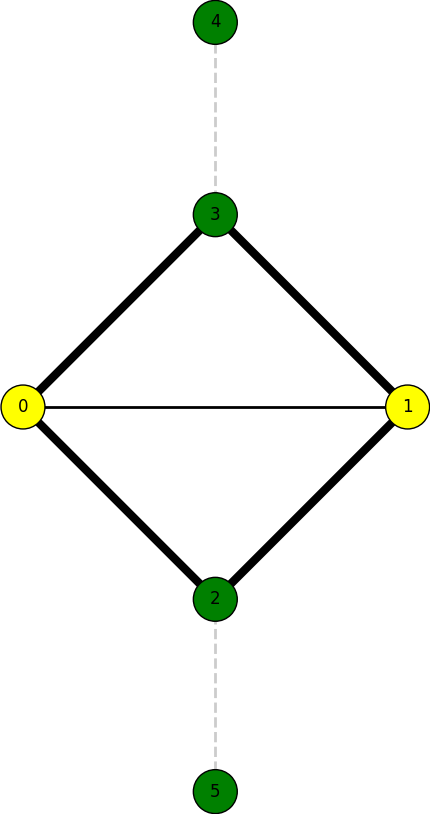}}} &

    \begin{tabular}{c}
    $\mathbf{T_0}$\\\\
    Neighborhoods: \\
    $g_{5,1}, 2g_{5,3}$
    \end{tabular}\raisebox{-0.5\height}{{\includegraphics[scale=0.2]{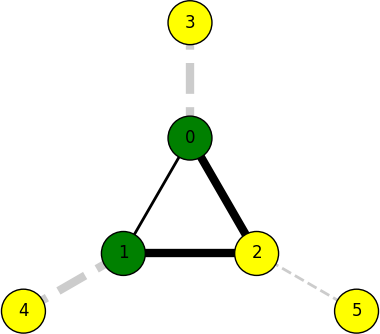}}} &

    \begin{tabular}{c}
    $\mathbf{T_1}$\\\\
    Neighborhoods: \\
    $g_{5,1}, 2g_{5,5}$
    \end{tabular}\raisebox{-0.5\height}{{\includegraphics[scale=0.2]{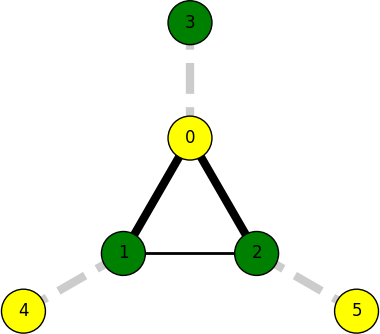}}} \\

    \begin{tabular}{c}
    $\mathbf{T_2}$\\\\
    Neighborhoods: \\
    $g_{5,6}, 2g_{5,2}$
    \end{tabular}\raisebox{-0.5\height}{{\includegraphics[scale=0.2]{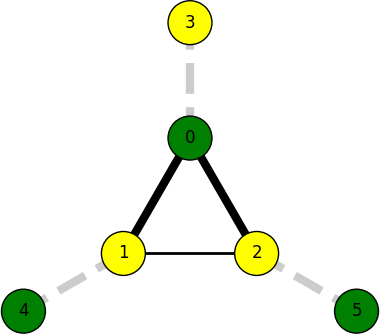}}} &

    \begin{tabular}{c}
    $\mathbf{T_3}$\\\\
    Neighborhoods: \\
    $g_{5,6}, 2g_{5,4}$
    \end{tabular}\raisebox{-0.5\height}{{\includegraphics[scale=0.2]{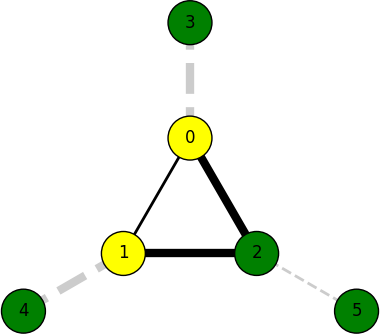}}} &

    \end{tabular}
    }

    \caption{\footnotesize\label{fig:coloredGraphStructures} An exhaustive list of possible ways to color an isolated triangle or crossed square under the assumption that the graph is colored using a bitstring that is 1-BLS optimal for Max-Cut. The thicker edges in each graph correspond to edges that would be cut in accordance with the bitstring coloring. The core edges of each graph structure are colored black with the remaining edges colored gray. Amongst the core edges of each graph strucutre, we list the count of each depth-1 colored edge-neighborhood type about those edges (see Figure \ref{fig:optimalLabeledSubgraphs}). }
\end{figure}

Similar to Equation \ref{eqn:mInTermsOfRST}, for each fixed value of $\theta$, we can write the expected cut value of warm-started QAOA (initialized with $\ket{b_\theta}$), with optimal variational parameters as a function of the counts of colored graph structures and the degeneracies of the $g_{6,i}$'s:
$$M_\theta(\mathbf{R},\mathbf{S},\mathbf{T}) = \max_{\gamma,\beta} F_{\mathbf{R},\mathbf{S},\mathbf{T}}(\gamma,\beta,\theta).$$

We can do as in Section \ref{sec:graphStructures} and scale each variable by $m$ to get $s_0 := S_0/m, s_0' := S_0'/m, s_{1,1} := S_{1,1}/m$, etc and also define the vectors $\mathbf{r} = \mathbf{R}/m,\mathbf{s} = \mathbf{S}/m, \mathbf{t} = \mathbf{T}/m$,  and the sums $r = R/m, s = S/m, t = T/m$.
Using the same arguments as in Section \ref{sec:graphStructures}, for any fixed $\theta$ and 3-regular graph $G$, we can bound the approximation ratio of depth-1 warm-started QAOA with initial state of the form $\ket{b_\theta}$ where $b$ is 1-BLS locally optimal as follows:
\begin{equation}\label{eqn:arLowerBoundMinimization}\alpha_{{(\text{1-BLS}+\text{QAOA}})_\theta} \geq \alpha'_{{(\text{1-BLS}+\text{QAOA}})_\theta} \geq  \min_{(\mathbf{r},\mathbf{s},\mathbf{t}) \in \mathcal{Q}} \frac{M_\theta(\mathbf{r},\mathbf{s},\mathbf{t})}{1-s - t}.\end{equation}

where the constraint polytope $\mathcal{Q}$ is defined as

$$\mathcal{Q} = \{(\mathbf{r}, \mathbf{s}, \mathbf{t}) \in \mathbb{R}^{17}: \mathbf{r}, \mathbf{s}, \mathbf{t} \geq \mathbf{0}, 5s+3t+r=1, 4s+3t \leq 2/3\}.$$

Once again, just like in Section \ref{sec:graphStructures}, we have a nested optimization with an outer minimization over the graph structures and an inner maximization over the variational parameters. We once again consider swapping the order of the nested optimization, obtaining that,
$$\min_{(\mathbf{r},\mathbf{s},\mathbf{t}) \in \mathcal{Q}} \frac{M_\theta(\mathbf{r},\mathbf{s},\mathbf{t})}{1-s - t} \geq \mathbf{LB}(\theta):= \max_{\gamma,\beta}\min_{(\mathbf{r}, \mathbf{s},\mathbf{t}) \in \mathcal{Q}}\frac{F_{\mathbf{r},\mathbf{s},\mathbf{t}}(\gamma,\beta, \theta)}{1-s-t}.$$

In section \ref{sec:graphStructures}, it was mentioned that swapping the order of the maximization and minimization has no effect for standard QAOA. Similarly, numerical evidence suggests swapping has little effect for warm-started QAOA. Swapping the order of the nested optimization yields the same advantages as discussed in \ref{sec:graphStructures}; for this reason, we numerically optimize $\mathbf{LB}(\theta)$ in order to obtain a lower bound of warm-started QAOA's approximation ratio.

In addition, it will be useful to define \begin{equation}\label{eqn:depthZeroARLowerBound}\mathbf{LB}_\kappa^{(0)}(\theta) = \frac{1}{4}((2\kappa-1)\cos(2\theta)+2\kappa+1),\end{equation} which denotes a lower bound on the approximation ratio achieved by simply measuring a warm-start of the form $\ket{b_\theta}$ where $b$ is obtained by an algorithm that returns a cut containing at least an $\kappa$ fraction of the number of edges in the graph (Corollary \ref{thm:depth0ApproxRatio2}).

If we solve $\mathbf{LB}(\theta)$ for various values of fixed $\theta$, we obtain a plot of (lower bounds on) the approximation ratio as a function of $\theta$ as seen in Figure \ref{fig:lowerBoundWithoutAdditionalConstraints}. The details of how this minimization was numerically solved can be found in Section \ref{sec:optimizationDetails}. As $\theta \to 90^\circ$, we see that the bound on the approximation ratio approaches $0.6924$; this is expected since at $\theta = 90^\circ$, warm-started QAOA is equivalent to the standard QAOA. Since 1-BLS is a $2/3$-approximation algorithm for 3-regular graphs, one might expect that as $\theta \to 0^\circ$, that the bound on the approximation ratio approaches $2/3$; however, it instead approaches 0. This is because the bound on the approximation ratio is not tight enough; there are additional constraints we can add to the minimization (that hold for all 3-regular graphs and warm-starts constructed with 1-BLS cuts) that can be added to obtain a tighter bound. The addition of such constraints is explored in the next section.

\begin{figure}
    \centering \label{fig:lowerBoundWithoutAdditionalConstraints}
    \includegraphics[scale=0.75]{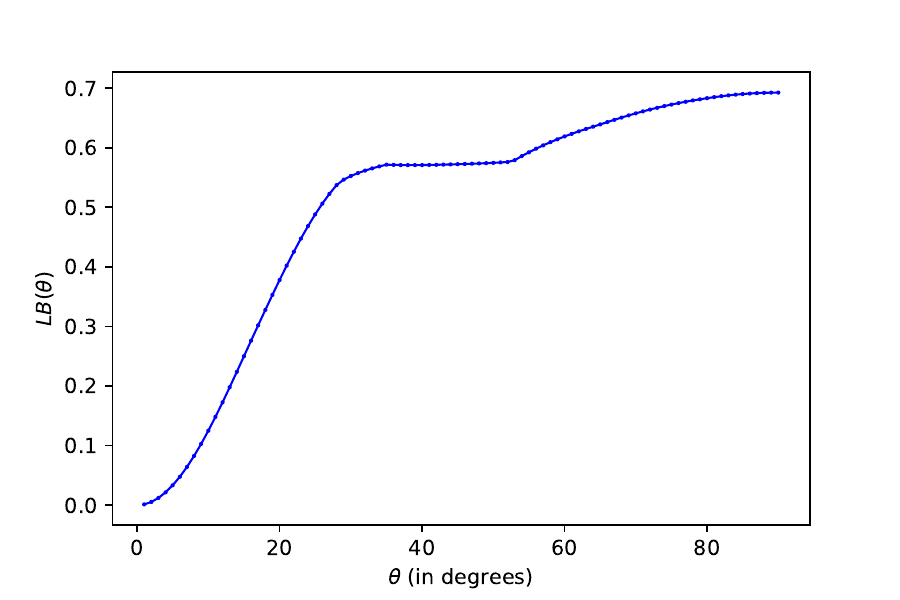}
    \caption{\footnotesize A plot of $\mathbf{LB}(\theta)$, a very loose lower bound on the approximation ratio of depth-1 \mc{} warm-stated QAOA initialized using cuts that are 1-BLS optimal, with varying $\theta$.}
\end{figure}

\subsection{Suitable Constraints for Warm-Start Optimization}
\label{sec:suitableConstraints}

Numerical experiments suggest that for small $\theta$, that the optimal solution to the minimization defined by $\textbf{LB}(\theta)$ has an optimal solution where all variables are equal to 0 with the exception of $r_1 = d_{g_{6,1}}$ and $r_6 = d_{g_{6,6}}$. At $\theta = 0$, warm-started QAOA with initial state $\ket{b_\theta}$ will only return the bitstring b with probability 1. In such a case, an edge is cut if and only if the central edge of its corresponding colored edge neighborhood are colored differently. Looking at $g_{6,1}$ and $g_{6,6}$ in Figure \ref{fig:optimalLabeledSubgraphs}, the central edge is \emph{never} cut, and thus, if the only colored edge neighborhoods types are $[g_{6,1}]$ and $[g_{6,6}]$, then the number of edges cut by $b$ is zero; this is a contradiction if we assume that $b$ is 1-BLS. In such a case, this means that the optimal solution returned by the minimization corresponds to an impossible distribution of colored edge neighborhoods types.

Moreover, at small $\theta$, we see that both $f_{g_{6,1}}(\gamma,\beta,\theta)$ and $f_{g_{6,1}}(\gamma,\beta,\theta)$, the probability of the central edge of $g_{6,1}$ and $g_{6,6}$ (respectively) being cut at depth-1 warm-started QAOA, is very small for any choice of $\gamma$ and $\beta$; see Figure \ref{fig:treeSubgraphLandscapes}. In general, as seen in Proposition \ref{thm:flatteningLandscape}, as $\theta$ approaches zero, the landscape of central-edge cut probabilities will either flatten to 0 or 1, depending on how the vertices of the central edge are colored.

\begin{prop}
\label{thm:flatteningLandscape}
Let $(G,b)$ be any colored graph ($G$ is not necessarily 3-regular and $b$ is not necessarily 1-BLS). Pick any colored-edge neighborhood $G_e$ for some $e = (u,v) \in E(G)$. Then, for all $\gamma,\beta \in \mathbb{R}$ and for all circuit depths $p \geq 0$, we have that,
$$\lim_{\theta \to 0} f^{(p)}_{G_e}(\gamma,\beta,\theta) = \begin{cases}
    0, & b_u = b_v\\
    1, & b_u \neq b_v
\end{cases}.$$
\end{prop}
\begin{proof}
    Suppose $b_u = b_v$. Then,
    $$\lim_{\theta \to 0} f^{(p)}_{G_e}(\gamma,\beta,\theta) = f^{(p)}_{G_e}(\gamma,\beta,0) = \bra{b_\theta}C_e\ket{b_\theta} = \bra{b}C_e\ket{b} = 0.$$
    The first equality above is due to continuity of $f$ as a function of $\theta$. The second equality is due to the fact that at $\theta = 0$, the cost Hamiltonian and mixing Hamiltonian only contain Pauli-$Z$ terms, and thus, both operators commute past the observable $C_e$ and cancel out. The third equality holds as $\ket{b_\theta} = \ket{b}$ at $\theta=0$. The last equality holds as $\bra{b}C_e\ket{b}$ is simply the value of the observable $C_e$ with bitstring $b$, i.e., the indicator of whether the edge $e$ is cut given $b$, which is 0 as $b_u = b_v$. The result in the case where $b_u \neq b_v$ holds similarly.
\end{proof}

Thus, if the optimal solution of the minimization only has support on $r_1 = d_{g_{6,1}}$ and $r_6 = d_{g_{6,6}}$, then, as a consequence of Proposition \ref{thm:flatteningLandscape} and Theorem \ref{thm:coloredSubgraphDecomposition}, the overall expected cut value $F_G(\gamma,\beta,\theta)$ will shrink as $\theta$ gets smaller, yielding a poor approximation ratio. Fortunately, for any 3-regular graph, there exists upper bounds on how frequently the colored edge neighborhoods types $[g_{6,1}]$ and $[g_{6,6}]$ can appear as seen Proposition \ref{thm:constraintIndividual}.

\begin{figure}
    \centering
    \begin{adjustbox}{max width=1.2\linewidth,center}
    \addtolength{\tabcolsep}{-0.4em}
    \begin{tabular}{cccc}
     \includegraphics[scale=0.28]{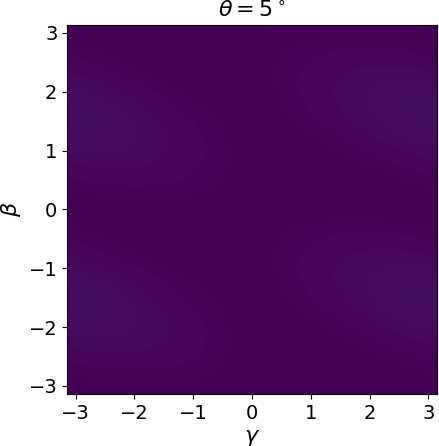} & \includegraphics[scale=0.28]{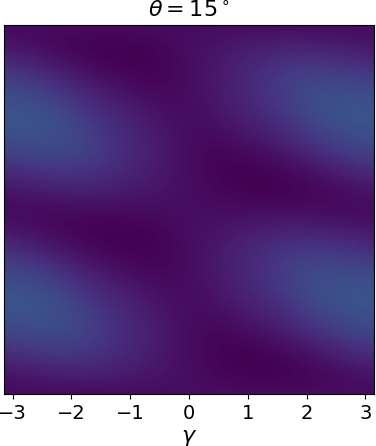} & \includegraphics[scale=0.28]{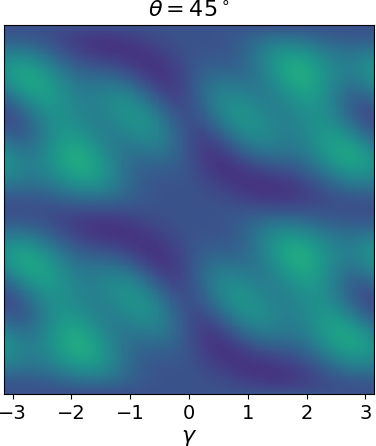}& \includegraphics[scale=0.28]{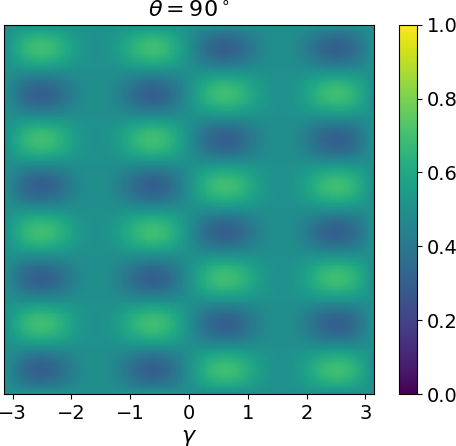} \\
     \includegraphics[scale=0.28]{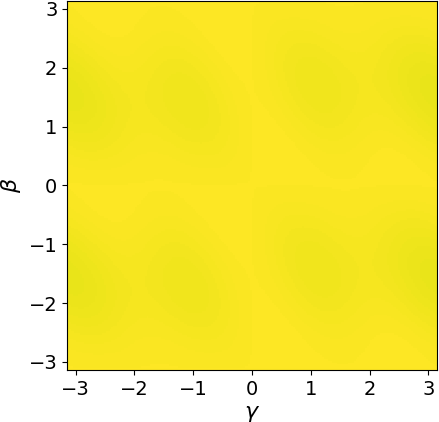} & \includegraphics[scale=0.28]{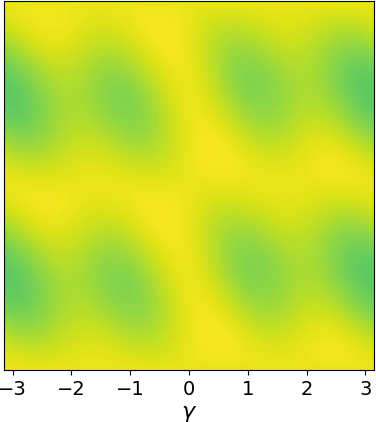} & \includegraphics[scale=0.28]{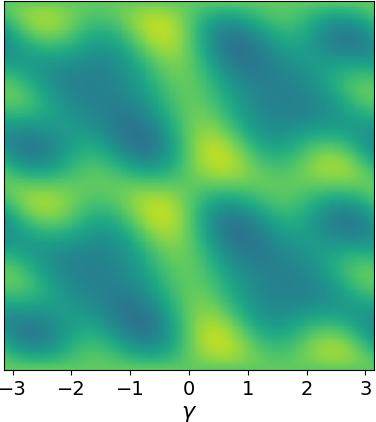}& \includegraphics[scale=0.28]{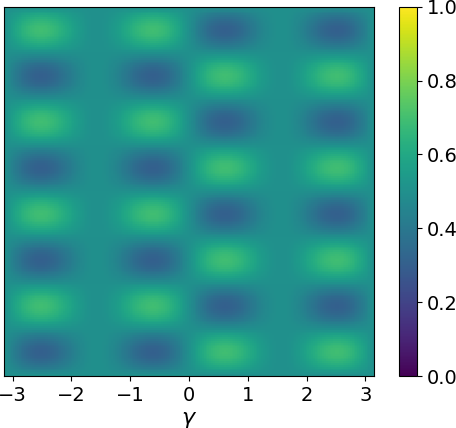}
    \end{tabular}
    \end{adjustbox}
    \caption{\label{fig:treeSubgraphLandscapes} \footnotesize The depth-1 parameter landscapes for warm-started QAOA at various fixed initialization angles $\theta$. Each point in the plot corresponds to a some choice $(\gamma,\beta)$ of variational parameters and the color indicates the probability of the central edge of $g_{6,1}$ (top row) or $g_{6,2}$ (bottom row) being cut at those parameters. The proof of Theorem \ref{thm:bitflipSameExpectedCut} can be slightly modified to show that the top and bottom row of these landscapes are also the landscapes for $g_{6,6}$ and $g_{6,5}$ respectively. At $\theta=90^\circ$, the highest central-edge cut probability achieved is 0.6924 which is also the approximation ratio achieved by depth-1 standard QAOA on 3-regular graphs.}
\end{figure}

\begin{prop}
\label{thm:constraintIndividual}
    For any 3-regular graph $G = (V,E)$ with a bitstring $b \in \{0,1\}^n$ corresponding to a cut, we have that $r_1 \leq \frac{1}{5}$ and $r_6 \leq \frac{1}{5}$.
\end{prop}
\begin{proof}
    We first prove that $r_1 \leq 0.5$.

    Suppose $e$ and $f$ are two distinct edges in $G$ such that the corresponding colored edge neighborhoods $G_e$ and $G_f$ are of type $[g_{6,1}]$. By inspection, it can be seen that $G_e$ and $G_f$ have no edges in common.
    
    Now, let $E^*$ be all the edges in $G$ that are part of some colored edge neighborhood of type $[g_{6,1}]$. By our previous observation, since the colored edge neighborhoods do not overlap and since $g_{6,1}$ has 5 edges, then there are $|E^*| = 5 \cdot d_{g_{6,1}} = 5R_1$ such edges. Clearly, $|E^*| \leq |E|$  and thus we have that,
    $$r_1 = \frac{R_1}{|E|} = \frac{\frac{1}{5}|E^*|}{|E|} \leq \frac{\frac{1}{5}|E|}{|E|}  = \frac{1}{5}.$$

    A similar argument can be used to show that $r_6 \leq \frac{1}{5}$ as well.
\end{proof}

Adding the inequalities, we have that $r_1 \leq \frac{1}{5}$ and $r_6 \leq \frac{1}{5}$ imply that $r_1+r_6 \leq \frac{2}{5} = 0.4$, i.e. at most 40\% of the edges have a colored edge neighborhood type of either $[g_{6,1}]$ or $[g_{6,6}]$. Can a tighter constraint be made regarding the sum of $r_1$ and $r_6$? In the case where the coloring/cut for the graph $b$ is 1-BLS, the answer is yes as implied by Proposition \ref{thm:sumOfBadSubgraphCounts} and Corollary \ref{thm:sumOfBadSubgraphCountsLocalOpt} below. The key idea of the proof is the following: algorithms such as 1-BLS are guaranteed to cut a certain fraction of edges, meaning that it is impossible for there to be a large number of colored edge-neighborhoods where the vertices of the central edge are uncut (according to the coloring).

\begin{prop}
\label{thm:sumOfBadSubgraphCounts}
    Let $b$ be a cut obtained by an algorithm of type $\mathcal{A}_\kappa$ on a suitable cubic graph $G$. Then, $s+t+s_{1,2}+r_1+r_6 \leq 1-\kappa$.
\end{prop}
\begin{proof}
    By definition of $\mathcal{A}_k$, $b$ cuts a $\kappa$-fraction of the total number of edges in the graph. Suppose that $e \in E(G)$ is such that the colored edge neighborhood $G_e$ is of type $[g_{4,1}], [g_{4,3}], [g_{5,1}], [g_{5,6}], [g_{6,1}]$ or $[g_{6,6}]$. Then, by construction, we have that $e$ is an edge that is \emph{not} cut by $b$. Conversely, any other edge would be cut in $b$. Thus, there must be exactly $$d_{g_{4,1}} +  d_{g_{4,3}} + d_{g_{5,1}} + d_{g_{5,6}} + d_{g_{6,1}} + d_{g_{6,6}}$$ edges that are not cut by $b$. However, from Equation \ref{eqn:expectedCutValueInTermsOfColoredGraphStructures}, we have that,
    $$d_{g_{4,1}} = S'_0+S_2+S_{1,1},$$
    $$d_{g_{4,3}} = S_0 + S_{1,3} + S_2',$$
    $$d_{g_{5,1}} = T_0 + T_1 + S_{1,2},$$
    $$d_{g_{5,6}} = T_2 + T_3 + S_{1,2},$$
    $$d_{g_{6,1}} = R_1,$$
    $$d_{g_{6,6}} = R_6.$$

    Making the substitutions above for the degeneracies yields the following for the number of uncut edges:
    $$S_0+S_0'+S_{1,1}+2S_{1,2}+S_{1,3}+S_2+S_2'+T_0+T_1+T_2+T_3+R_1+R_6;$$

    alternatively, the above expression can be derived by counting the number of core edges that would be uncut according to the coloring of each individual graph structure.
    
    Note that the fraction of \emph{uncut} edges is bounded above by $1-\kappa$. Thus, letting, $E^*$ denote the set of edges that $b$ does \emph{not} cut in $G$, we have that,
    \begin{align*}
    & s+t+s_{1,2}+r_1+r_6 \\
    =& s_0+s_0'+s_{1,1}+2s_{1,2}+s_{1,3}+s_2+s_2'+t_0+t_1+t_2+t_3+r_1+r_6 \\
    =& \frac{S_0+S_0'+S_{1,1}+2S_{1,2}+S_{1,3}+S_2+S_2'+T_0+T_1+T_2+T_3+R_1+R_6}{|E|}\\
    =& \frac{|E^*|}{|E|} \\
    \leq & \frac{(1-\alpha)|E|}{|E|} \\
    =&   1-\kappa.
    \end{align*}
\end{proof}
\begin{corollary}
    \label{thm:sumOfBadSubgraphCountsLocalOpt}
    Let $G = (V,E)$ be any 3-regular graph and let $b$ be a cut that is locally optimal with respect to 1-BLS. Then, $s+t+s_{1,2}+r_1+r_6 \leq \frac{1}{3}$.
\end{corollary}
\begin{proof}
    Since 1-BLS on cubic graphs is an algorithm of type $\mathcal{A}_{2/3}$ (Proposition \ref{thm:localMaxCut3Reg}) and $b$ can clearly be obtained from 1-BLS, then we can apply Proposition \ref{thm:sumOfBadSubgraphCounts} with $\alpha = \frac{2}{3}$ to obtain that,
    $$s+t+s_{1,2}+r_1+r_6 \leq 1-\alpha = \frac{1}{3}.$$
\end{proof}

We will later see that there exists infinitely many 3-regular graphs such that there exists a 1-BLS bitstring $b$ that yields $s+t+s_{1,2}+r_1+r_6 = \frac{1}{3}$. In order to make an even tighter constraint of the form $s+t+s_{1,2}+r_1+r_6 \leq k$ with $k < \frac{1}{3}$, additional assumptions need to be made regarding the bitstring $b$. Proposition \ref{thm:sumOfBadSubgraphCounts} suggests that if there is an algorithm that generates a cut that cuts a high fraction of edges (compared to the total number of edges), then such a constraint could be tightened. Indeed, as discussed in Section \ref{sec:classicalMaxCutResults}, there exists such algorithms such as the algorithms of \cite{bondy1986largest} and \cite{zhu2009bipartite} (with 1-BLS post-processing) of type $\mathcal{A}'_{4/5}$ and $\mathcal{A}'_{17/21}$ respectively. These algorithms only run on triangle-free graphs (with additional restrictions in the case of the algorithm by  \cite{zhu2009bipartite}) and hence the constraint $s+t+s_{1,2}+r_1+r_6 \leq 1-\kappa$ in Proposition \ref{thm:sumOfBadSubgraphCounts} can be simplified as seen in Corollary \ref{thm:sumOfBadSubgraphCountTriangleFree}.

\begin{corollary}\label{thm:sumOfBadSubgraphCountTriangleFree}
Let $b$ be a cut obtained by an algorithm of type $\mathcal{A}'_\kappa$ on a suitable triangle-free cubic graph $G$. Then, $r_1+r_6 \leq 1-\kappa$.
\end{corollary}
\begin{proof}
    Since $G$ is triangle-free, then $s=t=0$ and the constraint $s+t+s_{1,2}+r_1+r_6 \leq 1-\kappa$ simplifies to $r_1+r_6 \leq 1-\kappa$.
\end{proof}

 Moreover, for cuts obtained by algorithms of type $\mathcal{A}'_\kappa$ on triangle-free cubic graphs, the constraints $\mathbf{s}=\mathbf{0}$ and $\mathbf{t}=\mathbf{0}$ can also be added to the minimization while still also being a lower bound for the approximation ratio. With this in mind, we can obtain lower bounds on the approximation ratio of warm-started QAOA with warm-start $\ket{b_\theta}$ where $b$ is obtained from an algorithm of type $\mathcal{A}_\kappa$ and $\mathcal{A}'_\kappa$ respectively:

\begin{equation}
    \alpha_{\text{QAOA-WS}(\mathcal{A}_\kappa,\theta)} \geq \min_{(\mathbf{r}, \mathbf{s}, \mathbf{r}) \in \mathcal{P}(\kappa)}\frac{M_\theta(\mathbf{r},\mathbf{s},\mathbf{t})}{1-s - t},
\end{equation}

\begin{equation}
    \alpha_{\text{QAOA-WS}(\mathcal{A}'_\kappa,\theta)} \geq  \min_{\mathbf{t} \in \mathcal{P'}(\kappa)}M_\theta(\mathbf{r},\mathbf{0},\mathbf{0}),
\end{equation}

with the parameterized constraint polytopes $\mathcal{P}(\kappa)$ and $\mathcal{P}'(\kappa)$ defined as
\begin{align*}\mathcal{P}(\kappa) =  \{(\mathbf{r}, \mathbf{s}, \mathbf{t}) \in \mathbb{R}^{17}: & \mathbf{r}, \mathbf{s}, \mathbf{t} \geq \mathbf{0},\\ 
& 5s+3t+r=1, \\
& 4s+3t \leq 2/3, \\
& r_1,r_6 \leq 1/5, \\ & s+t+s_{1,2}+r_1+r_6 \leq 1-\kappa\},
\end{align*}
\begin{align*}\mathcal{P}'(\kappa) =  \{\mathbf{r} \in \mathbb{R}^{6}: & \mathbf{r} \geq  \mathbf{0},\\ 
& r=1, \\
& r_1,r_6 \leq 1/5, \\ & r_1+r_6 \leq 1-\kappa\},
\end{align*}

Note that for $\kappa \geq 4/5$, because of the constraint $r_1+r_6 \leq 1-\kappa \leq 1/5$ in $\mathcal{P}'(\kappa)$,  the constraints $r_1 \leq \frac{1}{5}$ and $r_6 \leq \frac{1}{5}$ become redundant (due to the non-negativity constraints) and can thus be removed without changing the constraint polytope.

Similar to what was done in Sections \ref{sec:graphStructures} and \ref{sec:coloredGraphStructures}, we can obtain an even lower bound by swapping the order of maximization and minimization of the optimization problems above; this is formalized in Theorem \ref{thm:lowerBoundVaryingKappa}. We denote these lower bounds as $\mathbf{LB}_\kappa(\theta)$ and $\mathbf{LB}'_\kappa(\theta)$ corresponding to the case where the warm-start is generated from $\mathcal{A}_\kappa$ and $\mathcal{A}'_\kappa$ respectively. This swap has numerous benefits already discussed in Section \ref{sec:graphStructures} (with further details in Section \ref{sec:optimizationDetails}) and preliminary numerical experiments suggest this swap has zero or little effect on the final objective value of the optimization.

\begin{theorem}
\label{thm:lowerBoundVaryingKappa}
    The following inequalities hold:
    \begin{equation}
    \alpha_{\text{QAOA-WS}(\mathcal{A}_\kappa,\theta)} \geq  \min_{(\mathbf{r}, \mathbf{s}, \mathbf{t}) \in \mathcal{P}(\kappa)}\frac{M_\theta(\mathbf{r},\mathbf{s},\mathbf{t})}{1-s - t} \geq \textbf{LB}_\kappa(\theta) := \max_{\gamma,\beta} \min_{(\mathbf{r}, \mathbf{s}, \mathbf{t}) \in \mathcal{P}(\kappa)} \frac{\frac{1}{m} F_{\mathbf{r},\mathbf{s},\mathbf{t}}(\gamma,\beta,\theta)}{1-s-t},
\end{equation}

\begin{equation}
    \alpha_{\text{QAOA-WS}(\mathcal{A}'_\kappa,\theta)} \geq  \min_{\mathbf{r} \in \mathcal{P}'(\kappa)}M_\theta(\mathbf{r},\mathbf{0},\mathbf{0}) \geq \textbf{LB}'_\kappa(\theta) := \max_{\gamma,\beta} \min_{\mathbf{r} \in \mathcal{P}'(\kappa)} \frac{1}{m} F_{\mathbf{r},\mathbf{0},\mathbf{0}}(\gamma,\beta,\theta).
\end{equation}
\end{theorem}

For certain choices of $\kappa$ corresponding to known classical algorithms described in Section \ref{sec:classicalMaxCutResults}, these lower bounds are shown in Figure \ref{fig:plotAR_varyingConstraints} which plots the above lower bounds for values of $\theta = 1^\circ, 2^\circ, \dots, 90^\circ$; see Section \ref{sec:optimizationDetails} for additional details regarding how these lower bounds were numerically calculated. In Observation \ref{thm:plotObservation}, we list several observations that can be made from Figure \ref{fig:plotAR_varyingConstraints}.

\begin{obs}
\label{thm:plotObservation}
The following are observations regarding lower bounds on the approximation ratios.
\begin{enumerate}
    \item Since adding constraints can only ever increase the optimal objective value of a minimization problem, we have that for all choices of initilization angle $\theta$:
$$\textbf{LB}(\theta) \leq \textbf{LB}_{2/3}(\theta) \leq \textbf{LB}'_{4/5}(\theta) \leq \textbf{LB}'_{17/21}(\theta).$$
Similarly, consistent with the proof of Corollary \ref{thm:depth0ApproxRatio2}, we have
$$\textbf{LB}_0^{(0)}(\theta) \leq \textbf{LB}^{(0)}_{2/3}(\theta) \leq \textbf{LB}^{(0)}_{4/5}(\theta) \leq \textbf{LB}^{(0)}_{17/21}(\theta).$$
    \item For $\textbf{LB}_{2/3}(\theta), \textbf{LB}'_{4/5}(\theta)$, $\textbf{LB}'_{7/21}(\theta)$, and their depth-0 counterparts, at $\theta=0$ (where QAOA simply returns the cut used to construct the initial state), the lower bound achieved is identical to the approximation ratio achieved for 1-BLS, $\mathcal{A}_{4/5}$, and $\mathcal{A}_{17/21}$ respectively.
    \item Since warm-started QAOA with initialization angle $\theta=\pi/2$ is equivalent to the standard QAOA, all of the depth-0 and depth-1 lower bounds achieve a value of $0.5$ and $0.6924$ (respectively) at $\theta = \pi/2$, coinciding with the approximation ratio of standard QAOA on 3-regular graphs \cite{farhi2014quantum}; Figure \ref{fig:plotAR_zoom_90_degrees} gives a closer look at the behavior of these lower bounds near $\theta = \pi/2$.
    \item With the exception of $\theta=0$, Figure \ref{fig:plotAR_varyingConstraints} suggests that the depth-0 lower bounds are strictly worse than the corresponding depth-1 lower bounds, i.e.,
    $$\mathbf{LB}_0^{(0)}(\theta)  < \mathbf{LB}(\theta),$$
    $$\mathbf{LB}^{(0)}_{2/3}(\theta)  < \mathbf{LB}_{2/3}(\theta),$$
    $$\mathbf{LB}^{(0)}_{4/5}(\theta)  < \mathbf{LB}'_{4/5}(\theta),$$
    $$\mathbf{LB}^{(0)}_{\mathcal{A}_{17/21}}(\theta)  < \mathbf{LB}_{\mathcal{A}_{17/21}}(\theta).$$
    \item Since $4/5 = 0.8000$ and $17/21 \approx 0.8095$ are close in value, it is not too surprisingly to see that, for any $\theta$, that $\mathbf{LB}_{\mathcal{A}_{4/5}}(\theta)$ and $\mathbf{LB}_{\mathcal{A}_{17/21}}(\theta)$ are also close in value; the same can be said about $\mathbf{LB}^{(0)}_{\mathcal{A}_{4/5}}(\theta)$ and $\mathbf{LB}^{(0)}_{\mathcal{A}_{17/21}}(\theta)$. In Figure \ref{fig:plotAR_zoom_4/5_17/21}, we plot the differences $\mathbf{LB}_{\mathcal{A}_{17/21}}(\theta) - \mathbf{LB}_{\mathcal{A}_{4/5}}(\theta)$ and $\mathbf{LB}^{(0)}_{\mathcal{A}_{17/21}}(\theta) - \mathbf{LB}^{(0)}_{\mathcal{A}_{4/5}}(\theta)$.
    \item Around $\theta \approx 50^\circ$, there is a sharp corner in the plots of both $\mathbf{LB}_{\mathcal{A}_{4/5}}(\theta)$ and $\mathbf{LB}_{\mathcal{A}_{17/21}}(\theta)$ where the bound on the approximation ratio reaches a local minimum (as $\theta$ varies). A closer look at this corner is provided in Figure \ref{fig:plotAR_zoom_4/5_17/21}.
    \item Around $\theta\approx 60$, both $\mathbf{LB}_{\mathcal{A}_{4/5}}(\theta)$ and $\mathbf{LB}_{\mathcal{A}_{17/21}}(\theta)$ reach their largest value; this is more clearly seen in Figure \ref{fig:plotAR_zoom_4/5_17/21}. In particular, at this maximum, the lower bounds both are \emph{nearly} equal to the corresponding classical approximation ratios of $4/5$ and $17/21$ respectively.
    \item For 1-BLS, the lower bound $\mathbf{LB}_\text{1-BLS}(\theta)$ exhibits strange, non-monotonic behavior around both $\theta \approx 45^\circ$ and $\theta \approx 60^\circ$ which is more clearly seen in Figure \ref{fig:plotAR_zoom_1BLS}.
\end{enumerate} 
\end{obs}

Interestingly, Part 7 of Observation \ref{thm:plotObservation} is not a coincidence. In general, for warm-started QAOA on any $k$-regular graph with $k$ odd, with the right choices of variational parameters, one can evolve the warm-start state with initialization angle $\theta$ to the same warm-start state but with initialization angle $3\theta$, which we will prove in Theorem \ref{thm:recoverCutAt60DegreesGeneralized} below. At $\theta = 60^\circ$, we have $3\theta = 180^\circ$, meaning that the original bitstring used in the warm-start initialization is \emph{effectively} recovered as a bitstring and its bitwise-negation have the same objective value in the context of \MC.

Before proving this observation, we first state a useful lemma which states that at $\gamma=\pi$, the cost Hamiltonian for \MC{} takes on a special form for regular graphs.
\begin{lemma}
    \label{thm:unitaryisPhaseFlip}
        Let $G$ be a $k$-regular graph with $n$ vertices and let $C$ be the \mc{} Hamiltonian for $G$ and let $\gamma = \pi$. Then, the cost unitary $U(C, \gamma)$ of QAOA is equal to the following:
        $$U(C, \gamma) = \begin{cases} Z^{\otimes n}, & \text{$k$ odd}\\ I, & \text{$k$ even}\end{cases}.$$
    \end{lemma}

Lemma \ref{thm:unitaryisPhaseFlip} was already known, e.g. in \cite{zhou2020quantum}, this result is stated (but not proved) and is used to determine the periodicity of the variational parameters $\gamma$ and $\beta$ for standard QAOA on regular graphs; for the curious reader, we provide a complete proof in the Supplementary Materials. We now prove Theorem \ref{thm:recoverCutAt60DegreesGeneralized} and Corollary \ref{thm:recoverCutAt60Degrees}.

\begin{theorem}
\label{thm:recoverCutAt60DegreesGeneralized}
    Let $G$ be a (unit-weight) $k$-regular graph where $k$ is odd and let $b$ be a bitstring corresponding to a cut in $G$. Then there exists variational parameters $(\gamma,\beta)$ such that depth-1 warm-started QAOA with warm-start $\ket{b_\theta}$ outputs the state $\ket{\psi_p(\gamma,\beta, \ket{b_\theta})} = \ket{b_{3\theta}}$.
\end{theorem}
\begin{proof}
        Let $\gamma = \pi$ and let $\beta = \pi/2$ and let $\ket{\psi} = U(B,\beta)U(C,\gamma)\ket{b_\theta}$ denote the state that is output from warm-started QAOA with warm-start $\ket{b_\theta}$. From Lemma \ref{thm:unitaryisPhaseFlip}, we have that $U(C,\gamma) = Z^{\otimes n}$, where $n$ is the number of vertices in the graph $G$. Following the notation of Section \ref{sec:constructionOfWarmStartedStates} and writing $\ket{b_\theta} = \bigotimes_{j=1}^n \ket{\vec{n}_j}$, the mixing unitary can also be written as a Kronecker product: $U(B,\beta) = \bigotimes_{j=1}^n U(B_{\vec{n_j}},\beta).$ We can thus write,
        $$\ket{\psi} =  U(B,\beta)U(C,\gamma)\ket{b_\theta} = \left(\bigotimes_{j=1}^n U(B_{\vec{n_j}},\beta)\right)\left(Z^{\otimes n}\right)\ket{b} = \bigotimes_{j=1}^n \left(U(B_{\vec{n_j}}, \beta)Z\ket{(b_j)_\theta}\right),$$
        or equivalently, $\ket{\psi} = \bigotimes_{j=1}^n \ket{\psi_j}$ where,
        $$\ket{\psi_j} = U(B_{\vec{n_j}}, \beta)Z\ket{(b_j)_\theta}.$$
        From here, it suffices to show that $U(B_{\vec{n_j}}, \beta) Z\ket{(b_j)_\theta} = \ket{(b_j)_{3\theta}}$ for all $j=1,\dots,n$. We now fix a $j$ and first consider the case where $\ket{(b_j)_\theta} = \ket{0_\theta}$. We can then write $U(B_{\vec{n_j}}, \beta)$ as a sequence of rotations about the $y$ and $z$ axis of the Bloch sphere:
        $$U(B_{\vec{n_j}}, \beta) = R_y(\theta)R_z(2\beta)R_y(-\theta),$$
        where
        $$R_y(\theta) = e^{-i\frac{\theta}{2}Y} = \cos(\theta/2)I - i\sin(\theta/2)Y,$$
        $$R_z(\theta) = e^{-i\frac{\theta}{2}Z} = \cos(\theta/2)I - i\sin(\theta/2)Z;$$
        we refer the reader to \cite{blochSphereRotations} for more details regarding the decomposition of the mixer in terms of Bloch sphere rotations. Modulo an (unobservable) global phase, $Z = R_z(\pi)$, and thus, a geometrical argument makes the following clear:
        \begin{align*}
            U(B_{\vec{n}_j},\beta)Z\ket{(b_j)_\theta} &= \Big(R_y(\theta)R_z(2\cdot \pi/2)R_y(-\theta)\Big)R_z(\pi)\ket{0_\theta}\\
            &= R_y(\theta)R_z(\pi)R_y(-\theta)\ket{0_{-\theta}}\\
            &= R_y(\theta)R_z(\pi)\ket{0_{-2\theta}}\\
            &= R_y(\theta)\ket{0_{2\theta}}\\
            &= \ket{0_{3\theta}},\\
        \end{align*}

        as desired.

        Now, suppose that $\ket{(b_j)_\theta} = \ket{1_\theta}$. Then,
        \begin{align*}
            U(B_{\vec{n}_j},\beta)Z\ket{(b_j)_\theta} &=  U(B_{\vec{n}_j},\beta)Z\ket{1_\theta}\\
            &=U(B_{\vec{n}_j},\beta)Z\ket{0_{\pi-\theta}}\\
            &= \ket{0_{3(\pi - \theta)}}\tag{previous calculations with $\pi-\theta$ in place of $\theta$}\\
            &= \ket{0_{\pi-3\theta}} \tag{$3(\pi-\theta) = \pi-3\theta \pmod{2\pi}$}\\
            &= \ket{1_{3\theta}},
        \end{align*}

        finishing the proof.
    \end{proof}

\begin{corollary}
    \label{thm:recoverCutAt60Degrees}
    Let $G$ be a (unit-weight) $k$-regular graph where $k$ is odd and let $b$ be a bitstring corresponding to a cut in $G$. Then there exists variational parameters $(\gamma,\beta)$ such that depth-1 warm-started QAOA with warm-start $\ket{b_\theta}$ with $\theta = 60^\circ$ yields an expected cut value of $\cut(b)$. More specifically, the state that is output from the QAOA circuit is $\ket{\psi_p(\gamma,\beta, \ket{b_\theta})} = \ket{\bar{b}}$ where $\bar{b}$ is the bitwise-negation of the bitstring $b$.
\end{corollary}

We next investigate the optimal values of the variables in the inner-minimization problems of these lower bounds. We found that the optimal values correspond to graphs with no triangles and with many colored-edge neighborhoods of the type $[g_{6,1}]$ and $[g_{6,6}]$ as seen in Conjecture \ref{thm:noTrianglesConjecture}.

\begin{conjecture}
\label{thm:noTrianglesConjecture}
    For the variables in the minimization problem defined by $\textbf{LB}_{\kappa}(\theta)$ and $\textbf{LB}'_{\kappa}(\theta)$ , for all initialization angles $\theta$, the optimal values satisfy the following:
    \begin{itemize}
        \item $\mathbf{r} = \mathbf{s} = \mathbf{0}$, i.e., $G$ is triangle-free, and 
        \item $r_1 + r_6 = 1 - \kappa$; together with the previous bullet, this implies that the constraint $s+t+s_{1,2}+r_1+r_6 \leq 1 - \kappa$ is tight.
    \end{itemize}
\end{conjecture}

Our computations (see Section \ref{sec:optimizationDetails} and the data tables in the Supplementary Materials) are consistent with Conjecture \ref{thm:noTrianglesConjecture} for values of $\theta$ we tested ($\theta=1^\circ, 2^\circ, \dots, 90^\circ$).

In regards to tightness, we believe that the lower bound $\textbf{LB}_\text{2/3}(\theta)$ is \emph{nearly} tight with respect to $\alpha'_{{(\text{1-BLS}+\text{QAOA}})_\theta}$ (see Section \ref{sec:warmstarted_QAOA_AR} regarding this notation of approximation ratio); in particular, the following holds for all 3-regular graphs $G$ and all 1-BLS cuts $b$:
$$\textbf{LB}_{2/3}(\theta) \leq \alpha'_{{(\text{1-BLS}+\text{QAOA}})_\theta} \leq \alpha_{{(\text{WS-QAOA}})_\theta}(G,b),$$
which implies that,
$$|\alpha'_{{(\text{1-BLS}+\text{QAOA}})_\theta} - \textbf{LB}_{2/3}(\theta)| \leq |\alpha_{{(\text{WS-QAOA}})_\theta}(G,b) - \textbf{LB}_{2/3}(\theta)|.$$

We were able to find a graph $G$ and 1-BLS bitstring $b$ such that the right side of the above inequality is at most $0.03$ across all values of $\theta$ we tested, meaning that the lower bound $\textbf{LB}_{2/3}(\theta)$ is within $0.03$ of $\alpha'_{{(\text{1-BLS}+\text{QAOA}})_\theta}$. In fact there are infinitely many such graphs (and corresponding bitstrings); these are given by the 3-regular bipartite graphs defined in Figure \ref{fig:cubicBipartiteGraph} whose number of vertices is divisible by 4. Using the same vertex-labeling in Figure \ref{fig:cubicBipartiteGraph}, the corresponding bitstring $b$ is given by $b_{u_j} = b_{v_j} = 0$ if $j$ is odd and $b_{u_j} = b_{v_j} = 1$ if $j$ is even. It is not too difficult to verify that $b$ is locally optimal with respect to 1-BLS. An example of such a graph (on 12 nodes) and corresponding bitstring $b$ is given in Figure \ref{fig:coloredCubicBipartiteGraph}. For such graphs, we have the following in regards to their distribution of colored-edge neighborhoods:
$$\mathbf{s} = \mathbf{t} = 0,$$ $$r_1 = r_6 = 1/6,$$ $$r_4 = 2/3,$$ $$ r_2 = r_3 = r_5 = 0.$$
For many values of $\theta$, we found (see data tables in Supplementary Materials) that this distribution of colored-edge neighborhoods matches those found in the inner-minimization of $\textbf{LB}_{2/3}(\theta)$. For those $\theta$'s, one should note that the optimal $\gamma$ and $\beta$ for the outer-maximization of $\textbf{LB}_{2/3}$ is not necessarily the same as the optimal value of $\gamma$ and $\beta$ for the nearly-tight colored graphs described earlier; this is due to ordering of the inner and outer optimizations that define $\textbf{LB}_{2/3}$ (see Section \ref{sec:optimizationDetails}).

Lastly, for each of the above lower bounds, we look at the corresponding optimal values of $\gamma$ and $\beta$ that were found in the outer maximization defining each lower bound. Figure \ref{fig:gammaBetaScatterPlot} illustrates that the optimal choice of these variational parameters are heavily dependent on the choice of initialization angle. Moreover, as the initialization angle $\theta$ changes, the distant clusters of points in Figure \ref{fig:gammaBetaScatterPlot} suggests that the optimal $\gamma$ and $\beta$ parameters do not change in a continuous manner; this explains the sudden changes in behavior in the approximation ratio curves found in Figure \ref{fig:plotAR_varyingConstraints}. The numerical data in both Figures \ref{fig:plotAR_varyingConstraints} and \ref{fig:gammaBetaScatterPlot} are tabulated in various tables in the Supplementary Materials.

\begin{figure}
    \centering
    \includegraphics[scale=.95]{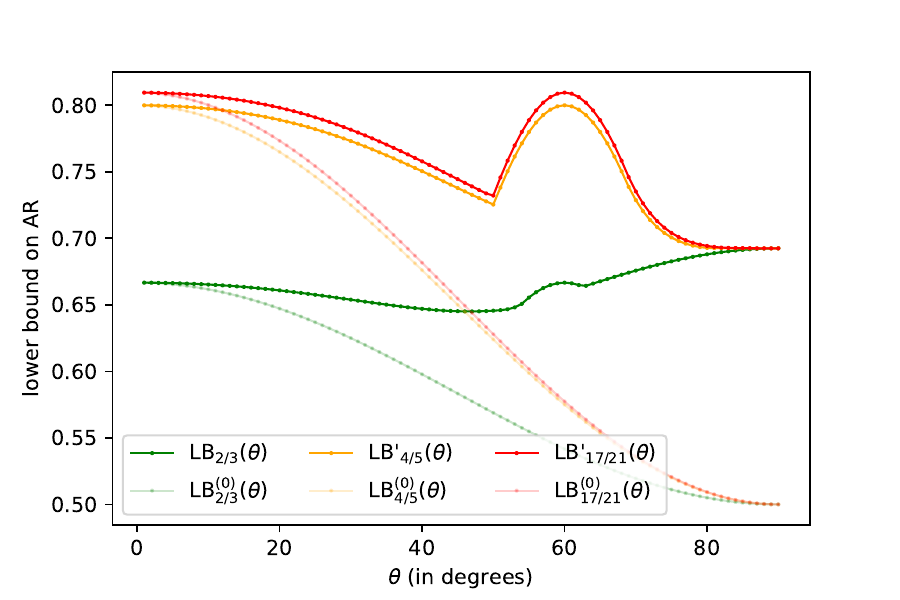}
\caption{\footnotesize\label{fig:plotAR_varyingConstraints} A plot demonstrating how the lower bounds $\textbf{LB}^{(0)}_\kappa(\theta), \textbf{LB}_\kappa(\theta), \textbf{LB}'_\kappa(\theta)$ on the approximation ratio of depth-0 and depth-1 warm-started QAOA (respectively) change as a function of the initialization angle $\theta$ for $\kappa \in \{2/3, 4/5, 17/21\}$.}
\end{figure}

\begin{figure}
    \centering
    \includegraphics[scale=0.5]{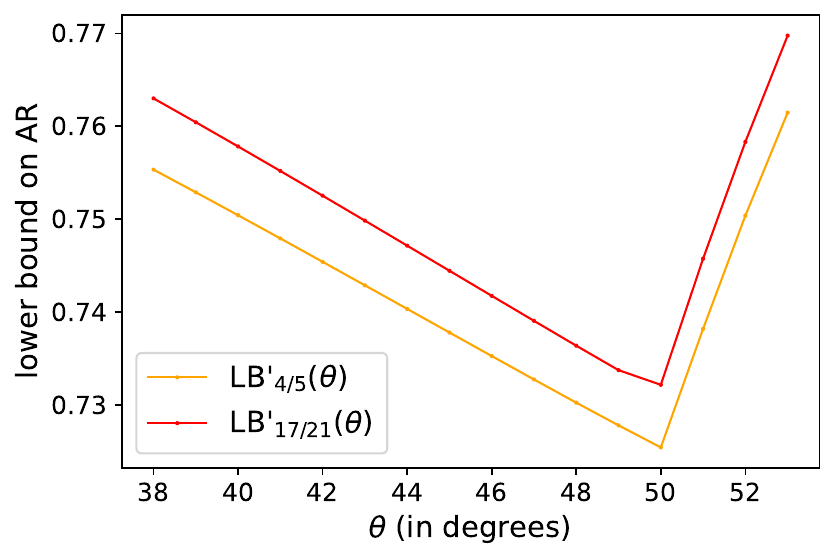}
    \includegraphics[scale=0.5]{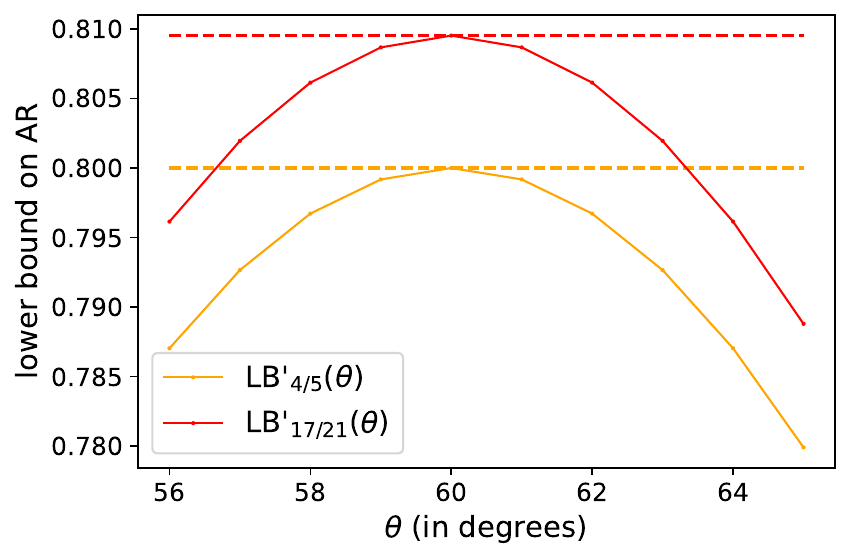}
    
    \includegraphics[scale=0.5]{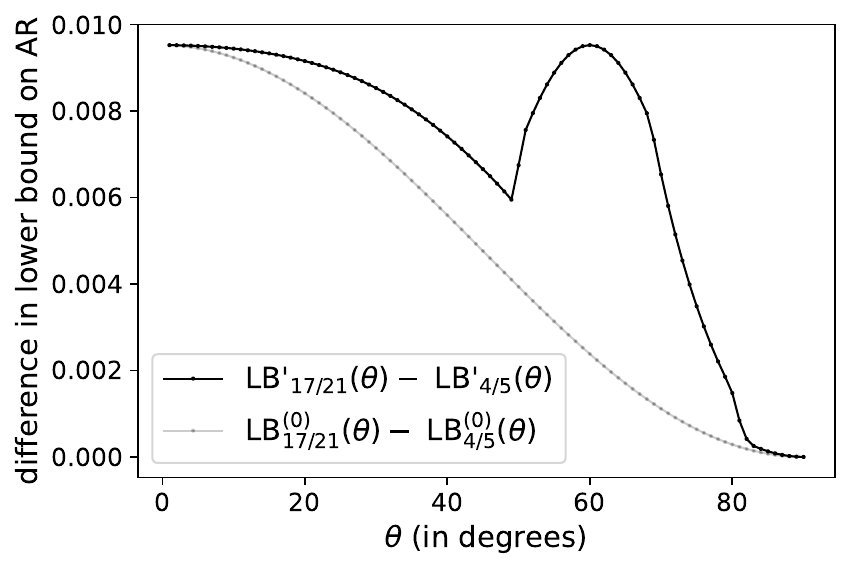}
    
    \caption{\label{fig:plotAR_zoom_4/5_17/21}\footnotesize In the top row, $\mathbf{LB}'_{\mathcal{A}_{4/5}}(\theta)$ and $\mathbf{LB}'_{\mathcal{A}_{17/21}}(\theta)$ are plotted around $\theta = 45^\circ$ (top-left) and $\theta = 60^\circ$ (top-right). The bottom plot depicts the differences $\mathbf{LB}'_{\mathcal{A}_{17/21}}(\theta)-\mathbf{LB}'_{\mathcal{A}_{4/5}}(\theta)$ and $\mathbf{LB}^{(0)}_{\mathcal{A}_{17/21}}(\theta)-\mathbf{LB}^{(0)}_{\mathcal{A}_{4/5}}(\theta)$. The horizontal lines in the top-right plot occur at $4/5$ and $17/21$, i.e., the approximation ratios guaranteed by classical algorithms of type $\mathcal{A}'_{4/5}$ and $\mathcal{A}'_{17/21}$.}
\end{figure}

\begin{figure}
    \centering
    \includegraphics[scale=0.5]{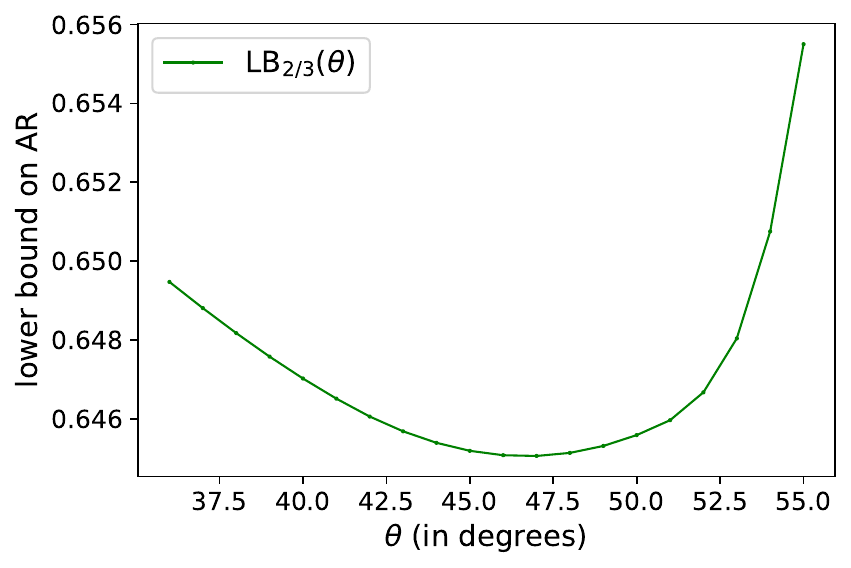}
    \includegraphics[scale=0.5]{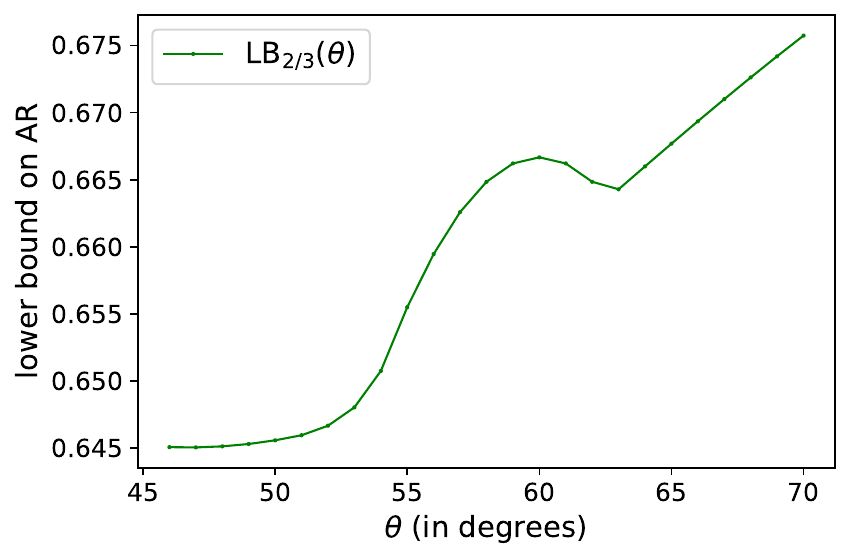}
    \caption{\label{fig:plotAR_zoom_1BLS}\footnotesize The lower bound $\mathbf{LB}_{2/3}(\theta)$ is plotted around $\theta = 45^\circ$ (left) and $\theta = 60^\circ$ (right). The horizontal line in the right plot occurs at $2/3$, i.e., the approximation ratios guaranteed by a classical algorithm of type $\mathcal{A}_{2/3}$.}
\end{figure}

\begin{figure}
    \centering
    \includegraphics[scale=0.5]{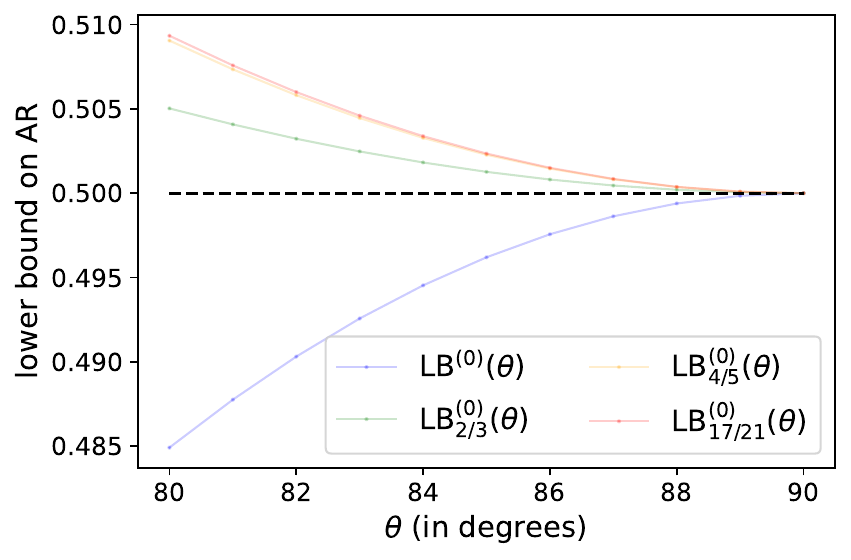}
    \includegraphics[scale=0.5]{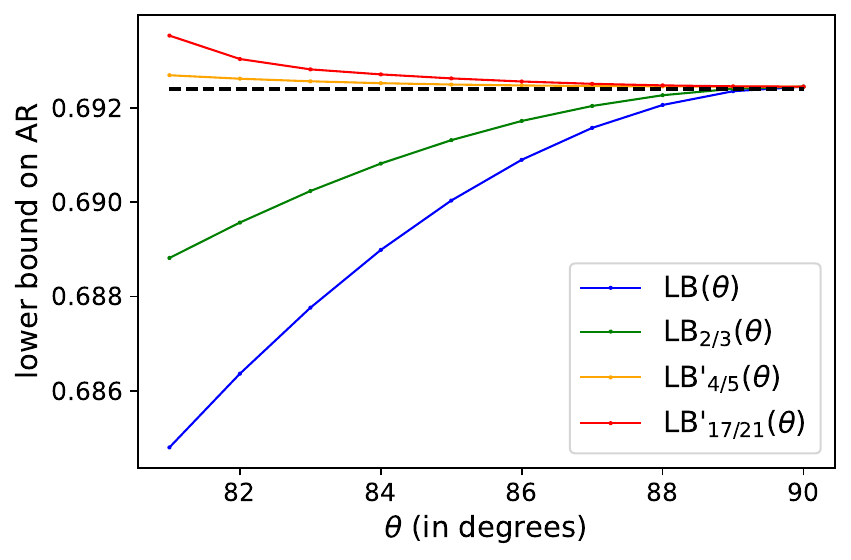}
    \caption{\label{fig:plotAR_zoom_90_degrees}\footnotesize Plots of lower bounds of the approximation ratio for various kinds of warm-started QAOA near $\theta = 90^\circ$. The left and right plot these bounds for depth-0 and depth-1 warm-started QAOA respectively.}
\end{figure}

\begin{figure}
    \centering
    \includegraphics[scale=0.3]{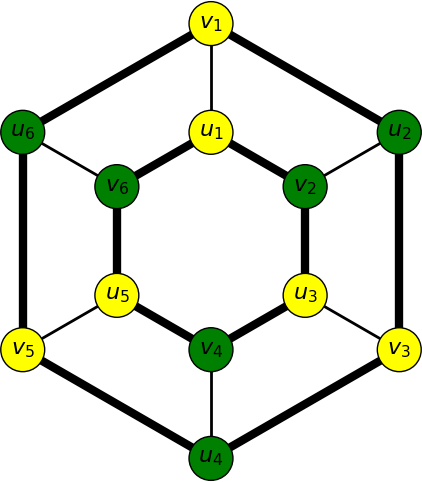}
    \caption{\label{fig:coloredCubicBipartiteGraph} \footnotesize An example of a 3-regular bipartite graph on 12-nodes and a bitstring coloring of the vertices that shows that $\textbf{LB}_\text{1-BLS}(\theta)$ is nearly tight.}
\end{figure}

\begin{figure}
    \label{fig:gammaBetaScatterPlot}
    \centering
    \includegraphics[scale=0.5]{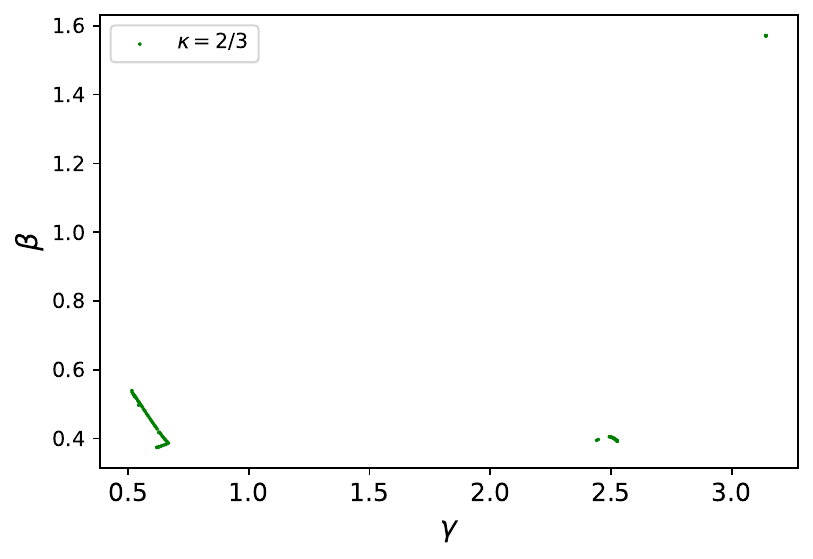}\includegraphics[scale=0.5]{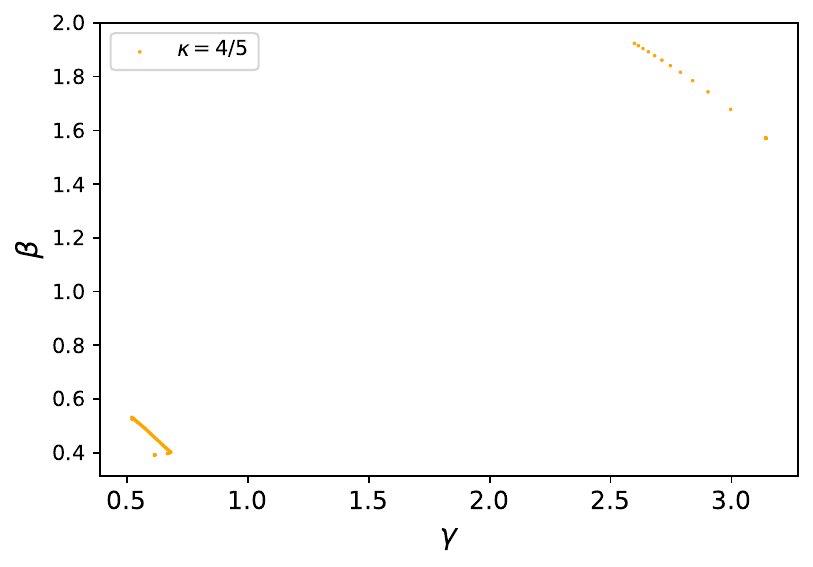}\\
    \includegraphics[scale=0.5]{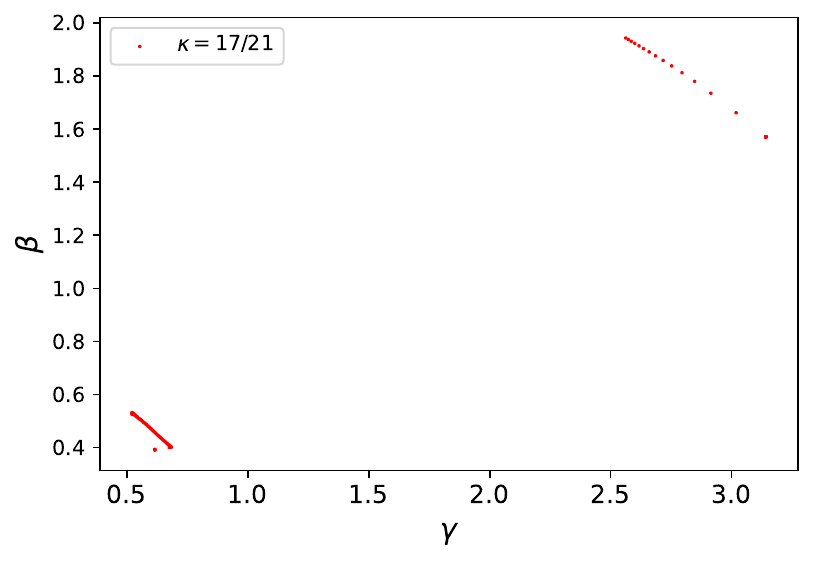}
    \caption{\footnotesize Scatter-plot of optimal variational parameters obtained as a result of numerically calculating various lower bounds on the approximation ratio for various values of $\theta$, the initialization angle.}
\end{figure}

\section{Implementation Details}
\subsection{Graph Construction Details}
\label{sec:subgraphConstructionDetails}

Consider a $k$-regular colored graph $(G,b)$ and consider a non-colored depth-$p$ edge neighborhood $G_e^{(p)}$. Borrowing some of the terminology regarding graph structures (Section \ref{sec:coloredGraphStructures}), the vertices of $G_e^{(p)}$ can be divided into two groups: the exterior vertices $V_\text{ext}$ which are adjacent to at least one other vertex in $G$ but outside of $G_e^{(p)}$ and the remaining core vertices $V_\text{core}$. When $G$ is a $k$-regular graph, it is easy to identify the exterior and core vertices of $G_e^{(p)}$  just by looking at $G_e^{(p)}$: the vertices with degree $k$ (with respect to $G_e^{(p)}$) are exactly the core vertices (since if they were adjacent to any other vertices in $G$, then $G$ would no longer be $k$-regular) and the remaining vertices must be the exterior vertices.

When coloring these edge-neighborhoods according to $b$, observe that if the bit/color corresponding to a core vertex is flipped, then it is easy to calculate the change in the cut value for the overall graph as seen in the example in Figure \ref{fig:ImprovementExample}. The same cannot be said if the exterior vertices are flipped since calculating the change in the overall cut in $G$ would require one to have knowledge of the colors of the vertices that are adjacent to the exterior vertices but outside of the colored-edge neighborhood.

Thus, in order to find all the colored-edge neighborhoods that are possible assuming the bitstring coloring $b$ for $G$ is 1-BLS optimal, one way to do this is to first consider all fixed sub-colorings of the exterior vertices and for each of those sub-colorings, consider all the ways to color the core vertices and remove any colorings that are not consistent with $b$ being 1-BLS with respect to $G$. This idea is formalized in Algorithm \ref{alg:coloredEdgeNeighborhoods}, which, given all the \emph{non-colored} edge-neighborhoods of a certain depth, finds all the ways to color them that is consistent with some overall $k$-regular graph $G$ having a coloring that is 1-BLS optimal.

\begin{algorithm}
\label{alg:coloredEdgeNeighborhoods}
    \caption{\footnotesize Finding colored edge-neighborhoods that are valid with respect to 1-BLS}
    \KwData{Graph regularity $k$, circuit depth $p$, and $\mathcal{G}^{(p)}_k$, a set of representatives of non-colored depth-$p$ edge-neighborhood types that are possible for $k$-regular graphs.}
    
    \BlankLine
    let $\mathcal{H}$ be an empty list\;
    
    \ForEach{marked graph $(g,e) \in \mathcal{G}_k^{(p)}$}
    {
     $V_\text{core} = \{v \in V(g) : \degree_g(v) = k\}$\;
     $V_\text{ext} = \{v \in V(g) : \degree_g(v) \neq k\}$\;
        \ForEach{subcoloring $b_\text{ext}$ of $V_\text{ext}$}
        {
            \ForEach{subcoloring $b_\text{core}$ of $V_\text{core}$}
            {
                let $b_\text{core}^{(v)}$ denote the bitstring $b_\text{core}$ with the bit corresponding to vertex $v$ being flipped\;
                \If{$\cut_g(b_\text{ext} \cup b_\text{core}) \geq \max\{\cut_g(b_\text{ext} \cup b_\text{core}^{(v)}) : v \in V_\text{core}\}$}
                {
                 append the colored marked graph $(g, e, b_\text{ext} \cup b_\text{core})$ to $\mathcal{H}$\;
                }
            }
        }
    }
   remove duplicate colored edge-neighborhoods that are of the same type in $\mathcal{H}$\;
   \Return{$\mathcal{H}$}
    
\end{algorithm}

Running the procedure above with $k=3$ and $p=1$ yields the colored edge-neighborhoods found in Figure \ref{fig:optimalLabeledSubgraphs}. In the above algorithm, we have unions of bitstrings of the form $b_1 \cup b_2$ where $b_1$ and $b_2$ act on \emph{disjoint} sets of vertices $V_1$ and $V_2$. Formally, $(b_1 \cup b_2): (V_1 \cup V_2) \to \{0,1\}$ such that $(b_1 \cup b_2)(v) = b_1(v)$ if $v \in V_1$ and $(b_1 \cup b_2)(v) = b_2(v)$ if $v \in V_2$. Also note that the function $\cut_g(\cdot)$ is with respect to the non-colored edge neighborhood $g$, not the overall graph $G$ that these neighborhoods may possibly live in; in the context of the overall graph, $\cut_g(\cdot)$ can be viewed as the \emph{local} Max-Cut of that neighborhood. There is room for further runtime optimization of Algorithm \ref{alg:coloredEdgeNeighborhoods}, i.e., due to bit-flip symmetry of \MC{} and due to the symmetry of the underlying non-colored edge-neighborhoods themselves, it is not necessary to run the entire procedure for each choice of subcoloring of the exterior vertices. Step 8 can also be sped up by caching the cut values each time we calculate them. At depth $p=1$ and graph regularity $k=3$, Algorithm \ref{alg:coloredEdgeNeighborhoods} runs relatively quickly so such optimizations are not necessary; however, one may find such optimizations useful for higher values of $p$ or $k$.

The procedure for finding all of the colored graph structures is similar. The only difference is that we use the (non-colored) graph structures in the input of Algorithm \ref{alg:coloredEdgeNeighborhoods} (instead of the non-colored edge neighborhoods) and there is no notion of a central or marked edge for the graph structures. Doing this at $p=1$ and $k=3$ yields the colored graph structures found in Figure \ref{fig:coloredGraphStructures}.

\subsection{Optimization Details}
\label{sec:optimizationDetails}
Throughout this work, we bound the approximation ratio by numerically calculating some optimization problem. These optimization problems are very non-trivial, in particular, they contain nested minimizations and maximizations. It is important to consider the order of these minimizations/maximizations since they have the possibility of yielding different optimal values; moreover, each choice of ordering corresponds to certain assumptions that are being made regarding the warm-started QAOA algorithm. Before continuing, we encourage the reader to review the notations and definitions regarding approximation ratio in Sections \ref{sec:approximationRatio} and \ref{sec:warmstarted_QAOA_AR}.

In general, for a function $f: X \times Y \to \mathbb{R}$, the following inequality is well-known:
\begin{equation}\label{eqn:minMaxInequality}\max_{x \in X}\min_{y\in Y} f(x,y) \leq \min_{y \in Y}\max_{x \in X} f(x,y).\end{equation}

For the above inequality, we will frequently refer to the left and right sides as a max-min and min-max problem respectively. There are various \emph{minimax} theorems that show that, in certain contexts \cite{v1928theorie,du1995minimax,sion1958general, parthasarathy1970games}, the above inequality actually becomes equality; however, equality is not guaranteed for general functions.

As a result of Equation \ref{eqn:minMaxInequality} and the earlier calculations in Section \ref{sec:graphStructures}, we have the following relationships for standard depth-1 QAOA:
\begin{equation}\label{eqn:minMaxInequalityStandardQAOA}\alpha_\text{QAOA} \geq \min_{(r,s,t) \in \mathcal{P}} \max_{\gamma,\beta} \frac{\frac{1}{m}F_{r,s,t}(\gamma,\beta)}{1-s-t} \geq  \max_{\gamma,\beta} \min_{(r,s,t) \in \mathcal{P}} \frac{\frac{1}{m}F_{r,s,t}(\gamma,\beta)}{1-s-t}.\end{equation}

We take a moment to interpret the meaning of the expression on the right-hand side of Equation \ref{eqn:minMaxInequalityStandardQAOA}. Let $\alpha^{\gamma,\beta}_\text{QAOA}$ denote the approximation-ratio that is possible by running the depth-1 standard QAOA circuit with $\gamma$ and $\beta$ and not running any classical outer loop; we let $\alpha^{\gamma,\beta}_\text{QAOA}(G)$ be the corresponding instance-specific approximation ratio on a specific instance $G$. We let $r(G), s(G), t(G)$ be the values of $r,s,t$ that correspond to the graph $G$, then, for each choice of $\gamma,\beta$, we have the following bound on the instance-specific approximation ratio:
$$\alpha^{\gamma,\beta}_\text{QAOA}(G) \geq \frac{\frac{1}{m}F_{r(G),s(G),t(G)}(\gamma,\beta)}{1-s(G)-t(G)}.$$

Taking a minimum over
all possibles graphs $G \in \mathcal{G}$ (where $\mathcal{G}$ is the set of 3-regular graphs in this case) yields the following:
$$\alpha^{\gamma,\beta}_\text{QAOA} = \min_{G \in \mathcal{G}}\alpha^{\gamma,\beta}_\text{QAOA}(G) \geq \min_{G \in \mathcal{G}} \frac{\frac{1}{m}F_{r(G),s(G),t(G)}(\gamma,\beta)}{1-s(G)-t(G)}  \geq \min_{(r,s,t) \in \mathcal{P}} \frac{\frac{1}{m}F_{r,s,t}(\gamma,\beta)}{1-s-t},$$
where the last inequality holds since points in the constraint polytope $\mathcal{P}$ may possibly correspond to values of $r,s,t$ that are not associated with any graph (on the other hand, $\mathcal{P}$ was specifically constructed in a way so that for any graph $G$, we have $(r(G), s(G), t(G)) \in \mathcal{P}$).

Taking a maximum on both sides of the above inequality yields:
\begin{equation}
\label{eqn:fixedParamsARMoreThanMaxMin}
\max_{\gamma,\beta} \alpha^{\gamma,\beta}_\text{QAOA} \geq \max_{\gamma,\beta} \min_{(r,s,t) \in \mathcal{P}} \frac{\frac{1}{m}F_{r,s,t}(\gamma,\beta)}{1-s-t},
\end{equation}
in other words, the max-min problem is a lower-bound for the best possible approximation ratio for standard QAOA with some choice of \emph{fixed} variational parameters. A priori, one should not assume that $\alpha_\text{QAOA} = \max_{\gamma,\beta} \alpha_\text{QAOA}(\gamma, \beta)$, since for a \emph{specific} problem instance, the best choice of $\gamma$ and $\beta$ may vary and there might not be a \emph{single} choice of $\gamma,\beta$ that does well across most instances. In fact, we have the following relation between $\alpha_\text{QAOA}$ and $\max_{\gamma,\beta}\alpha^{\gamma,\beta}_\text{QAOA}$:

\begin{prop}
\label{thm:approxRatioFixedParams}
The following holds: $\alpha_\text{QAOA} \geq \max_{\gamma,\beta} \alpha^{\gamma,\beta}_\text{QAOA}$.
\end{prop}
\begin{proof}
Unpacking definitions and utilizing the Equation \ref{eqn:minMaxInequality} yields,
\begin{align*}\alpha_\text{QAOA} &= \min_{G \in \mathcal{G}} \alpha_\text{QAOA}(G) \\
&=  \min_{G \in \mathcal{G}}  \frac{\max_{\gamma,\beta} F_{R(G),S(G),T(G)}(\gamma,\beta)}{\mc(G)} \\
&\geq  \max_{\gamma,\beta} \min_{G \in \mathcal{G}} \frac{F_{R(G),S(G),T(G)}(\gamma,\beta)}{\mc(G)} \\
& = \max_{\gamma,\beta}\min_{G \in \mathcal{G}} \alpha^{\gamma,\beta}_\text{QAOA}(G)  \\
&= \max_{\gamma,\beta} \alpha^{\gamma,\beta}_\text{QAOA}.
\end{align*}
\end{proof}

We remark that for a fixed choice of $\gamma$ and $\beta$, the inner minimization $\min_{(r,s,t) \in \mathcal{P}} \frac{\frac{1}{m}F_G(\gamma,\beta)}{1-s-t}$ of the max-min problem is computationally easy to solve exactly. To see this, we first expand $\frac{1}{m}F_G(\gamma,\beta)$ using Equation \ref{eqn:expectedCutAsFunctionOfGraphStructures} to get the equivalent minimization:
$$\min_{(r,s,t) \in \mathcal{P}} \frac{sf_{g_4}(\gamma,\beta) + (4s+3t)f_{g_5}(\gamma,\beta) + rf_{g_6}(\gamma,\beta)}{1-s-t}.$$

Rearranging the terms yields:
$$\min_{(r,s,t) \in \mathcal{P}} \frac{\Big(f_{g_4}(\gamma,\beta) + 4f_{g_5}(\gamma,\beta)\Big)s + \Big(3f_{g_5}(\gamma,\beta)\Big)t + \Big(f_{g_6}(\gamma,\beta)\Big)r}{1-s-t},$$

in other words, for a fixed choice of $\gamma$ and $\beta$, the coefficients in front of $r,s,t$ in the numerator of the above equation are simply constants. The above formulation of the minimization is in the form of a linear-fractional program (LFP) as the numerator and denominator are both linear functions in the variables being optimized and the constraints amongst $r,s,t$ are linear. LFPs are a generalization of linear programming (LP). Following the standard Charnes-Cooper transformation \cite{charnesCooper62}, in order to solve the inner-minimization, we calculate $f_{g_4}(\gamma,\beta), f_{g_5}(\gamma,\beta), f_{g_6}(\gamma,\beta)$ to find the coefficients for numerator of the LFP, we then transform the above LFP into an LP, solve the LP exactly, and map the solutions of the LP back to the LFP.

Let $\mathbf{Alg}_\text{max-min}$ be the optimal value obtained from numerically calculating the max-min problem for QAOA, using a grid and/or local search for the outer maximization, and solving the inner minimization exactly via the Charnes-Cooper transformation. Assuming that $f_{g_4}(\gamma,\beta), f_{g_5}(\gamma,\beta), f_{g_6}(\gamma,\beta)$ can be calculated exactly for any choice of $\gamma,\beta$, we have the following inequality:
\begin{equation}\label{eqn:maxMinDiscretizationInequality}\mathbf{Alg}_\text{max-min} \leq \max_{\gamma,\beta} \min_{(r,s,t) \in \mathcal{P}} \frac{\frac{1}{m}F_G(\gamma,\beta)}{1-s-t},\end{equation}
which approaches equality as the grid gets increasingly fine.

The above inequality is important as it ensures that the result of our numerical optimization does not accidentally overestimate the true approximation ratio of QAOA; we formally state this in Theorem \ref{thm:underestimateApproxRatio} below.

\begin{theorem}
\label{thm:underestimateApproxRatio}
    The following inequalities hold:
    $$\alpha_\text{QAOA} \geq \max_{\gamma,\beta}\alpha^{\gamma,\beta}_\text{QAOA} \geq \max_{\gamma,\beta} \min_{(r,s,t) \in \mathcal{P}} \frac{\frac{1}{m}F_G(\gamma,\beta)}{1-s-t} \geq \mathbf{Alg}_\text{max-min}.$$
\end{theorem}
\begin{proof}
    The sequence of inequalities follows from Proposition \ref{thm:approxRatioFixedParams}, Equation \ref{eqn:fixedParamsARMoreThanMaxMin}, and Equation \ref{eqn:maxMinDiscretizationInequality} respectively.
\end{proof}

While the min-max optimization has the potential to produce a better lower bound for QAOA (see Equation \ref{eqn:minMaxInequalityStandardQAOA}), an inequality such as what is seen in Theorem \ref{thm:underestimateApproxRatio} cannot as easily be obtained as the inner optimization problem can no longer be solved exactly. In particular, due to the simultaneous over- and under-estimation effects of finding approximately optimal values of the outer and inner optimization problems (respectively), one can not readily determine if the overall result of the numerical optimization is an underestimate or an overestimate of the true objective value of the min-max problem for QAOA.

For our numerical optimization of the max-min problem, for the outer maximization, we performed a uniform $20 \times 20$ grid-search over the parameter search space $(\gamma,\beta) \in [0, \pi]^2 \times [-\pi/4, \pi/4]$ (see \cite{zhou2020quantum} regarding the periodicity of the variational parameters on regular graphs); afterwards, we used the \texttt{optimize.fmin} function provided by \texttt{Scipy} \cite{2020SciPy-NMeth} initialized on the best grid point. We solve the inner minimization exactly using the Charnes-Cooper transformation described earlier. With this optimization procedure, we obtain $\mathbf{Alg}_\text{max-min} \geq 0.6924$ with optimal variational parameters $(\gamma,\beta) = (\gamma^*, \beta^*)$ with $\gamma^* = 0.616, \beta^* = 0.393$. This means that using these values of $\gamma^*,\beta^*$, for \emph{any} instance of a 3-regular graph $G$, standard QAOA will return a cut with at least $0.6924\cdot \mc(G)$ edges \emph{without having to further optimize $\gamma$ or $\beta$}.

Surprisingly, the value of $0.6924$ is \emph{also} the value obtained by Farhi et al. \cite{farhi2014quantum} for their numerical optimization of the min-max optimization problem. In particular, they numerically find that the min-max problem is minimized at $s=t=0$; if one assumes these are the true minimizers, then Farhi et al.'s result is consistent with standard QAOA's approximation ratio on \emph{triangle-free} regular graphs which is known to be $\frac{1}{2}+\frac{1}{3\sqrt{3}} \approx 0.6924$ due to Wang et al. \cite{wang2018quantum}. This seems to suggest that 

$$\min_{(r,s,t)\in \mathcal{P}} \max_{\gamma,\beta} \frac{\frac{1}{m}F_G(\gamma,\beta)}{1-s-t}  =  \max_{\gamma,\beta} \min_{(r,s,t) \in \mathcal{P}} \frac{\frac{1}{m}F_G(\gamma,\beta)}{1-s-t},$$

and possibly that,

$$\alpha_\text{QAOA} = \max_{\gamma,\beta} \alpha^{\gamma,\beta}_\text{QAOA};$$

however, we do not have an analytical proof of such a claim.

\sloppy
All of the above results also applies to warm-starts and the min-max and max-min problems of {$ \min_{(\mathbf{r},\mathbf{s},\mathbf{t}) \in \mathcal{P(\kappa)}} \max_{\gamma,\beta} \frac{\frac{1}{m}F_{\mathbf{r},\mathbf{s},\mathbf{t}}(\gamma,\beta,\theta)}{1-s-t}$} and $\max_{\gamma,\beta} \min_{(\mathbf{r},\mathbf{s},\mathbf{t}) \in \mathcal{P(\kappa)}}  \frac{\frac{1}{m}F_{\mathbf{r},\mathbf{s},\mathbf{t}}(\gamma,\beta,\theta)}{1-s-t}$ respectively (and also the corresponding min-max and max-min problems where $\mathcal{P}(\kappa)$ is replaced with $\mathcal{P}'(\kappa)$ and $\mathbf{s}$ and $\mathbf{t}$ are replaced with $\mathbf{0}$). In particular, letting $\mathbf{Alg}_\text{max-min}(\theta,\kappa)$ and $\mathbf{Alg}'_\text{max-min}(\theta,\kappa)$ denote the output of numerically optimizing $\max_{\gamma,\beta} \min_{(\mathbf{r},\mathbf{s},\mathbf{t}) \in \mathcal{P(\kappa)}}  \frac{\frac{1}{m}F_{\mathbf{r},\mathbf{s},\mathbf{t}}(\gamma,\beta,\theta)}{1-s-t}$ and $\max_{\gamma,\beta} \min_{\mathbf{r} \in \mathcal{P'(\kappa)}}  \frac{1}{m}F_{\mathbf{r},\mathbf{0},\mathbf{0}}(\gamma,\beta,\theta)$ where the inner minimization is solved exactly using the Charnes-Cooper transformation discussed earlier, then Theorem \ref{thm:underestimateApproxRatio} can be extended to warm-started QAOA as follows.

\begin{theorem}
\label{thm:underestimateApproxRatioWarm}
    If there exists an algorithm of type $\mathcal{A}_\kappa$, then for all $\theta\in \mathbb{R}$, the following inequalities hold:
    $$\alpha_{(\mathcal{A}_\kappa+\text{QAOA})_\theta} \geq \max_{\gamma,\beta}\alpha^{\gamma,\beta}_{(\mathcal{A}_\kappa+\text{QAOA})_\theta} \geq \max_{\gamma,\beta} \min_{(\mathbf{r},\mathbf{s},\mathbf{t}) \in \mathcal{P(\kappa)}}  \frac{\frac{1}{m}F_{\mathbf{r},\mathbf{s},\mathbf{t}}(\gamma,\beta,\theta)}{1-s-t} \geq \mathbf{Alg}_\text{max-min}(\theta,\kappa).$$
    Similarly, if there exists an algorithm of type $\mathcal{A}'_\kappa$, then for all $\theta\in \mathbb{R}$, the following inequalities hold:
    $$\alpha_{(\mathcal{A}'_\kappa+\text{QAOA})_\theta} \geq \max_{\gamma,\beta}\alpha^{\gamma,\beta}_{(\mathcal{A}'_\kappa+\text{QAOA})_\theta} \geq \max_{\gamma,\beta} \min_{\mathbf{r} \in \mathcal{P'(\kappa)}}  \frac{1}{m}F_{\mathbf{r},\mathbf{0},\mathbf{0}}(\gamma,\beta,\theta) \geq \mathbf{Alg}'_\text{max-min}(\theta,\kappa).$$
\end{theorem}
\begin{proof}
    The same arguments used to prove Theorem \ref{thm:underestimateApproxRatio} can be used for this theorem.
\end{proof}

For various values of $\theta$ and $\kappa$, we obtain $\mathbf{Alg}_\text{max-min}(\theta,\kappa)$ using the similar method used to obtained $\mathbf{Alg}_\text{max-min}$ for standard QAOA described earlier, i.e., we performed a $20 \times 20$ grid-search over the parameter search space $(\gamma,\beta) \in [0, \pi]^2$ (see Supplementary Materials regarding periodicity of the variational parameters in QAOA) and afterwards, we used the previously mentioned \texttt{optimize.fmin} to further improve from all locally maximal grid points\footnote{For a 2-dimensional grid, we say that a grid-point is a \emph{locally maximal grid point} if the function value at that grid point is strictly higher than the function value at any of the 8 neighboring grid points. Effectively, given a sufficiently fine grid, these locally maximal grid points correspond to grid points that are near local maximums of the functions (off the grid).} if  and return the best set of parameters found. We solve the inner minimization exactly using the Charnes-Cooper transformation described earlier. We also compared the values of $\mathbf{Alg}_\text{max-min}(\theta,\kappa)$ and $\mathbf{Alg}'_\text{max-min}(\theta,\kappa)$ to numerically obtained values of the corresponding min-max problems for warm-started QAOA and for all values of $\theta$ and $\kappa$ that we tested, we found that these values were nearly equal to one another, suggesting that 

\begin{equation}\label{eqn:maxMinEqualsMinMaxWarm}\max_{\gamma,\beta} \min_{(\mathbf{r},\mathbf{s},\mathbf{t}) \in \mathcal{P(\kappa)}}  \frac{\frac{1}{m}F_{\mathbf{r},\mathbf{s},\mathbf{t}}(\gamma,\beta,\theta)}{1-s-t} \approx  \min_{(\mathbf{r},\mathbf{s},\mathbf{t}) \in \mathcal{P(\kappa)}}  \max_{\gamma,\beta} \frac{\frac{1}{m}F_{\mathbf{r},\mathbf{s},\mathbf{t}}(\gamma,\beta,\theta)}{1-s-t},\end{equation}
and
\begin{equation}\label{eqn:maxMinEqualsMinMaxWarm2}\max_{\gamma,\beta} \min_{\mathbf{r} \in \mathcal{P'(\kappa)}}  \frac{1}{m}F_{\mathbf{r},\mathbf{0},\mathbf{0}}(\gamma,\beta,\theta) \approx  \min_{\mathbf{r} \in \mathcal{P'(\kappa)}}  \max_{\gamma,\beta} \frac{1}{m}F_{\mathbf{r},\mathbf{0},\mathbf{0}}(\gamma,\beta,\theta);\end{equation}

however, we do not have an analytical proof of such claims. Note that for warm-started QAOA, solving the min-max problem is much more difficult since the inner optimization can no longer be solved exactly and the outer optimization requires searching a 17-dimensional space meaning that a fine grid search is computationally impractical. Thus, it may potentially be the case that the minimizers and maximizers that we numerically obtained for the min-max problem (for warm-started QAOA) are not the true minimizers and maximizers and thus, the reader should consider Equations \ref{eqn:maxMinEqualsMinMaxWarm} and \ref{eqn:maxMinEqualsMinMaxWarm2} with a grain of salt.

    \section{Discussion and Conclusion}
    \label{sec:discussionConclusion}
    In this work, we considered warm-started \mc{} QAOA on 3-regular graphs, parameterized by an initialization angle $\theta$, where the warm-start state is generated from a single (classically-obtained) cut that is locally optimal with respect to single bitflips. This restriction on the warm-start improves the approximation ratios that are possible (compared to warm-starts generated from arbitrary cuts) while also greatly simplifying the analysis due to the reduced number of possible colored edge neighborhoods and graph structures. Moreover, this work shows that if the classical algorithm used to generate the warm-start has guarantees on the fraction of the total number of edges across the cuts it produces (compared to the total number of edges in the graph), then QAOA warm-started with such cuts yields increasingly better approximation ratios as the fraction of edges cut increases. By adding restrictions to the 3-regular instances allowed, certain algorithms (such as those in \cite{bondy1986largest} and \cite{zhu2009bipartite}), cut a larger fraction of the total number of edges, yielding better approximation ratios for warm-started QAOA.

 %
Our lower bounds on the approximation ratio $\mathbf{LB}_\kappa(\theta)$ and $\mathbf{LB}'_\kappa(\theta)$ never exceeded that of standard QAOA or the corresponding classical algorithm; in other words, for the values of $\theta,\kappa$ considered in this work, we observed that $\mathbf{LB}_\kappa(\theta) \leq \max(\mathbf{LB}_\kappa(0), \mathbf{LB}_\kappa(\pi/2)).$ This suggests that, in the worst-case, warm-started \mc{} QAOA on 3-regular graphs has no provable advantage compared to standard QAOA or simply performing the classical algorithm used to obtain the warm-start, at least at depth $p=1$. It is still unclear if an advantage can be gained by considering higher depths.

    We would also like to remind the reader that the bounds on the approximation ratio were obtained by solving a max-min optimization problem (i.e. a nested optimization problem with an outside maximization and inside minimization). This is was done since (as is discussed in Section \ref{sec:optimizationDetails}) the max-min optimization is easier to solve in a sense (compared to the corresponding min-max problem) and any numerically obtained objective values are a lower bound on the true optimal objective value (meaning we do not accidentally overestimate the approximation ratio). Moreover, for each initialization angle $\theta$, solving the max-min problem yields optimal $\theta$-dependent variational parameters $\gamma^*(\theta), \beta^*(\theta)$. These variational parameters only need to be calculated \emph{once} (for each value of $\theta$); afterwards, for any instance, these same variational parameters can be used in the warm-started QAOA circuit and, \emph{without any further optimization of $\gamma$ and $\beta$}, the instance-specific approximation ratio obtained would be at least the bound determined by the max-min optimization. Theoretically, the min-max optimization problem would yield bounds on the approximation ratios that were at least as good as the corresponding max-min optimization; however, for various values of $\theta$ and $\kappa$, we found (numerically) that the bounds obtained from both the max-min and min-max optimizations were in fact equal. In future work, it would be interesting to see if this fact can be extended to higher depth circuits or to other ansatz or combinatorial optimization problems; moreover, an analytical proof of such a phenomena would be highly valued.

    We remark that the behavior of the approximation ratio curves with changing initialization angle $\theta$ as seen in Figure \ref{fig:plotAR_varyingConstraints} may be due to the structure of 3-regular graphs. In the work of Egger et al., they test the same kind of warm-start in this work on 10 weighted complete graphs with 30 nodes and for such graphs, the median (instance-specific) approximation ratio decreases nearly monotonically with increased initialization angle. Interestingly, the shape of our approximation ratio curves looks similar to another curve obtained by Egger et al. (Figure 7 of \cite{egger2021warm}) where a modified mixer is used on the weighted complete graphs previously mentioned; both curves have a peak\footnote{Instead of using a notation of initialization angle $\theta$ as is done in this work, Egger et al. \cite{egger2021warm} use a notion of regularization parameter $\varepsilon \in [0,0.5]$. The relationship between these is given by $\theta = 2\arcsin(\sqrt{\varepsilon}).$} at $\theta = 60^\circ$ due to both kinds of warm-started QAOA being able to recover the warm-start cut with a suitable choice of variational parameters.

    In general, the results of this work suggest $\theta = 60^\circ$ will usually be a suitable choice for warm-started QAOA since it is able to recover the cut (with a specific choice of variational parameters) and for some instances, there are potentially different choice of variational parameters could potentially yield a better cut; as long as the cut satisfies $\cut(b)/\mc(G) > 0.6924$, then warm-started QAOA at $\theta=60^\circ$ is guaranteed to yield a better approximation ratio than standard QAOA. For future work, it would be interesting to see if $\theta=60^\circ$ is a good choice for typical instances aside from the worst-case; more generally, it would be interesting to see how various graph properties affect the behavior of warm-started QAOA.

    Comparing the bounds on the approximation ratio for depth-0 and depth-1 warm-started QAOA, it appears that there is very little change when the initialization angle $\theta$ is near zero. While this work only considers warm-starts that are a fixed angle from the poles of the Bloch sphere, we anticipate that this phenomenon more generally applies to the case where each qubits can have varying qubit positions from one another but are still within some angle away from the poles of the Bloch sphere.

    It is still not clear how the analysis of the relatively-simple warm-starts considered in this work can be extended to complex warm-starts such as those proposed in \cite{tate2023bridging} and \cite{Tate2023warmstartedqaoa}. However, the techniques in this work, can be potentially extended to more general warm-starts in a few ways. First, instead of forcing the two initial qubit positions to be the same angle from the poles, we could consider other relative positions for the qubits. Second, we could consider $k>2$ initial qubit positions which would correspond to $k$-colorings of the edge neighborhoods; for example, with $k=3$ qubit positions, we could consider two qubit positions some angle from the poles with the third position being at the equator. Note that increasing $k$ causes the number of colored edge-neighborhoods to increase dramatically meaning the analysis required would be more computationally intensive.  Third, for each initial qubit position, the mixer is defined as a rotation about the axis defined by that position; however, we could instead consider rotation about some different axis (as long as the axis chosen is the same for qubits at the same initial position, otherwise more colors would be needed to color the edge neighborhoods). Such unaligned mixers will break QAOA's connection to quantum adiabatic computing (which guarantees that QAOA will reach the optimal solution with large enough circuit depth and correct choice of parameters), but perhaps for low-depth QAOA, better guarantees can be found by dropping such a connection (although there is evidence that suggests that it is best to use aligned mixers \cite{he2023alignment}). Warm-started QAOA with other mixers and phase separators, such as the those proposed in \cite{wang2020x,bartschi2020grover,golden2021threshold} may be interesting to consider; however, it is not immediately clear if a decomposition result similar to Theorem \ref{thm:coloredSubgraphDecomposition} can be obtained with such mixers since the usual commutativity arguments would no longer hold. Finally, instead of coloring vertices in the colored edge-neighborhoods, one could instead color the edges; doing so would provide a method of analyzing standard \mc{} QAOA on graphs whose edge weights come from a small set of possible weights.

    This work carefully studied depth-1 warm-started \mc{} QAOA on 3-regular graphs; however, the techniques used in this work could in principle be extended to any fixed depth and graphs with bounded degree. Unfortunately, an increase in the circuit depth and/or graph degree causes the number of colored edge-neighborhoods to explode, so an analysis with high circuit depth or large bounded degree would be difficult with the techniques used in this work.

\section{CRediT Authorship Contribution Statement}

\textbf{Reuben Tate:} Conceptualization, Methodology, Software, Formal analysis, Investigation, Writing - Original Draft, Writing - Review \& Editing, Visualization.   \textbf{Stephan Eidenbenz:} Supervision, Writing - Review \& Editing.

\bibliographystyle{elsarticle-num}
\bibliography{references}

\appendix

\section{Bit-Flip QAOA Symmetry}
    \label{sec:bitFlipSymmetry}

    We prove Theorem \ref{thm:bitflipSameExpectedCut}, which we restate here for convenience.
    \begin{theorem}
    Let $G=(V,E)$ be a graph. Let $b \in \{0,1\}^n$ be a bitstring, let $\bar{b}$ be its bitwise-negation, let $\theta \in \RR$, let $\gamma$ and $\beta$ be any choice of variational parameters, and let $p$ be any non-negative circuit depth. Then, for the \MC{} problem on $G$, and for both choices of initialization ($\ket{b_\theta}$ and $\ket{\bar{b}_\theta}$), we have that the expected cut size obtained by warm-started QAOA is the same, i.e., 
    $$F^{(p)}(\gamma,\beta, \ket{b_\theta}) = F^{(p)}(\gamma, \beta, \ket{\bar{b}_\theta}).$$

\end{theorem}
    Let $$\textbf{X} = \prod_{j=1}^n X_j.$$
    
    In order to prove bit-flip symmetry result in the main paper, we first make the following key observation (Lemma \ref{thm:bitflipPauliX}) regarding the relationship between $\ket{b_\theta}$ and $\ket{\bar{b}_\theta}$.
    \begin{lemma}
    \label{thm:bitflipPauliX}
    Let $b \in \{0,1\}^n$ be a bitstring and let $\theta \in \RR$. Then,
    $$\ket{b_\theta} = \textbf{X}\ket{\bar{b}_\theta} \quad \text{and} \quad \ket{\bar{b}_\theta} = \textbf{X}\ket{b_\theta}.$$
    \end{lemma}
    \begin{proof}
        This immediately follows from the fact that $X\ket{0_\theta} = \ket{1_\theta}$ and $X\ket{1_\theta} = \ket{0_\theta}$ (which can be easily checked using the definition of $\ket{0_\theta}$ and $\ket{1_\theta}$.
    \end{proof} 

    Next, in Lemma \ref{thm:bitflipRelationOnMixer}, we see how $\mathbf{X}$ influences the mixer $B_{\ket{b_\theta}}$.
    \begin{lemma}
    \label{thm:bitflipRelationOnMixer}
        Let $b \in \{0,1\}^n$ be a bitstring and let $\theta \in \RR$. Then,
        $$\mathbf{X}B_{\ket{b_\theta}} = B_{\ket{\bar{b}_\theta}}\mathbf{X} \quad \text{and} \quad \mathbf{X}B_{\ket{\bar{b}_\theta}} = B_{\ket{b_\theta}}\mathbf{X},$$
    \end{lemma}
    \begin{proof}
       Let $(x_j, y_j, z_j)$ be the Cartesian coordinates of the $j$th qubit's position in $\ket{b_\theta}$. By construction, we have that $y_j = 0$ for all $j$. Moreover, also by construction, we have that the Cartesian coordinates of the $j$th qubit's position in $\ket{\bar{b}_\theta}$ is given by $(x_j, 0, -z_j)$. Thus,
       \begin{align*}
           \mathbf{X}B_{\ket{b_\theta}} & = \mathbf{X}\left(\sum_{j=1}^n x_jX_j+z_jZ_j\right) \\
            &=\sum_{j=1}^n x_j\mathbf{X}X_j+z_j\mathbf{X}Z_j \\
            &=\sum_{j=1}^n x_jX_j\mathbf{X}-z_jZ_j\mathbf{X} \\
            &= \left(\sum_{j=1}^n x_jX_j-z_jZ_j\right)\mathbf{X} \\
            &= B_{\ket{\bar{b}_\theta}}\mathbf{X}.
       \end{align*}

    The second equality $\mathbf{X}B_{\ket{\bar{b}_\theta}} = B_{\ket{b_\theta}}\mathbf{X}$ follows similarly.
    \end{proof}

    Next, if $C$ is the \MC{} Hamiltonian, the relation in Lemma \ref{thm:bitflipRelationOnCostHam} is well-known (we refer the reader to \cite{shaydulin2021classical} for a more general overview of the effects of symmetry on QAOA). For the sake of completion, we provide a proof.
    \begin{lemma}
    \label{thm:bitflipRelationOnCostHam}
    Let $G=(V,E)$ be a graph and let $C$ be its corresponding cost Hamiltonian for \MC. Then,
    $$\mathbf{X}C = C\mathbf{X}.$$
    \end{lemma}
    \begin{proof}
        It suffices to prove that $\mathbf{X}C_e = C_e\mathbf{X}$ for all $e \in E$, the result then follows by linearity as $C = \sum_{e \in E}C_e$.

        Let $e = (u,v) \in E$. Then,
        \begin{align*}
            \mathbf{X}C_e &= \mathbf{X}(I-Z_uZ_v)\\
            &= \mathbf{X} - \mathbf{X}Z_uZ_v \\
            &= \mathbf{X} + Z_u\mathbf{X}Z_v \\
            &= \mathbf{X} - Z_uZ_v\mathbf{X}\\
            &= (I - Z_uZ_v)\mathbf{X}\\
            &= C_e\mathbf{X}.
        \end{align*}
    \end{proof}
    Lemma \ref{thm:bitflipRelationOnCostHam} intuitively makes sense classically since for a cut corresponding to $\ket{b}$ (where $b$ is a bitstring), we have that the cut given by $\ket{\bar{b}} = \mathbf{X}\ket{b}$ corresponds to simply swapping the partitions of the vertices (given by $b$) and hence $b$ and $\bar{b}$ cut the same number of edges.
    
    Lemma \ref{thm:extendToUnitaryCommutation} (below) together with Lemmas \ref{thm:bitflipRelationOnMixer} and \ref{thm:bitflipRelationOnCostHam}, give us Corollary \ref{thm:bitflipAppliedToUnitaries} which shows how $\mathbf{X}$ influences the corresponding unitary operations $U(C, \gamma)$ and $U(B_{\ket{b_\theta}}, \beta)$.
    \begin{lemma}
    \label{thm:extendToUnitaryCommutation}
        Let $A,B,B'$ be square matrices of the same size and let $\gamma \in \RR$. If $AB = B'A$, then $AU(B,\gamma) = U(B', \gamma)A$.
    \end{lemma}
    \begin{proof}
        First, observe that if $AB = B'A$, then for any non-negative integer $k$, we have that $AB^k = (B')^kA$. Thus,
        $$AU(B, \gamma) = A \sum_{k=1}^\infty \frac{(-i\gamma B)^k}{k!} =  \sum_{k=1}^\infty \frac{(-i\gamma)^k AB^k}{k!} $$$$= \sum_{k=1}^\infty \frac{(-i\gamma)^k (B')^kA}{k!} = \left(\sum_{k=1}^\infty \frac{(-i\gamma B')^k}{k!}\right)A = U(B',\gamma)A.$$
    \end{proof}
    \begin{corollary}
    \label{thm:bitflipAppliedToUnitaries}
     Let $G=(V,E)$ be a graph and let $C$ be its corresponding cost Hamiltonian for \MC. Let $b \in \{0,1\}^n$ be a bitstring and let $\gamma \in \RR$. Then
        $$\mathbf{X}U(C,\gamma) = U(C, \gamma)\mathbf{X}$$
        and
        $$\mathbf{X}U(B_{\ket{b_\theta}},\beta) = U(B_{\ket{\bar{b}_\theta}},\beta)\mathbf{X} \quad  \text{and} \quad \mathbf{X}U(B_{\ket{\bar{b}_\theta}},\beta) = U(B_{\ket{b_\theta}},\beta)\mathbf{X}.$$
    \end{corollary}
    \begin{proof}
        This follows immediately from Lemmas \ref{thm:bitflipRelationOnMixer}, \ref{thm:bitflipRelationOnCostHam}, and \ref{thm:extendToUnitaryCommutation}. 
    \end{proof}

    Next, in Proposition \ref{thm:bitflitQAOAOutputStates}, we show that the output state of QAOA initialized with $\ket{b_\theta}$ and $\ket{\bar{b}_\theta}$ (and using aligned mixers) are nearly the same, just with all the bits being flipped.

    \begin{prop}
    \label{thm:bitflitQAOAOutputStates}
        Let $G=(V,E)$ be a graph with corresponding cost Hamiltonian $C$ for \MC. Let $b \in \{0,1\}^n$ be a bitstring, let $\theta \in \RR$, let $p$ be a circuit depth, and let $(\gamma,\beta)$ be any choice of variational parameters. Then,
        $$\ket{\psi_p(\gamma,\beta, \ket{b_\theta})} = \mathbf{X}\ket{\psi_p(\gamma,\beta, \ket{\bar{b}_\theta})}.$$
    \end{prop}
    \begin{proof}
        We proceed by induction on the circuit depth $p$. When $p=0$, we have,
        $$\ket{\psi_0(\gamma,\beta, \ket{b_\theta})} = \ket{b_\theta} = \mathbf{X}\ket{\bar{b}_\theta} =\mathbf{X}\ket{\psi_0(\gamma,\beta, \ket{\bar{b}_\theta})},$$
        where the middle equality holds as a result of Lemma \ref{thm:bitflipPauliX}. Let $\beta' = (\beta_1, \dots, \beta_{p-1})$ and let $\gamma' = (\gamma_1, \dots, \gamma_{p-1})$. Then for $p>0$,
        \begin{align*}
            \ket{\psi_p(\gamma,\beta, \ket{b_\theta})} &= U(B_{\ket{b_\theta}}, \beta_p)U(C, \gamma_p)\ket{\psi_{p-1}(\gamma',\beta', \ket{b_\theta})} \\
            &= U(B_{\ket{b_\theta}}, \beta_p)U(C, \gamma_p)\Big(\mathbf{X}\ket{\psi_{p-1}(\gamma',\beta', \ket{\bar{b}_\theta})}\Big) \tag{induction} \\
            &= U(B_{\ket{b_\theta}}, \beta_p)\mathbf{X}U(C, \gamma_p)\ket{\psi_{p-1}(\gamma',\beta', \ket{\bar{b}_\theta})} \tag{Corollary \ref{thm:bitflipAppliedToUnitaries}} \\
            &= \mathbf{X} U(B_{\ket{\bar{b}_\theta}}, \beta_p)U(C, \gamma_p)\ket{\psi_{p-1}(\gamma',\beta', \ket{\bar{b}_\theta})} \tag{Corollary \ref{thm:bitflipAppliedToUnitaries}} \\
            &= \mathbf{X}\ket{\psi_p(\gamma,\beta, \ket{\bar{b}_\theta})}.
        \end{align*}
    \end{proof}

    We now have everything needed to prove the result.
    \begin{proof}
        Observe,
        \begin{align*}
            F^{(p)}(\gamma,\beta, \ket{b_\theta}) &= \bra{\psi_p(\gamma,\beta, \ket{b_\theta})}C\ket{\psi_p(\gamma,\beta, \ket{b_\theta})} \\
            &=\bra{\psi_p(\gamma,\beta, \ket{b_\theta})}\mathbf{X}^\dagger C\mathbf{X}\ket{\psi_p(\gamma,\beta, \ket{\bar{b}_\theta})} \tag{Proposition \ref{thm:bitflitQAOAOutputStates}} \\
            &=\bra{\psi_p(\gamma,\beta, \ket{b_\theta})}\mathbf{X}^\dagger\mathbf{X} C\ket{\psi_p(\gamma,\beta, \ket{\bar{b}_\theta})} \tag{Lemma \ref{thm:bitflipRelationOnCostHam}} \\
            &= F^{(p)}(\gamma, \beta, \ket{\bar{b}_\theta}). \tag{$\mathbf{X}^\dagger \mathbf{X} = I$}
        \end{align*}
    \end{proof}

\section{Proof of Lemma on Regularity and Cost Hamiltonian}
\label{sec:proofOfRecoverCutTheorem}
We first prove a few other useful lemmas.

\begin{lemma}
    \label{thm:kRegularParityOdd}
    Let $G$ be a $k$-regular graph with $k$ odd and let $b$ be a bistring corresponding to a cut in $G$. Let $H(b)$ be the Hamming-weight of $b$. Then,
    $$H(b) = \cut(b) \pmod{2}.$$
\end{lemma}
\begin{proof}
    We proceed by induction on $H(b)$. When $H(b) = 0$, then $b = 0\dots 0$ and it is clear that $\cut(b) = 0$ and the result follows.

    Now suppose that $H(b) > 0$, then there exists an index $j$ such that $b_j = 1$. Let $b'$ be the same bitstring as $b$ with the exception that $b'_j = 0$. Going from $b$ to $b'$ corresponds to moving vertex $j$ from one side of the cut to the other. For the cut corresponding to $b$, let $\delta_\text{out}$ and $\delta_\text{in}$ be the number of neighbors that vertex $j$ has on the other side and same side of the cut respectively. Since $G$ is $k$-regular, it must be that $\delta_\text{out} + \delta_\text{in} = k$. Going from $b$ to $b'$, i.e., moving vertex $j$ to the other side of the cut, the cut value decreases by $\delta_\text{out}$ and increases by $\delta_\text{in}$, thus, working in modulo 2:
    \begin{align*}
    \cut(b') &= \cut(b) + \delta_\text{in} - \delta_\text{out}\\
    &= \cut(b) + \delta_\text{in} - (k - \delta_\text{in}) \tag{as $\delta_\text{out} + \delta_\text{in} = k$}\\
    &= H(b) + 2\delta_\text{in} - k \tag{induction}\\
    &= H(b)-1 \tag{working in modulo 2, $k$ odd}\\
    &= (H(b')+1)-1 \tag{$b,b'$ differ by a bitflip}\\
    &= H(b').
    \end{align*}

    \end{proof}

    \begin{lemma}
    Let $G$ be a $k$-regular graph with $k$ even and let $b$ be a bistring corresponding to a cut in $G$.  Then,
    $$\cut(b) = 0 \pmod{2}.$$
\end{lemma}
\begin{proof}
    We proceed by induction on the Hamming weight $H(b)$. When $H(b) = 0$, it is clear that $\cut(b) = 0$ and the result follows.

    Now suppose that $H(b) > 0$, we define $b', \delta_\text{in}, \delta_\text{out}$ in the same way as is done in the proof of Lemma \ref{thm:kRegularParityOdd} and using the same argument, working in modulo 2, we have that:
    \begin{align*}
    \cut(b') &= \cut(b) + \delta_\text{in} - \delta_\text{out}\\
    &= \cut(b) + \delta_\text{in} - (k - \delta_\text{in}) \tag{as $\delta_\text{out} + \delta_\text{in} = k$}\\
    &= \cut(b) \tag{working in modulo 2, $k$ even}\\
    &= 0. \tag{induction, $H(b') = H(b)-1$}
    \end{align*}

    \end{proof}

    We are now ready to prove Lemma \ref{thm:unitaryisPhaseFlip}, which we restate below.
    \begin{lemma}
        Let $G$ be a $k$-regular graph with $n$ vertices and let $C$ be the \mc{} Hamiltonian for $G$ and let $\gamma = \pi$. Then, the cost unitary $U(C, \gamma)$ of QAOA is equal to the following:
        $$U(C, \gamma) = \begin{cases} Z^{\otimes n}, & \text{$k$ odd}\\ I, & \text{$k$ even}\end{cases}.$$
    \end{lemma}
    \begin{proof}
        Let us first consider the case where $k$ is odd. It suffices to show that $U(C,\gamma)\ket{b} = (Z^{\otimes n})\ket{b}$ for all $b \in \{0,1\}^n$. As $Z\ket{0} = \ket{0}$ and $Z\ket{1} = -\ket{1}$, it follows that $$(Z^{\otimes n})\ket{b} = \bigotimes_{j=1}^n (Z\ket{b_j}) = (-1)^{H(b)}\ket{b},$$ where $H(b)$ is the Hamming weight of $b$.

        For the cost unitary, we have,
        \begin{equation}
        \label{eqn:costUnitaryPhaseCutParity}
        U(C, \gamma)\ket{b} = e^{-i\gamma \cdot \cut(b)}\ket{b}
        = e^{-i\pi \cdot \cut(b)}\ket{b}
        = (-1)^{\cut(b)}\ket{b},
        \end{equation}
        where the last equality follows as $\cut(b)$ is an integer. For $k$ odd, we have that $\cut(b)$ and $H(b)$ have the same parity (Lemma \ref{thm:kRegularParityOdd}), and thus,
        $$U(C,\gamma)\ket{b} = (-1)^{\cut(b)}\ket{b} = (-1)^{H(b)}\ket{b} = (Z^{\otimes n})\ket{b},$$
        as desired.

        For $k$ even, it suffices to show that $U(C,\gamma)\ket{b} = \ket{b}$ for all $b \in \{0,1\}^n$. We have that $\cut(b) = 0 \pmod{2}$, and thus, from Equation \ref{eqn:costUnitaryPhaseCutParity}, we have that,
        $$U(C,\gamma)\ket{b} = (-1)^{\cut(b)}\ket{b} = (-1)^0 \ket{b} = \ket{b},$$
        as desired.
    \end{proof}

\section{Periodicity in Variational Parameters}
\label{sec:periodicityInVariationalParameters}
For depth-$p$ standard QAOA on regular graphs, it is known that it suffices to restrict $\beta_i \in [-\pi/4, \pi/4]$ and $\gamma_i \in [-\pi/2, \pi/2]$ (for all $i=1,2, \dots, p$) in order to thoroughly explore the QAOA landscape generated by $F^{(p)}(\gamma,\beta)$. \cite{zhou2020quantum}. However, the arguments used in \cite{zhou2020quantum} no longer work due to the change in the mixer. In this appendix, we prove that for depth-1 warm-started QAOA with fixed initialization angle, it suffices to check $(\gamma, \beta) \in [0, \pi]^{2}$ in order to thoroughly explore the landscape. (A similar analysis can be used to extend these results to depth-$p$ warm-started QAOA.)

\begin{prop}
\label{thm:betaPeriod}
    Let $\ket{s} = \bigotimes_{j=1}^n \ket{\vec{n}_j}$ be an arbitrary product state where each $\vec{n_j}$ is a unit-vector corresponding to a point on the Bloch sphere and let $C$ be any cost Hamiltonian for QAOA ($C$ does not necessarily need to correspond to \mc{}).  Then, the final state of QAOA is $\pi$-periodic in $\beta$, i.e., $\ket{\psi(\gamma,\beta, \ket{s})} = \ket{\psi(\gamma,\beta + \pi, \ket{s})}$  for all $\gamma,\beta \in \mathbb{R}$. Consequently, the expectations are also $\pi$-periodic, i.e., $F(\gamma,\beta,\ket{s}) = F(\gamma,\beta+\pi,\ket{s})$.
\end{prop}
\begin{proof}
Recall that the warm-started mixer $B_{\ket{s}}$ is given by $B_{\ket{s}} = \sum_{j=1}^n B_{\vec{n}_j,j}$ where $\vec{n}_j$ is the unit vector corresponding to the initial position of the $j$th qubit on the Bloch sphere. Since each $B_{\vec{n},j}$ acts on a different qubit, then each term in the sum commutes with one another, and thus, for all $\beta \in \mathbb{R}$, $U(B_{\ket{s}}, \beta) = \prod_{j=1}^n U(B_{\vec{n},j}, \beta)$.  Recall that the operation $U(B_{\vec{n}_j,j}, \beta)$ geometrically performs a rotation of the $j$th qubit by angle $2\beta$ about the axis that points in the $\vec{n}_j$ direction; thus $U(B_{\vec{n}_j,j}, \beta)$ is $\pi$-periodic in $\beta$ since a rotation by $2\beta = 2\pi$ leaves the qubit unchanged. Since $U(B_{\ket{s}}, \beta)$ is a product of $\pi$-periodic functions in $\beta$, then $U(B_{\ket{s}}, \beta)$ is also $\pi$-periodic in $\beta$. It is then clear that $\ket{\psi(\gamma,\beta, \ket{s})} = U(B_{\ket{s}}, \beta)U(C, \gamma)\ket{s}$ is $\pi$-periodic in $\beta$. As a consequence, $F(\gamma,\beta,\ket{s})$ is also $\pi$-periodic $\beta$.
\end{proof}

The next proposition is widely known; however, we include its proof for the sake of completion.
\begin{prop}
\label{thm:gammaPeriod}
    Let $C$ be a cost Hamiltonian corresponding to classical objective function $c: \{0,1\}^n \to \mathbb{Z}$, i.e., $c$ outputs integer values. Then $U(C, \gamma)$ is $2\pi$-periodic.
\end{prop}
\begin{proof}
    First, we show that $e^{-2i\pi C} = I$. To show this, it suffices to show that $e^{-2i\pi C}\ket{b}$ for all $\ket{b}$ in some choice of basis of the full Hilbert space; we consider the basis $\{\ket{b} : b \in \{0,1\}^n\}$. Now observe,
    $$e^{-2i\pi C}\ket{b} = e^{-2i\pi c(b)}\ket{b} = \ket{b},$$
    where $e^{-2i\pi c(b)} = 1$ since $c(b)$ is an integer. This shows that $e^{-2i\pi C} = I$ and hence: $$U(C, \gamma + 2\pi) = e^{-i(\gamma+2\pi) C} = e^{-i\gamma C}e^{-2i\pi C} = e^{-i\gamma C},$$
    as desired.
\end{proof}

\begin{prop}
\label{thm:negateParameters}
    Negating the variational parameters has no effect, i.e., for  $F(\gamma,\beta, \ket{s}) = F(-\gamma, -\beta, \ket{s})$. This holds for any cost Hamiltonian $C$ that is Hermitian and the result still holds if the mixer $B_{\ket{s}}$ is replaced with any mixing Hamiltonian $B$ that is Hermitian.
\end{prop}
\begin{proof}
    An elegant proof of this is given in \cite{lee2023depth} in the case that $\ket{s} = \ket{+}^{\otimes n}$; however, the same proof still works regardless of the choice of initial state (as long as it is independent of $\gamma$ and $\beta$). 
\end{proof}

From Propositions \ref{thm:betaPeriod} and \ref{thm:gammaPeriod}, it is clear that in order to explore the possible values of the landscape generated by $F(\gamma,\beta, \ket{s})$ (with fixed $\ket{s}$), it suffices to only consider $\gamma \in [0, 2\pi]$ and $\beta \in [0, \pi]$. We next show that we can further restrict $\gamma$ to $[0, \pi]$; to see this, consider $\gamma \in (\pi, 2\pi]$ and $\beta \in [0, \pi]$, then:
\begin{align*}
    F(\gamma, \beta, \ket{s}) &= F(-\gamma, -\beta, \ket{s}) \tag{Proposition \ref{thm:negateParameters}} \\
    &= F(2\pi-\gamma, \pi-\beta, \ket{s}) \tag{Propositions  \ref{thm:betaPeriod} and \ref{thm:gammaPeriod}}\\
    &= F(\gamma', \beta', \ket{s}),
\end{align*}
where $\gamma' = 2\pi-\gamma \in [0, \pi]$ and $\beta' = \pi - \beta \in [0, \pi]$.

\section{Tables of Values}
\label{sec:AR_Tables}
In this appendix, in Table \ref{tab:LB} Table \ref{tab:LB_2/3}, Table \ref{tab:LB_4/5}, and  Table \ref{tab:LB_17/21}, we provide numerically obtained values for $\mathbf{LB}(\theta), \mathbf{LB}_{2/3}(\theta), \mathbf{LB}'_{4/5}(\theta), \mathbf{LB}'_{17/21}(\theta)$ (respectively) for varying values of $\theta$. Due to challenges with the numerical optimization, the values reported for the various lower bounds may possibly be lower than the true values; see the section on optimization details in the main paper regarding how the numerical optimization was performed.

Each calculation of the lower bound involves calculating a max-min optimization, i.e., a nested optimization with an outer maximization (over $\gamma$ and $\beta$) and an inner minimization (over $\mathbf{r}, \mathbf{s}, \mathbf{t}$). For each $\theta$ and each lower bound, we also report the optimal values found for these outer and inner optimizations. For all of the lower bounds and for all values of $\theta$ we consider, our optimization procedure found that the optimal values for $\mathbf{s}$ and $\mathbf{t}$ are always $\mathbf{s} = \mathbf{t} = \mathbf{0}$ and hence we omit these values from the table.

As mentioned in the main text, Theorem \ref{thm:bitflipSameExpectedCut} can be extended to show that $f_{g_{6,1}}(\gamma,\beta,\theta) = f_{g_{6,6}}(\gamma,\beta,\theta)$ and $f_{g_{6,2}}(\gamma,\beta,\theta) = f_{g_{6,5}}(\gamma,\beta,\theta)$ for any choice of $\gamma,\beta,\theta$. Looking at the expected cut value of warm-started QAOA as a function of $\mathbf{r},\mathbf{s},\mathbf{t}$, if $r_1+r_6 = \xi$ for some constant $\xi$, then we can replace $r_1$ with $r_1'$ and $r_6$ with $r_6'$ satisfying $r_1'+r_6' = \xi$ and the expected cut value would remain the same. For this reason, we report the sum $r_1+r_6$ instead of reporting possible individual values of $r_1$ and $r_6$. Similarly, for the same reason, we report the sum $r_2+r_5$ instead of the possible individual values of $r_2$ and $r_5$.

\begin{table}
    \footnotesize
    \resizebox{1\textwidth}{!}{
    \begin{tabular}{|c|c|c|c|c|c|c|c|}
\hline $\theta$ & $\mathbf{LB}(\theta)$ & $\gamma$ & $\beta$ & $r_1+r_6$ & $r_2+r_5$ & $r_3$ & $r_4$ \\ \hline
1 & 0.0014 & 3.1361 & 1.5642 & 1.0000 & 0.0000 & 0.0000 & 0.0000 \\ \hline
2 & 0.0055 & 3.1463 & 1.5662 & 1.0000 & 0.0000 & 0.0000 & 0.0000 \\ \hline
3 & 0.0122 & 3.1431 & 1.5701 & 1.0000 & 0.0000 & 0.0000 & 0.0000 \\ \hline
4 & 0.0216 & 3.1435 & 1.5711 & 1.0000 & 0.0000 & 0.0000 & 0.0000 \\ \hline
5 & 0.0335 & 3.1438 & 1.5711 & 1.0000 & 0.0000 & 0.0000 & 0.0000 \\ \hline
6 & 0.0477 & 3.1435 & 1.5711 & 1.0000 & 0.0000 & 0.0000 & 0.0000 \\ \hline
7 & 0.0642 & 3.1424 & 1.5717 & 1.0000 & 0.0000 & 0.0000 & 0.0000 \\ \hline
8 & 0.0827 & 3.1442 & 1.5705 & 1.0000 & 0.0000 & 0.0000 & 0.0000 \\ \hline
9 & 0.1031 & 3.1442 & 1.5705 & 1.0000 & 0.0000 & 0.0000 & 0.0000 \\ \hline
10 & 0.1250 & 3.1442 & 1.5705 & 1.0000 & 0.0000 & 0.0000 & 0.0000 \\ \hline
11 & 0.1483 & 3.1436 & 1.5707 & 1.0000 & 0.0000 & 0.0000 & 0.0000 \\ \hline
12 & 0.1727 & 3.1434 & 1.5705 & 1.0000 & 0.0000 & 0.0000 & 0.0000 \\ \hline
13 & 0.1980 & 3.1398 & 1.5710 & 1.0000 & 0.0000 & 0.0000 & 0.0000 \\ \hline
14 & 0.2239 & 3.1410 & 1.5705 & 1.0000 & 0.0000 & 0.0000 & 0.0000 \\ \hline
15 & 0.2500 & 3.1422 & 1.5701 & 1.0000 & 0.0000 & 0.0000 & 0.0000 \\ \hline
16 & 0.2761 & 3.1454 & 1.5701 & 1.0000 & 0.0000 & 0.0000 & 0.0000 \\ \hline
17 & 0.3020 & 2.9850 & 1.5961 & 1.0000 & 0.0000 & 0.0000 & 0.0000 \\ \hline
18 & 0.3277 & 2.8494 & 1.6181 & 1.0000 & 0.0000 & 0.0000 & 0.0000 \\ \hline
19 & 0.3531 & 2.7622 & 1.6316 & 1.0000 & 0.0000 & 0.0000 & 0.0000 \\ \hline
20 & 0.3779 & 2.7002 & 1.6419 & 1.0000 & 0.0000 & 0.0000 & 0.0000 \\ \hline
21 & 0.4021 & 2.6469 & 1.6511 & 1.0000 & 0.0000 & 0.0000 & 0.0000 \\ \hline
22 & 0.4253 & 2.6030 & 1.6586 & 1.0000 & 0.0000 & 0.0000 & 0.0000 \\ \hline
23 & 0.4475 & 2.5638 & 1.6661 & 1.0000 & 0.0000 & 0.0000 & 0.0000 \\ \hline
24 & 0.4685 & 2.5293 & 1.6727 & 1.0000 & 0.0000 & 0.0000 & 0.0000 \\ \hline
25 & 0.4881 & 2.4975 & 1.6800 & 1.0000 & 0.0000 & 0.0000 & 0.0000 \\ \hline
26 & 0.5061 & 2.4672 & 1.6879 & 1.0000 & 0.0000 & 0.0000 & 0.0000 \\ \hline
27 & 0.5225 & 2.4396 & 1.6962 & 1.0000 & 0.0000 & 0.0000 & 0.0000 \\ \hline
28 & 0.5372 & 2.4099 & 1.6986 & 0.9999 & 0.0000 & 0.0000 & 0.0001 \\ \hline
29 & 0.5464 & 2.3894 & 1.5986 & 0.9987 & 0.0000 & 0.0000 & 0.0013 \\ \hline
30 & 0.5525 & 2.3925 & 1.5260 & 0.6667 & 0.0000 & 0.0000 & 0.3333 \\ \hline
31 & 0.5574 & 2.4007 & 1.4693 & 0.0021 & 0.0000 & 0.0000 & 0.9979 \\ \hline
32 & 0.5616 & 2.4100 & 1.4240 & 0.5005 & 0.0000 & 0.0000 & 0.4995 \\ \hline
33 & 0.5653 & 3.8632 & 1.7540 & 0.0008 & 0.0000 & 0.0000 & 0.9992 \\ \hline
34 & 0.5686 & 2.4305 & 1.3590 & 0.6667 & 0.0000 & 0.0000 & 0.3333 \\ \hline
35 & 0.5715 & 2.4660 & 1.3196 & 0.0000 & 0.0000 & 0.0000 & 1.0000 \\ \hline
36 & 0.5712 & 3.7468 & 1.9068 & 0.6664 & 0.3336 & 0.0000 & 0.0000 \\ \hline
37 & 0.5709 & 2.5651 & 1.1805 & 0.6662 & 0.3338 & 0.0000 & 0.0000 \\ \hline
38 & 0.5708 & 2.5886 & 1.1344 & 0.6655 & 0.3345 & 0.0000 & 0.0000 \\ \hline
39 & 0.5708 & 3.6749 & 2.0473 & 0.6659 & 0.3341 & 0.0000 & 0.0000 \\ \hline
40 & 0.5709 & 3.6583 & 2.0826 & 0.6661 & 0.3339 & 0.0000 & 0.0000 \\ \hline
41 & 0.5710 & 2.6400 & 1.0271 & 0.9998 & 0.0002 & 0.0000 & 0.0000 \\ \hline
42 & 0.5713 & 2.6529 & 0.9984 & 0.9998 & 0.0002 & 0.0000 & 0.0000 \\ \hline
43 & 0.5716 & 3.6185 & 2.1695 & 0.9998 & 0.0002 & 0.0000 & 0.0000 \\ \hline
44 & 0.5719 & 3.6085 & 2.1933 & 0.9998 & 0.0002 & 0.0000 & 0.0000 \\ \hline
45 & 0.5722 & 3.5993 & 2.2153 & 1.0000 & 0.0000 & 0.0000 & 0.0000 \\ \hline
\end{tabular} \begin{tabular}{|c|c|c|c|c|c|c|c|}
\hline $\theta$ & $\mathbf{LB}(\theta)$ & $\gamma$ & $\beta$ & $r_1+r_6$ & $r_2+r_5$ & $r_3$ & $r_4$ \\ \hline
46 & 0.5726 & 2.6923 & 0.9059 & 1.0000 & 0.0000 & 0.0000 & 0.0000 \\ \hline
47 & 0.5731 & 2.7002 & 0.8869 & 1.0000 & 0.0000 & 0.0000 & 0.0000 \\ \hline
48 & 0.5735 & 3.5763 & 2.2722 & 1.0000 & 0.0000 & 0.0000 & 0.0000 \\ \hline
49 & 0.5740 & 2.7130 & 0.8497 & 1.0000 & 0.0000 & 0.0000 & 0.0000 \\ \hline
50 & 0.5746 & 3.5648 & 2.3125 & 1.0000 & 0.0000 & 0.0000 & 0.0000 \\ \hline
51 & 0.5752 & 3.5608 & 2.3321 & 1.0000 & 0.0000 & 0.0000 & 0.0000 \\ \hline
52 & 0.5760 & 2.7258 & 0.7907 & 1.0000 & 0.0000 & 0.0000 & 0.0000 \\ \hline
53 & 0.5790 & 0.7493 & 1.6506 & 0.6667 & 0.3333 & 0.0000 & 0.0000 \\ \hline
54 & 0.5859 & 0.7292 & 1.6667 & 0.5000 & 0.5000 & 0.0000 & 0.0000 \\ \hline
55 & 0.5923 & 0.7102 & 1.6830 & 0.0076 & 0.9924 & 0.0000 & 0.0000 \\ \hline
56 & 0.5983 & 0.6921 & 1.6996 & 0.5024 & 0.4976 & 0.0000 & 0.0000 \\ \hline
57 & 0.6039 & 0.6751 & 1.7160 & 0.0059 & 0.9941 & 0.0000 & 0.0000 \\ \hline
58 & 0.6092 & 0.6592 & 1.7317 & 0.4970 & 0.5030 & 0.0000 & 0.0000 \\ \hline
59 & 0.6141 & 0.6442 & 1.7471 & 0.0002 & 0.9998 & 0.0000 & 0.0000 \\ \hline
60 & 0.6188 & 0.6302 & 1.7619 & 0.5000 & 0.5000 & 0.0000 & 0.0000 \\ \hline
61 & 0.6232 & 0.6174 & 1.7754 & 0.0003 & 0.9997 & 0.0000 & 0.0000 \\ \hline
62 & 0.6274 & 0.6051 & 1.7888 & 0.0000 & 1.0000 & 0.0000 & 0.0000 \\ \hline
63 & 0.6314 & 0.5936 & 1.8019 & 0.0000 & 1.0000 & 0.0000 & 0.0000 \\ \hline
64 & 0.6352 & 0.5844 & 1.8138 & 0.0000 & 1.0000 & 0.0000 & 0.0000 \\ \hline
65 & 0.6390 & 0.6164 & 1.8233 & 0.0000 & 1.0000 & 0.0000 & 0.0000 \\ \hline
66 & 0.6429 & 0.6264 & 1.8382 & 0.0000 & 1.0000 & 0.0000 & 0.0000 \\ \hline
67 & 0.6468 & 0.6290 & 1.8538 & 0.0000 & 1.0000 & 0.0000 & 0.0000 \\ \hline
68 & 0.6506 & 0.6292 & 1.8677 & 0.0000 & 0.9981 & 0.0000 & 0.0019 \\ \hline
69 & 0.6542 & 0.6285 & 1.8804 & 0.0000 & 0.5018 & 0.0000 & 0.4982 \\ \hline
70 & 0.6576 & 0.6270 & 1.8921 & 0.0000 & 0.9980 & 0.0000 & 0.0020 \\ \hline
71 & 0.6609 & 0.6247 & 1.9027 & 0.0000 & 1.0000 & 0.0000 & 0.0000 \\ \hline
72 & 0.6640 & 0.6224 & 1.9126 & 0.0000 & 0.5008 & 0.0000 & 0.4992 \\ \hline
73 & 0.6670 & 0.6197 & 1.9215 & 0.0000 & 0.9998 & 0.0000 & 0.0002 \\ \hline
74 & 0.6698 & 0.6171 & 1.9298 & 0.0000 & 0.5010 & 0.0000 & 0.4990 \\ \hline
75 & 0.6724 & 0.6137 & 1.9372 & 0.0000 & 0.5005 & 0.0000 & 0.4995 \\ \hline
76 & 0.6748 & 0.6116 & 1.9443 & 0.0000 & 0.0001 & 0.0000 & 0.9999 \\ \hline
77 & 0.6771 & 0.6090 & 1.9505 & 0.0000 & 0.0000 & 0.0000 & 1.0000 \\ \hline
78 & 0.6792 & 0.6100 & 1.9528 & 0.0000 & 0.0000 & 0.0000 & 1.0000 \\ \hline
79 & 0.6812 & 0.6117 & 1.9546 & 0.0000 & 0.0000 & 0.0000 & 1.0000 \\ \hline
80 & 0.6831 & 0.6126 & 1.9560 & 0.0000 & 0.0000 & 0.0000 & 1.0000 \\ \hline
81 & 0.6848 & 0.6134 & 1.9575 & 0.0000 & 0.0000 & 0.0000 & 1.0000 \\ \hline
82 & 0.6864 & 0.6143 & 1.9589 & 0.0000 & 0.0000 & 0.0000 & 1.0000 \\ \hline
83 & 0.6878 & 0.6147 & 1.9596 & 0.0000 & 0.0000 & 0.0000 & 1.0000 \\ \hline
84 & 0.6890 & 0.6149 & 1.9609 & 0.0000 & 0.0000 & 0.0000 & 1.0000 \\ \hline
85 & 0.6900 & 0.6153 & 1.9616 & 0.0000 & 0.0000 & 0.0000 & 1.0000 \\ \hline
86 & 0.6909 & 0.6153 & 1.9623 & 0.0000 & 0.0000 & 0.0000 & 1.0000 \\ \hline
87 & 0.6916 & 0.6153 & 1.9628 & 0.0000 & 0.0000 & 0.0000 & 1.0000 \\ \hline
88 & 0.6921 & 0.6154 & 1.9634 & 0.0000 & 0.0000 & 0.0000 & 1.0000 \\ \hline
89 & 0.6924 & 0.6155 & 1.9634 & 0.0000 & 0.0000 & 0.0000 & 1.0000 \\ \hline
90 & 0.6924 & 0.6156 & 1.9636 & 0.3333 & 0.3333 & 0.1667 & 0.1667 \\ \hline
\end{tabular}
}
    \caption{\footnotesize Table of values of $\mathbf{LB}(\theta)$ for various $\theta$ including the optimal values of the minimizers and maximizers found in the inner and outer optimizations respectively.}
    \label{tab:LB}
\end{table}

\begin{table}
    \centering
    \footnotesize
    \resizebox{1\textwidth}{!}{
    \begin{tabular}{|c|c|c|c|c|c|c|c|}
\hline $\theta$ & $\mathbf{LB}_{2/3}(\theta)$ & $\gamma$ & $\beta$ & $r_1+r_6$ & $r_2+r_5$ & $r_3$ & $r_4$ \\ \hline
1 & 0.6667 & 0.5168 & 0.5341 & 0.3333 & 0.0000 & 0.0000 & 0.6666 \\ \hline
2 & 0.6666 & 0.5158 & 0.5396 & 0.3333 & 0.6666 & 0.0000 & 0.0001 \\ \hline
3 & 0.6665 & 0.5207 & 0.5307 & 0.3333 & 0.3333 & 0.0000 & 0.3333 \\ \hline
4 & 0.6664 & 0.5226 & 0.5276 & 0.3333 & 0.0001 & 0.0000 & 0.6665 \\ \hline
5 & 0.6663 & 0.5242 & 0.5259 & 0.3333 & 0.0000 & 0.0000 & 0.6666 \\ \hline
6 & 0.6662 & 0.5269 & 0.5221 & 0.3333 & 0.4444 & 0.0000 & 0.2222 \\ \hline
7 & 0.6660 & 0.5263 & 0.5252 & 0.3333 & 0.3334 & 0.0000 & 0.3333 \\ \hline
8 & 0.6658 & 0.5283 & 0.5233 & 0.3333 & 0.0010 & 0.0000 & 0.6657 \\ \hline
9 & 0.6655 & 0.5311 & 0.5200 & 0.3333 & 0.0000 & 0.0000 & 0.6667 \\ \hline
10 & 0.6653 & 0.5434 & 0.4974 & 0.3333 & 0.0000 & 0.0000 & 0.6667 \\ \hline
11 & 0.6650 & 0.5358 & 0.5161 & 0.3333 & 0.0000 & 0.0000 & 0.6666 \\ \hline
12 & 0.6646 & 0.5399 & 0.5110 & 0.3333 & 0.3334 & 0.0000 & 0.3333 \\ \hline
13 & 0.6643 & 0.5433 & 0.5073 & 0.3333 & 0.0000 & 0.0000 & 0.6667 \\ \hline
14 & 0.6639 & 0.5466 & 0.5042 & 0.3333 & 0.0000 & 0.0000 & 0.6667 \\ \hline
15 & 0.6635 & 0.5503 & 0.5006 & 0.3333 & 0.3333 & 0.0000 & 0.3333 \\ \hline
16 & 0.6630 & 0.5544 & 0.4963 & 0.3333 & 0.6666 & 0.0000 & 0.0001 \\ \hline
17 & 0.6625 & 0.5583 & 0.4925 & 0.3333 & 0.0000 & 0.0000 & 0.6667 \\ \hline
18 & 0.6620 & 0.5634 & 0.4861 & 0.3333 & 0.6667 & 0.0000 & 0.0000 \\ \hline
19 & 0.6614 & 0.5678 & 0.4816 & 0.3333 & 0.0000 & 0.0000 & 0.6667 \\ \hline
20 & 0.6609 & 0.5721 & 0.4776 & 0.3333 & 0.0001 & 0.0000 & 0.6666 \\ \hline
21 & 0.6603 & 0.5773 & 0.4714 & 0.3333 & 0.0000 & 0.0000 & 0.6667 \\ \hline
22 & 0.6596 & 0.5817 & 0.4672 & 0.3333 & 0.0000 & 0.0000 & 0.6666 \\ \hline
23 & 0.6590 & 0.5865 & 0.4622 & 0.3333 & 0.0000 & 0.0000 & 0.6667 \\ \hline
24 & 0.6583 & 0.5915 & 0.4568 & 0.3333 & 0.0000 & 0.0000 & 0.6667 \\ \hline
25 & 0.6576 & 0.5963 & 0.4519 & 0.3333 & 0.0001 & 0.0000 & 0.6666 \\ \hline
26 & 0.6569 & 0.6012 & 0.4468 & 0.3333 & 0.0000 & 0.0000 & 0.6667 \\ \hline
27 & 0.6561 & 0.6058 & 0.4425 & 0.3333 & 0.0000 & 0.0000 & 0.6667 \\ \hline
28 & 0.6554 & 0.6107 & 0.4374 & 0.3333 & 0.0000 & 0.0000 & 0.6666 \\ \hline
29 & 0.6546 & 0.6152 & 0.4335 & 0.3333 & 0.0002 & 0.0000 & 0.6665 \\ \hline
30 & 0.6539 & 0.6202 & 0.4278 & 0.3333 & 0.0000 & 0.0000 & 0.6666 \\ \hline
31 & 0.6531 & 0.6267 & 0.4178 & 0.3333 & 0.0000 & 0.0000 & 0.6667 \\ \hline
32 & 0.6524 & 0.6295 & 0.4186 & 0.3333 & 0.0000 & 0.0000 & 0.6666 \\ \hline
33 & 0.6516 & 0.6339 & 0.4146 & 0.3333 & 0.0000 & 0.0000 & 0.6667 \\ \hline
34 & 0.6509 & 0.6384 & 0.4102 & 0.3333 & 0.0000 & 0.0000 & 0.6666 \\ \hline
35 & 0.6502 & 0.6428 & 0.4060 & 0.3333 & 0.0000 & 0.0000 & 0.6667 \\ \hline
36 & 0.6495 & 0.6468 & 0.4028 & 0.3333 & 0.0000 & 0.0000 & 0.6667 \\ \hline
37 & 0.6488 & 0.6510 & 0.3991 & 0.3333 & 0.0000 & 0.0000 & 0.6667 \\ \hline
38 & 0.6482 & 0.6549 & 0.3960 & 0.3333 & 0.0000 & 0.0000 & 0.6667 \\ \hline
39 & 0.6476 & 0.6588 & 0.3929 & 0.3333 & 0.0000 & 0.0000 & 0.6667 \\ \hline
40 & 0.6470 & 0.6625 & 0.3901 & 0.3333 & 0.0000 & 0.0000 & 0.6667 \\ \hline
41 & 0.6465 & 0.6663 & 0.3870 & 0.3333 & 0.0000 & 0.0000 & 0.6667 \\ \hline
42 & 0.6461 & 0.6621 & 0.3850 & 0.3333 & 0.0000 & 0.0000 & 0.6667 \\ \hline
43 & 0.6457 & 0.6550 & 0.3830 & 0.3333 & 0.0000 & 0.0000 & 0.6667 \\ \hline
44 & 0.6454 & 0.6486 & 0.3813 & 0.3333 & 0.0000 & 0.0000 & 0.6667 \\ \hline
45 & 0.6452 & 0.6423 & 0.3799 & 0.3333 & 0.0000 & 0.0000 & 0.6667 \\ \hline
\end{tabular} \begin{tabular}{|c|c|c|c|c|c|c|c|}
\hline $\theta$ & $\mathbf{LB}_{2/3}(\theta)$ & $\gamma$ & $\beta$ & $r_1+r_6$ & $r_2+r_5$ & $r_3$ & $r_4$ \\ \hline
46 & 0.6451 & 0.6369 & 0.3785 & 0.3333 & 0.0000 & 0.0000 & 0.6667 \\ \hline
47 & 0.6451 & 0.6332 & 0.3771 & 0.3333 & 0.0000 & 0.0000 & 0.6667 \\ \hline
48 & 0.6451 & 0.6288 & 0.3763 & 0.3333 & 0.0000 & 0.0000 & 0.6667 \\ \hline
49 & 0.6453 & 0.6250 & 0.3754 & 0.3333 & 0.0000 & 0.0000 & 0.6667 \\ \hline
50 & 0.6456 & 0.6211 & 0.3746 & 0.3333 & 0.0000 & 0.0000 & 0.6667 \\ \hline
51 & 0.6460 & 0.6185 & 0.3744 & 0.3333 & 0.0000 & 0.0000 & 0.6667 \\ \hline
52 & 0.6467 & 2.4403 & 0.3950 & 0.3333 & 0.0000 & 0.0000 & 0.6667 \\ \hline
53 & 0.6480 & 2.4474 & 0.3974 & 0.3333 & 0.0000 & 0.0000 & 0.6667 \\ \hline
54 & 0.6508 & 3.1417 & 1.5708 & 0.3333 & 0.2223 & 0.2222 & 0.2222 \\ \hline
55 & 0.6555 & 3.1414 & 1.5709 & 0.3333 & 0.2223 & 0.2222 & 0.2222 \\ \hline
56 & 0.6595 & 3.1418 & 1.5707 & 0.3333 & 0.2227 & 0.2220 & 0.2220 \\ \hline
57 & 0.6626 & 3.1420 & 1.5706 & 0.3333 & 0.6667 & 0.0000 & 0.0000 \\ \hline
58 & 0.6648 & 3.1414 & 1.5713 & 0.3333 & 0.6667 & 0.0000 & 0.0000 \\ \hline
59 & 0.6662 & 3.1417 & 1.5707 & 0.3333 & 0.4440 & 0.2220 & 0.0006 \\ \hline
60 & 0.6667 & 3.1413 & 1.5710 & 0.3333 & 0.2224 & 0.2222 & 0.2222 \\ \hline
61 & 0.6662 & 3.1418 & 1.5706 & 0.3333 & 0.2223 & 0.2222 & 0.2222 \\ \hline
62 & 0.6648 & 3.1415 & 1.5710 & 0.3333 & 0.6667 & 0.0000 & 0.0000 \\ \hline
63 & 0.6643 & 2.4927 & 0.4054 & 0.3333 & 0.0000 & 0.0000 & 0.6667 \\ \hline
64 & 0.6660 & 2.4956 & 0.4053 & 0.3333 & 0.0000 & 0.0000 & 0.6667 \\ \hline
65 & 0.6677 & 2.4985 & 0.4047 & 0.3333 & 0.0000 & 0.0000 & 0.6667 \\ \hline
66 & 0.6694 & 2.5006 & 0.4043 & 0.3333 & 0.0000 & 0.0000 & 0.6667 \\ \hline
67 & 0.6710 & 2.5033 & 0.4037 & 0.3333 & 0.0000 & 0.0000 & 0.6667 \\ \hline
68 & 0.6726 & 2.5052 & 0.4030 & 0.3333 & 0.0000 & 0.0000 & 0.6667 \\ \hline
69 & 0.6742 & 2.5076 & 0.4026 & 0.3333 & 0.0000 & 0.0000 & 0.6667 \\ \hline
70 & 0.6757 & 2.5096 & 0.4017 & 0.3333 & 0.0000 & 0.0000 & 0.6667 \\ \hline
71 & 0.6772 & 2.5112 & 0.4012 & 0.3333 & 0.0000 & 0.0000 & 0.6667 \\ \hline
72 & 0.6787 & 2.5130 & 0.4004 & 0.3333 & 0.0000 & 0.0000 & 0.6667 \\ \hline
73 & 0.6801 & 2.5144 & 0.3996 & 0.3333 & 0.0000 & 0.0000 & 0.6667 \\ \hline
74 & 0.6814 & 2.5157 & 0.3990 & 0.3333 & 0.0000 & 0.0000 & 0.6667 \\ \hline
75 & 0.6827 & 2.5173 & 0.3982 & 0.3333 & 0.0000 & 0.0000 & 0.6667 \\ \hline
76 & 0.6839 & 2.5184 & 0.3974 & 0.3333 & 0.0000 & 0.0000 & 0.6667 \\ \hline
77 & 0.6850 & 2.5196 & 0.3970 & 0.3333 & 0.0000 & 0.0000 & 0.6667 \\ \hline
78 & 0.6861 & 2.5203 & 0.3964 & 0.3333 & 0.0000 & 0.0000 & 0.6667 \\ \hline
79 & 0.6871 & 2.5214 & 0.3959 & 0.3333 & 0.0000 & 0.0000 & 0.6667 \\ \hline
80 & 0.6880 & 2.5224 & 0.3954 & 0.3333 & 0.0000 & 0.0000 & 0.6667 \\ \hline
81 & 0.6888 & 2.5228 & 0.3948 & 0.3333 & 0.0000 & 0.0000 & 0.6667 \\ \hline
82 & 0.6896 & 2.5238 & 0.3945 & 0.3333 & 0.0000 & 0.0000 & 0.6667 \\ \hline
83 & 0.6902 & 2.5242 & 0.3939 & 0.3333 & 0.0000 & 0.0000 & 0.6667 \\ \hline
84 & 0.6908 & 2.5247 & 0.3935 & 0.3333 & 0.0000 & 0.0000 & 0.6667 \\ \hline
85 & 0.6913 & 2.5251 & 0.3934 & 0.3333 & 0.0000 & 0.0000 & 0.6667 \\ \hline
86 & 0.6917 & 2.5256 & 0.3932 & 0.3333 & 0.0000 & 0.0000 & 0.6667 \\ \hline
87 & 0.6920 & 2.5256 & 0.3931 & 0.3333 & 0.0000 & 0.0000 & 0.6667 \\ \hline
88 & 0.6923 & 2.5259 & 0.3925 & 0.3333 & 0.0000 & 0.0000 & 0.6667 \\ \hline
89 & 0.6924 & 2.5264 & 0.3926 & 0.3333 & 0.0000 & 0.0000 & 0.6667 \\ \hline
90 & 0.6924 & 0.6156 & 1.9636 & 0.1994 & 0.4003 & 0.2001 & 0.2001 \\ \hline
\end{tabular}

    }
    \caption{\footnotesize Table of values of $\mathbf{LB}_{2/3}(\theta)$ for various $\theta$ including the optimal values of the minimizers and maximizers found in the inner and outer optimizations respectively.}
    \label{tab:LB_2/3}
\end{table}

\begin{table}
    \centering
    \footnotesize
    \resizebox{1\textwidth}{!}{
    \begin{tabular}{|c|c|c|c|c|c|c|c|}
\hline $\theta$ & $\mathbf{LB}'_{4/5}(\theta)$ & $\gamma$ & $\beta$ & $r_1+r_6$ & $r_2+r_5$ & $r_3$ & $r_4$ \\ \hline
1 & 0.8000 & 0.5231 & 0.5235 & 0.2000 & 0.7993 & 0.0000 & 0.0007 \\ \hline
2 & 0.7999 & 0.5222 & 0.5264 & 0.2000 & 0.0000 & 0.0000 & 0.8000 \\ \hline
3 & 0.7998 & 0.5232 & 0.5255 & 0.2000 & 0.0001 & 0.0000 & 0.7999 \\ \hline
4 & 0.7996 & 0.5235 & 0.5261 & 0.2000 & 0.0015 & 0.0000 & 0.7985 \\ \hline
5 & 0.7994 & 0.5211 & 0.5322 & 0.2000 & 0.8000 & 0.0000 & 0.0000 \\ \hline
6 & 0.7991 & 0.5243 & 0.5275 & 0.2000 & 0.7998 & 0.0000 & 0.0002 \\ \hline
7 & 0.7988 & 0.5243 & 0.5293 & 0.2000 & 0.0001 & 0.0000 & 0.7999 \\ \hline
8 & 0.7984 & 0.5273 & 0.5252 & 0.2000 & 0.0000 & 0.0000 & 0.8000 \\ \hline
9 & 0.7980 & 0.5284 & 0.5255 & 0.2000 & 0.0002 & 0.0000 & 0.7998 \\ \hline
10 & 0.7975 & 0.5295 & 0.5262 & 0.2000 & 0.0000 & 0.0000 & 0.8000 \\ \hline
11 & 0.7969 & 0.5319 & 0.5241 & 0.2000 & 0.0001 & 0.0000 & 0.7999 \\ \hline
12 & 0.7963 & 0.5351 & 0.5208 & 0.2000 & 0.0000 & 0.0000 & 0.8000 \\ \hline
13 & 0.7956 & 0.5375 & 0.5196 & 0.2000 & 0.8000 & 0.0000 & 0.0000 \\ \hline
14 & 0.7949 & 0.5409 & 0.5163 & 0.2000 & 0.4000 & 0.0000 & 0.4000 \\ \hline
15 & 0.7941 & 0.5439 & 0.5144 & 0.2000 & 0.0001 & 0.0000 & 0.7999 \\ \hline
16 & 0.7932 & 0.5469 & 0.5123 & 0.2000 & 0.0000 & 0.0000 & 0.8000 \\ \hline
17 & 0.7923 & 0.5502 & 0.5101 & 0.2000 & 0.0000 & 0.0000 & 0.8000 \\ \hline
18 & 0.7913 & 0.5544 & 0.5063 & 0.2000 & 0.7991 & 0.0000 & 0.0009 \\ \hline
19 & 0.7902 & 0.5580 & 0.5037 & 0.2000 & 0.4000 & 0.0000 & 0.4000 \\ \hline
20 & 0.7890 & 0.5620 & 0.5006 & 0.2000 & 0.7997 & 0.0000 & 0.0003 \\ \hline
21 & 0.7878 & 0.5658 & 0.4981 & 0.2000 & 0.0087 & 0.0000 & 0.7913 \\ \hline
22 & 0.7864 & 0.5702 & 0.4943 & 0.2000 & 0.4000 & 0.0000 & 0.4000 \\ \hline
23 & 0.7850 & 0.5745 & 0.4911 & 0.2000 & 0.7998 & 0.0000 & 0.0002 \\ \hline
24 & 0.7836 & 0.5786 & 0.4883 & 0.2000 & 0.4000 & 0.0000 & 0.4000 \\ \hline
25 & 0.7820 & 0.5835 & 0.4841 & 0.2000 & 0.5333 & 0.0000 & 0.2667 \\ \hline
26 & 0.7804 & 0.5885 & 0.4794 & 0.2000 & 0.7970 & 0.0000 & 0.0030 \\ \hline
27 & 0.7787 & 0.5932 & 0.4755 & 0.2000 & 0.0000 & 0.0000 & 0.8000 \\ \hline
28 & 0.7769 & 0.5981 & 0.4713 & 0.2000 & 0.7999 & 0.0000 & 0.0001 \\ \hline
29 & 0.7750 & 0.6028 & 0.4676 & 0.2000 & 0.0001 & 0.0000 & 0.7999 \\ \hline
30 & 0.7731 & 0.6080 & 0.4625 & 0.2000 & 0.5333 & 0.0000 & 0.2667 \\ \hline
31 & 0.7711 & 0.6125 & 0.4593 & 0.2000 & 0.0000 & 0.0000 & 0.8000 \\ \hline
32 & 0.7690 & 0.6175 & 0.4550 & 0.2000 & 0.0000 & 0.0000 & 0.8000 \\ \hline
33 & 0.7669 & 0.6224 & 0.4505 & 0.2000 & 0.0000 & 0.0000 & 0.8000 \\ \hline
34 & 0.7647 & 0.6269 & 0.4474 & 0.2000 & 0.0000 & 0.0000 & 0.8000 \\ \hline
35 & 0.7624 & 0.6318 & 0.4429 & 0.2000 & 0.0000 & 0.0000 & 0.8000 \\ \hline
36 & 0.7601 & 0.6362 & 0.4399 & 0.2000 & 0.0000 & 0.0000 & 0.8000 \\ \hline
37 & 0.7577 & 0.6411 & 0.4350 & 0.2000 & 0.0000 & 0.0000 & 0.8000 \\ \hline
38 & 0.7553 & 0.6456 & 0.4314 & 0.2000 & 0.0000 & 0.0000 & 0.8000 \\ \hline
39 & 0.7529 & 0.6501 & 0.4277 & 0.2000 & 0.0000 & 0.0000 & 0.8000 \\ \hline
40 & 0.7504 & 0.6545 & 0.4237 & 0.2000 & 0.0000 & 0.0000 & 0.8000 \\ \hline
41 & 0.7479 & 0.6588 & 0.4205 & 0.2000 & 0.0000 & 0.0000 & 0.8000 \\ \hline
42 & 0.7454 & 0.6629 & 0.4173 & 0.2000 & 0.0000 & 0.0000 & 0.8000 \\ \hline
43 & 0.7429 & 0.6671 & 0.4137 & 0.2000 & 0.0000 & 0.0000 & 0.8000 \\ \hline
44 & 0.7403 & 0.6710 & 0.4108 & 0.2000 & 0.0000 & 0.0000 & 0.8000 \\ \hline
45 & 0.7378 & 0.6748 & 0.4083 & 0.2000 & 0.0000 & 0.0000 & 0.8000 \\ \hline
\end{tabular} \begin{tabular}{|c|c|c|c|c|c|c|c|}
\hline $\theta$ & $\mathbf{LB}'_{4/5}(\theta)$ & $\gamma$ & $\beta$ & $r_1+r_6$ & $r_2+r_5$ & $r_3$ & $r_4$ \\ \hline
46 & 0.7352 & 0.6789 & 0.4032 & 0.2000 & 0.0000 & 0.0000 & 0.8000 \\ \hline
47 & 0.7327 & 0.6822 & 0.4028 & 0.2000 & 0.0000 & 0.0000 & 0.8000 \\ \hline
48 & 0.7302 & 0.6796 & 0.4005 & 0.2000 & 0.0000 & 0.0000 & 0.8000 \\ \hline
49 & 0.7278 & 0.6735 & 0.3985 & 0.2000 & 0.0000 & 0.0000 & 0.8000 \\ \hline
50 & 0.7254 & 0.6678 & 0.3972 & 0.2000 & 0.0000 & 0.0000 & 0.8000 \\ \hline
51 & 0.7382 & 3.1419 & 1.5707 & 0.2000 & 0.7997 & 0.0002 & 0.0001 \\ \hline
52 & 0.7504 & 3.1417 & 1.5708 & 0.2000 & 0.4000 & 0.4000 & 0.0000 \\ \hline
53 & 0.7615 & 3.1414 & 1.5709 & 0.2000 & 0.4000 & 0.2000 & 0.2000 \\ \hline
54 & 0.7714 & 3.1415 & 1.5708 & 0.2000 & 0.4000 & 0.0000 & 0.4000 \\ \hline
55 & 0.7799 & 3.1417 & 1.5707 & 0.2000 & 0.4000 & 0.4000 & 0.0000 \\ \hline
56 & 0.7870 & 3.1419 & 1.5707 & 0.2000 & 0.0008 & 0.7991 & 0.0001 \\ \hline
57 & 0.7927 & 3.1420 & 1.5706 & 0.2000 & 0.8000 & 0.0000 & 0.0000 \\ \hline
58 & 0.7967 & 3.1420 & 1.5707 & 0.2000 & 0.2669 & 0.2665 & 0.2665 \\ \hline
59 & 0.7992 & 3.1415 & 1.5708 & 0.2000 & 0.4000 & 0.4000 & 0.0001 \\ \hline
60 & 0.8000 & 3.1417 & 1.5708 & 0.2000 & 0.4000 & 0.3997 & 0.0003 \\ \hline
61 & 0.7992 & 3.1414 & 1.5709 & 0.2000 & 0.0000 & 0.4000 & 0.4000 \\ \hline
62 & 0.7967 & 3.1414 & 1.5710 & 0.2000 & 0.5320 & 0.0020 & 0.2660 \\ \hline
63 & 0.7927 & 3.1419 & 1.5708 & 0.2000 & 0.7999 & 0.0000 & 0.0000 \\ \hline
64 & 0.7870 & 3.1419 & 1.5707 & 0.2000 & 0.7999 & 0.0000 & 0.0000 \\ \hline
65 & 0.7799 & 3.1424 & 1.5703 & 0.2000 & 0.4007 & 0.0003 & 0.3990 \\ \hline
66 & 0.7714 & 3.1415 & 1.5710 & 0.2000 & 0.7985 & 0.0003 & 0.0012 \\ \hline
67 & 0.7615 & 3.1420 & 1.5703 & 0.2000 & 0.7998 & 0.0001 & 0.0001 \\ \hline
68 & 0.7504 & 3.1407 & 1.5716 & 0.2000 & 0.0011 & 0.0000 & 0.7989 \\ \hline
69 & 0.7387 & 2.9964 & 1.6782 & 0.2000 & 0.0000 & 0.0000 & 0.8000 \\ \hline
70 & 0.7287 & 2.9027 & 1.7433 & 0.2000 & 0.0000 & 0.0000 & 0.8000 \\ \hline
71 & 0.7205 & 2.8389 & 1.7853 & 0.2000 & 0.0000 & 0.0000 & 0.8000 \\ \hline
72 & 0.7138 & 2.7888 & 1.8163 & 0.2000 & 0.0000 & 0.0000 & 0.8000 \\ \hline
73 & 0.7084 & 2.7471 & 1.8410 & 0.2000 & 0.0000 & 0.0000 & 0.8000 \\ \hline
74 & 0.7040 & 2.7120 & 1.8614 & 0.2000 & 0.0000 & 0.0000 & 0.8000 \\ \hline
75 & 0.7006 & 2.6821 & 1.8782 & 0.2000 & 0.0000 & 0.0000 & 0.8000 \\ \hline
76 & 0.6980 & 2.6561 & 1.8925 & 0.2000 & 0.0000 & 0.0000 & 0.8000 \\ \hline
77 & 0.6960 & 2.6340 & 1.9046 & 0.2000 & 0.0000 & 0.0000 & 0.8000 \\ \hline
78 & 0.6945 & 2.6148 & 1.9150 & 0.2000 & 0.0000 & 0.0000 & 0.8000 \\ \hline
79 & 0.6934 & 2.5986 & 1.9238 & 0.2000 & 0.0000 & 0.0000 & 0.8000 \\ \hline
80 & 0.6928 & 0.6152 & 0.3911 & 0.2000 & 0.0000 & 0.0000 & 0.8000 \\ \hline
81 & 0.6927 & 0.6151 & 0.3913 & 0.2000 & 0.0000 & 0.0000 & 0.8000 \\ \hline
82 & 0.6926 & 0.6148 & 0.3916 & 0.2000 & 0.0000 & 0.0000 & 0.8000 \\ \hline
83 & 0.6926 & 0.6153 & 0.3920 & 0.2000 & 0.0000 & 0.0000 & 0.8000 \\ \hline
84 & 0.6925 & 0.6154 & 0.3921 & 0.2000 & 0.0000 & 0.0000 & 0.8000 \\ \hline
85 & 0.6925 & 0.6154 & 0.3923 & 0.2000 & 0.0000 & 0.0000 & 0.8000 \\ \hline
86 & 0.6925 & 0.6153 & 0.3923 & 0.2000 & 0.0000 & 0.0000 & 0.8000 \\ \hline
87 & 0.6925 & 0.6154 & 0.3926 & 0.2000 & 0.0000 & 0.0000 & 0.8000 \\ \hline
88 & 0.6925 & 0.6157 & 0.3927 & 0.2000 & 0.0000 & 0.0000 & 0.8000 \\ \hline
89 & 0.6925 & 0.6155 & 0.3928 & 0.2000 & 0.0000 & 0.0000 & 0.8000 \\ \hline
90 & 0.6924 & 0.6156 & 1.9636 & 0.1513 & 0.4244 & 0.2122 & 0.2122 \\ \hline
\end{tabular}

    }
    \caption{\footnotesize Table of values of $\mathbf{LB}_{4/5}'(\theta)$ for various $\theta$ including the optimal values of the minimizers and maximizers found in the inner and outer optimizations respectively.}
    \label{tab:LB_4/5}
\end{table}

\begin{table}
    \centering
    \footnotesize
    \resizebox{1\textwidth}{!}{
    \begin{tabular}{|c|c|c|c|c|c|c|c|}
\hline $\theta$ & $\mathbf{LB}'_{17/21}(\theta)$ & $\gamma$ & $\beta$ & $r_1+r_6$ & $r_2+r_5$ & $r_3$ & $r_4$ \\ \hline
1 & 0.8095 & 0.5226 & 0.5250 & 0.1905 & 0.0001 & 0.0000 & 0.8094 \\ \hline
2 & 0.8094 & 0.5222 & 0.5264 & 0.1905 & 0.0000 & 0.0000 & 0.8095 \\ \hline
3 & 0.8093 & 0.5232 & 0.5255 & 0.1905 & 0.0001 & 0.0000 & 0.8094 \\ \hline
4 & 0.8091 & 0.5235 & 0.5261 & 0.1905 & 0.0017 & 0.0000 & 0.8079 \\ \hline
5 & 0.8089 & 0.5211 & 0.5322 & 0.1905 & 0.8095 & 0.0000 & 0.0000 \\ \hline
6 & 0.8086 & 0.5225 & 0.5310 & 0.1905 & 0.4048 & 0.0000 & 0.4047 \\ \hline
7 & 0.8083 & 0.5243 & 0.5293 & 0.1905 & 0.0001 & 0.0000 & 0.8094 \\ \hline
8 & 0.8079 & 0.5256 & 0.5289 & 0.1905 & 0.4052 & 0.0000 & 0.4044 \\ \hline
9 & 0.8074 & 0.5273 & 0.5278 & 0.1905 & 0.0000 & 0.0000 & 0.8095 \\ \hline
10 & 0.8069 & 0.5291 & 0.5268 & 0.1905 & 0.0000 & 0.0000 & 0.8095 \\ \hline
11 & 0.8063 & 0.5319 & 0.5241 & 0.1905 & 0.0001 & 0.0000 & 0.8094 \\ \hline
12 & 0.8057 & 0.5351 & 0.5208 & 0.1905 & 0.0000 & 0.0000 & 0.8095 \\ \hline
13 & 0.8050 & 0.5375 & 0.5196 & 0.1905 & 0.8095 & 0.0000 & 0.0000 \\ \hline
14 & 0.8042 & 0.5409 & 0.5164 & 0.1905 & 0.0001 & 0.0000 & 0.8094 \\ \hline
15 & 0.8034 & 0.5431 & 0.5160 & 0.1905 & 0.0002 & 0.0000 & 0.8093 \\ \hline
16 & 0.8025 & 0.5471 & 0.5120 & 0.1905 & 0.0000 & 0.0000 & 0.8095 \\ \hline
17 & 0.8015 & 0.5506 & 0.5094 & 0.1905 & 0.0003 & 0.0000 & 0.8093 \\ \hline
18 & 0.8005 & 0.5538 & 0.5075 & 0.1905 & 0.0000 & 0.0000 & 0.8095 \\ \hline
19 & 0.7994 & 0.5573 & 0.5052 & 0.1905 & 0.4048 & 0.0000 & 0.4047 \\ \hline
20 & 0.7982 & 0.5611 & 0.5027 & 0.1905 & 0.0000 & 0.0000 & 0.8095 \\ \hline
21 & 0.7969 & 0.5656 & 0.4985 & 0.1905 & 0.0000 & 0.0000 & 0.8095 \\ \hline
22 & 0.7955 & 0.5700 & 0.4949 & 0.1905 & 0.4048 & 0.0000 & 0.4048 \\ \hline
23 & 0.7941 & 0.5739 & 0.4926 & 0.1905 & 0.0001 & 0.0000 & 0.8094 \\ \hline
24 & 0.7925 & 0.5786 & 0.4885 & 0.1905 & 0.0002 & 0.0000 & 0.8093 \\ \hline
25 & 0.7909 & 0.5828 & 0.4857 & 0.1905 & 0.0000 & 0.0000 & 0.8095 \\ \hline
26 & 0.7892 & 0.5880 & 0.4805 & 0.1905 & 0.0000 & 0.0000 & 0.8095 \\ \hline
27 & 0.7874 & 0.5923 & 0.4777 & 0.1905 & 0.0000 & 0.0000 & 0.8095 \\ \hline
28 & 0.7856 & 0.5974 & 0.4732 & 0.1905 & 0.4048 & 0.0000 & 0.4047 \\ \hline
29 & 0.7836 & 0.6019 & 0.4699 & 0.1905 & 0.0000 & 0.0000 & 0.8095 \\ \hline
30 & 0.7816 & 0.6074 & 0.4643 & 0.1905 & 0.0000 & 0.0000 & 0.8095 \\ \hline
31 & 0.7795 & 0.6120 & 0.4610 & 0.1905 & 0.0000 & 0.0000 & 0.8095 \\ \hline
32 & 0.7774 & 0.6169 & 0.4568 & 0.1905 & 0.0000 & 0.0000 & 0.8095 \\ \hline
33 & 0.7751 & 0.6217 & 0.4526 & 0.1905 & 0.8079 & 0.0000 & 0.0016 \\ \hline
34 & 0.7728 & 0.6267 & 0.4480 & 0.1905 & 0.0000 & 0.0000 & 0.8095 \\ \hline
35 & 0.7704 & 0.6314 & 0.4443 & 0.1905 & 0.0004 & 0.0000 & 0.8091 \\ \hline
36 & 0.7680 & 0.6360 & 0.4405 & 0.1905 & 0.0000 & 0.0000 & 0.8095 \\ \hline
37 & 0.7655 & 0.6406 & 0.4370 & 0.1905 & 0.0001 & 0.0000 & 0.8094 \\ \hline
38 & 0.7630 & 0.6451 & 0.4333 & 0.1905 & 0.0000 & 0.0000 & 0.8095 \\ \hline
39 & 0.7604 & 0.6497 & 0.4290 & 0.1905 & 0.0001 & 0.0000 & 0.8095 \\ \hline
40 & 0.7578 & 0.6540 & 0.4260 & 0.1905 & 0.8055 & 0.0000 & 0.0040 \\ \hline
41 & 0.7552 & 0.6582 & 0.4229 & 0.1905 & 0.8095 & 0.0000 & 0.0001 \\ \hline
42 & 0.7525 & 0.6626 & 0.4189 & 0.1905 & 0.0000 & 0.0000 & 0.8095 \\ \hline
43 & 0.7498 & 0.6667 & 0.4158 & 0.1905 & 0.0001 & 0.0000 & 0.8094 \\ \hline
44 & 0.7471 & 0.6708 & 0.4118 & 0.1905 & 0.0000 & 0.0000 & 0.8095 \\ \hline
45 & 0.7444 & 0.6746 & 0.4094 & 0.1905 & 0.0000 & 0.0000 & 0.8095 \\ \hline
\end{tabular} \begin{tabular}{|c|c|c|c|c|c|c|c|}
\hline $\theta$ & $\mathbf{LB}'_{17/21}(\theta)$ & $\gamma$ & $\beta$ & $r_1+r_6$ & $r_2+r_5$ & $r_3$ & $r_4$ \\ \hline
46 & 0.7417 & 0.6784 & 0.4066 & 0.1905 & 0.0000 & 0.0000 & 0.8095 \\ \hline
47 & 0.7390 & 0.6820 & 0.4047 & 0.1905 & 0.0000 & 0.0000 & 0.8095 \\ \hline
48 & 0.7364 & 0.6835 & 0.4016 & 0.1905 & 0.0000 & 0.0000 & 0.8095 \\ \hline
49 & 0.7337 & 0.6757 & 0.4005 & 0.1905 & 0.0000 & 0.0000 & 0.8095 \\ \hline
50 & 0.7321 & 3.1415 & 1.5708 & 0.1905 & 0.4048 & 0.4048 & 0.0000 \\ \hline
51 & 0.7457 & 3.1419 & 1.5707 & 0.1905 & 0.8092 & 0.0002 & 0.0001 \\ \hline
52 & 0.7583 & 3.1417 & 1.5708 & 0.1905 & 0.4048 & 0.4047 & 0.0000 \\ \hline
53 & 0.7698 & 3.1414 & 1.5709 & 0.1905 & 0.4048 & 0.2024 & 0.2024 \\ \hline
54 & 0.7800 & 3.1415 & 1.5708 & 0.1905 & 0.4048 & 0.0000 & 0.4047 \\ \hline
55 & 0.7888 & 3.1417 & 1.5707 & 0.1905 & 0.4048 & 0.4047 & 0.0000 \\ \hline
56 & 0.7961 & 3.1419 & 1.5707 & 0.1905 & 0.0009 & 0.8086 & 0.0001 \\ \hline
57 & 0.8019 & 3.1420 & 1.5706 & 0.1905 & 0.8095 & 0.0000 & 0.0000 \\ \hline
58 & 0.8061 & 3.1420 & 1.5707 & 0.1905 & 0.2701 & 0.2697 & 0.2697 \\ \hline
59 & 0.8087 & 3.1415 & 1.5708 & 0.1905 & 0.4047 & 0.4047 & 0.0001 \\ \hline
60 & 0.8095 & 3.1417 & 1.5708 & 0.1905 & 0.4048 & 0.4044 & 0.0003 \\ \hline
61 & 0.8087 & 3.1414 & 1.5709 & 0.1905 & 0.0000 & 0.4047 & 0.4047 \\ \hline
62 & 0.8061 & 3.1415 & 1.5706 & 0.1905 & 0.5397 & 0.2698 & 0.0000 \\ \hline
63 & 0.8019 & 3.1419 & 1.5708 & 0.1905 & 0.8095 & 0.0000 & 0.0000 \\ \hline
64 & 0.7961 & 3.1419 & 1.5707 & 0.1905 & 0.8095 & 0.0001 & 0.0000 \\ \hline
65 & 0.7888 & 3.1424 & 1.5703 & 0.1905 & 0.4054 & 0.0003 & 0.4038 \\ \hline
66 & 0.7800 & 3.1419 & 1.5703 & 0.1905 & 0.8090 & 0.0003 & 0.0003 \\ \hline
67 & 0.7698 & 3.1412 & 1.5709 & 0.1905 & 0.8087 & 0.0007 & 0.0001 \\ \hline
68 & 0.7583 & 3.1420 & 1.5707 & 0.1905 & 0.8086 & 0.0009 & 0.0000 \\ \hline
69 & 0.7460 & 3.0188 & 1.6616 & 0.1905 & 0.0000 & 0.0000 & 0.8095 \\ \hline
70 & 0.7352 & 2.9146 & 1.7351 & 0.1905 & 0.0000 & 0.0000 & 0.8095 \\ \hline
71 & 0.7263 & 2.8472 & 1.7796 & 0.1905 & 0.0000 & 0.0000 & 0.8095 \\ \hline
72 & 0.7189 & 2.7951 & 1.8124 & 0.1905 & 0.0000 & 0.0000 & 0.8095 \\ \hline
73 & 0.7129 & 2.7523 & 1.8382 & 0.1905 & 0.0000 & 0.0000 & 0.8095 \\ \hline
74 & 0.7080 & 2.7168 & 1.8586 & 0.1905 & 0.0000 & 0.0000 & 0.8095 \\ \hline
75 & 0.7041 & 2.6857 & 1.8760 & 0.1905 & 0.0000 & 0.0000 & 0.8095 \\ \hline
76 & 0.7010 & 2.6594 & 1.8907 & 0.1905 & 0.0000 & 0.0000 & 0.8095 \\ \hline
77 & 0.6986 & 2.6367 & 1.9031 & 0.1905 & 0.0000 & 0.0000 & 0.8095 \\ \hline
78 & 0.6967 & 2.6174 & 1.9136 & 0.1905 & 0.0000 & 0.0000 & 0.8095 \\ \hline
79 & 0.6953 & 2.5999 & 1.9230 & 0.1905 & 0.0000 & 0.0000 & 0.8095 \\ \hline
80 & 0.6943 & 2.5854 & 1.9309 & 0.1905 & 0.0000 & 0.0000 & 0.8095 \\ \hline
81 & 0.6935 & 2.5729 & 1.9375 & 0.1905 & 0.0000 & 0.0000 & 0.8095 \\ \hline
82 & 0.6930 & 2.5623 & 1.9432 & 0.1905 & 0.0000 & 0.0000 & 0.8095 \\ \hline
83 & 0.6928 & 0.6154 & 0.3920 & 0.1905 & 0.0000 & 0.0000 & 0.8095 \\ \hline
84 & 0.6927 & 0.6155 & 0.3919 & 0.1905 & 0.0000 & 0.0000 & 0.8095 \\ \hline
85 & 0.6926 & 0.6154 & 0.3922 & 0.1905 & 0.0000 & 0.0000 & 0.8095 \\ \hline
86 & 0.6926 & 0.6153 & 0.3927 & 0.1905 & 0.0000 & 0.0000 & 0.8095 \\ \hline
87 & 0.6925 & 0.6154 & 0.3926 & 0.1905 & 0.0000 & 0.0000 & 0.8095 \\ \hline
88 & 0.6925 & 0.6155 & 0.3928 & 0.1905 & 0.0000 & 0.0000 & 0.8095 \\ \hline
89 & 0.6925 & 0.6155 & 0.3928 & 0.1905 & 0.0000 & 0.0000 & 0.8095 \\ \hline
90 & 0.6924 & 0.6156 & 1.9636 & 0.1450 & 0.4275 & 0.2138 & 0.2138 \\ \hline
\end{tabular}

    }
    \caption{\footnotesize Table of values of $\mathbf{LB}_{17/21}'(\theta)$ for various $\theta$ including the optimal values of the minimizers and maximizers found in the inner and outer optimizations respectively.}
    \label{tab:LB_17/21}
\end{table}

\end{document}